\documentclass[a4paper,12pt]{article}

\usepackage[a4paper,margin=1in]{geometry}   
\usepackage{setspace}
\usepackage{graphicx}
\usepackage{amsmath,amsthm,amsfonts,mathtools}
\usepackage{amssymb}
\usepackage{subcaption}
\usepackage{algorithm}
\usepackage{algpseudocode}
\usepackage{array,booktabs,makecell,multicol}
\usepackage{physics}
\usepackage{graphicx} 
\usepackage{titlesec}
\titleformat{\section}{\bfseries\Large}{\thesection}{1em}{}          
\titlespacing*{\section}{0pt}{*2}{*1}                                
\titleformat{\subsection}{\bfseries\large}{\thesubsection}{0.8em}{}  
\titlespacing*{\subsection}{0pt}{*1.5}{*0.8}

\usepackage[colorlinks=true, citecolor=blue, linkcolor=red, urlcolor=blue]{hyperref}
\usepackage{natbib}                  
\newtheorem{assumption}{Assumption}
\newtheorem{theorem}{Theorem}[section]
\newtheorem{corollary}{Corollary}[theorem]
\newtheorem{lemma}[theorem]{Lemma}
\newtheorem{definition}{Definition}
\newtheorem{proposition}{Proposition}

\newcommand{\Eg}[1]{\mathbb{E}\!\left[\tilde g^{#1}\mid X=x\right]}

\begin{document}

\begin{center}
\LARGE
\textbf{Estimation of heterogeneous principal effects under principal ignorability} \\[14pt]

\large
Rui Zhang$^{1}$, Charles R.\ Doss$^{1}$, Jared D.\ Huling$^{2}$ \\[10pt]

\normalsize

$^{1}$School of Statistics, University of Minnesota, USA \\
$^{2}$Division of Biostatistics and Health Data Science, \\School of Public Health, University of Minnesota, USA\\[6pt]

\end{center}

\begin{abstract}

We study estimation and inference for heterogeneous principal causal effects with binary treatments and binary intermediate variables. Principal causal effects are subgroup effects within strata defined by potential values of an intermediate variable, including effects among compliers. We propose a framework for estimating and forming pointwise confidence intervals for heterogeneous principal causal effects under the principal ignorability assumption.
Several estimators are developed, and their robustness properties are characterized: one estimator is doubly robust, whereas the other two attain intermediate robustness between double and triple robustness; in contrast, principal causal effects can be estimated in a triply robust manner only. We establish large-sample theory under nonparametric smoothness conditions and analyze the bias contributions of each approach, providing insight into performance beyond the smooth setting, including in high-dimensional regimes. Camden Coalition hotspotting randomized trial are used to illustrate the methods by estimating heterogeneous complier effects.

\textbf{Keywords:} Complier effects; Double robustness; Effect heterogeneity; Influence functions; Nonparametric smoothing; Principal stratification

\end{abstract}

\newpage

\section{Introduction}\label{s:1}
In this paper, we develop a new framework for estimating and performing inference on heterogeneous principal causal effects under the principal ignorability assumption. Our framework allows the estimation of conditional principal causal effects among never takers, compliers, and always takers. Within our framework, we develop four estimators, three with robustness properties with respect to nuisance parameter estimation and one without. The three with robustness properties can be implemented with flexible machine learning methods. We apply our framework to the Camden Coalition’s hotspotting program, a care management intervention for high-utilizing patients, to assess whether its effects differ across patient characteristics among high engagers with the program.

The Camden Coalition’s hotspotting program \citep{finkelstein2020health} is a care management intervention designed to support high cost, medically and socially complex patients, with the goal of improving outcomes and reducing readmissions. To rigorously evaluate its effectiveness, the Coalition conducted a randomized controlled trial (RCT) in which recently hospitalized superutilizers were assigned to receive either the intervention or usual care, with hospital readmission as the primary outcome. \citet{finkelstein2020health} found no significant overall effect of the program on readmissions, raising questions about whether the intervention may nonetheless benefit certain subgroups of patients. \citet{yang2023hospital} conducted a secondary analysis which found that average null effects masked important subgroup differences: many individuals assigned to treatment did not fully engage, but among those with a high baseline probability of engagement, the program yielded significant reductions in readmissions.

In light of these findings, decision makers confront a fundamental question: what causes these subgroup differences? Two explanations are possible. The differences may be attributable to high engagement, and in this case, within the subgroup of highly engaged individuals, the treatment effect could be homogeneous, implying that the heterogeneous principal causal effect function (defined later) is constant within this subgroup. Alternatively, the difference may arise from how the treatment itself operates across individuals. Under this explanation, the heterogeneous principal causal effect function is non-constant and varies with baseline covariates, reflecting genuine effect heterogeneity within the principal stratum.

Distinguishing between these explanations is critical, as it determines whether observed benefits among high engagers reflect systematic differences in who participates in a uniformly effective intervention, or genuine heterogeneity in the causal effects of the intervention. This distinction, in turn, guides whether decision makers should prioritize improving outreach and engagement, refining the intervention itself, or both.
This challenge is not unique to hotspotting. In job training programs, for instance, positive outcomes may arise because more motivated individuals are the ones who engage, or because the program itself is genuinely more effective for those individuals. Similarly, in public health programs such as smoking cessation interventions, better results among compliers may reflect baseline differences in addiction severity, or true heterogeneity in how the intervention works across subgroups. Careful estimation of conditional principal causal effects is therefore essential to disentangle these mechanisms and provide a more accurate understanding of treatment effectiveness.

There is limited work on heterogeneous principal causal effects. Most existing work focuses on complier effects under an instrumental variable (IV) framework \citep{takatsu2025doubly, johnson2022detecting, bargagli2022heterogeneous, lee2023minimally}. For example, \citet{bargagli2022heterogeneous} derived a nonparametric identification formula for heterogeneous complier effects that parallels the classical LATE identification results of \citet{angrist1995identification} and \citet{angrist1996identification}.
However, the exclusion restriction (ER) can be difficult to justify in some applications. In the hotspotting trial, for example, the intervention was not double-blinded and patients knew their assignment, making ER particularly implausible. Likelihood-based alternatives (e.g., \citealp{chen2024bayesian}) can relax the ER, but are often sensitive to modeling assumptions and difficult to implement \citep{ho2022weak}. This motivates alternative identification strategies, such as those based on principal ignorability (PI) \citep{jo2009use}, which can provide more credible causal inference when ER assumptions fail and likelihood-based approaches are subject to misspecification.

The PI assumption requires that baseline covariates capture the dependence between stratum membership and potential outcomes \citep{jo2009use}. 
\citet{ding2017principal} proposed methods for estimating PCEs in randomized experiments under PI assumption, including a weighting-based identification formula and covariate-adjusted estimators. \citet{jiang2022multiply} developed a general framework for both randomized experiments and observational studies and developed efficient influence function (EIF)-based estimators that offer triple robustness and efficiency guarantees.

We develop a new framework for identifying, estimating, and conducting inference on conditional principal causal effects (CPCEs) for never-takers, compliers, and always-takers under a PI assumption. Our identification strategy is motivated by the conditional average treatment effect (CATE) literature \citep{kunzel2019metalearners, nie2021quasi, semenova2021debiased, kennedy2023towards}. As a starting point, we adopt a T-learner-style approach \citep{kunzel2019metalearners}, which models subgroup outcome means separately and takes their difference. While intuitive, this approach inherits the well-known limitations of the T-learner, including sensitivity to model misspecification, limited robustness, and tying the functional form of the estimated CPCE to the outcome model choice, limiting interpretability. 

To address these shortcomings, we develop three estimators that are robust to nuisance function estimation, the subset, EIF, and one-step estimators, and that are compatible with cross-fitting and flexible machine learning methods \citep{chernozhukov2018double}. For the subset estimator, we extend the DR-learner of \citet{kennedy2023towards} to a specific observed subset, which is tailored to principal stratification. The EIF estimator leverages the efficient influence function (EIF) for principal causal effects \citep{jiang2022multiply} and uses the full dataset, but we find it can be numerically unstable in small samples. To mitigate this, we propose a one-step estimator that refines a preliminary estimator (e.g., a T-learner) via an influence-function based correction.

Furthermore, we study pointwise error bounds for these three estimators under smoothness conditions \citep{kennedy2023towards}. The results show that the subset estimator is doubly robust, remaining consistent if either of two nuisance functions is correctly specified. Although \citet{jiang2022multiply} establish only triple robustness for marginal PCEs, we find stronger robustness for CPCEs: the EIF estimator and the one-step estimator (with the T-learner as the preliminary estimator) remain consistent if either the score components (i.e., the two score nuisance functions jointly) are correctly specified or the outcome regression is correctly specified. Together, these innovations provide a general and robust framework for learning heterogeneous principal causal effects, offering decision makers finer grained insights into which subpopulations stand to benefit most from interventions such as hotspotting.

The remainder of the paper is organized as follows. Section \ref{s:2} introduces notation, target estimands, and assumptions, and presents the T-learner estimator for CPCEs. Section \ref{s:3} develops multiple robust identification strategies for CPCEs and their oracle robustness properties, including the subset, EIF, and one-step estimators. Section \ref{s:4} establishes error bounds for the subset, EIF, and one-step estimators under stable second-stage regression. Section \ref{s:5} reports results from simulation experiments characterizing the performance of our proposed methods. Section \ref{s:6} applies the framework to data from the Healthcare Hotspotting Trial.

\section{Setup}\label{s:2}
\subsection{Estimand, notation and assumption}
Assume that $W=(X,Y,S,Z)$ and $\{W_i = (X_i, Y_i, S_i, Z_i) ,$ $\quad i = 1,...,n\}$ are independent and identically distributed observations, where $X_i \in \mathbb{R}^p$ denotes pretreatment covariates, $Y_i\in \mathbb{R}$ is the outcome of interest, $Z_i\in\{0,1\}$ is the binary treatment assignment, and $S_i\in\{0,1\}$ is the binary intermediate variable. For example, $S$ could represent an indicator of survival when the estimand of interest is the survivor average causal effect, or it could represent treatment usage (not just assignment) when the focus is on causal effects among compliers. To lighten notation, we omit the subscript $i$ denoting individual units. We adopt the potential outcomes framework. Under the treatment assignment $z$ for $z \in \{0,1\}$, we can define the potential values $Y(z)$ and $S(z)$. 

\citet{frangakis2002principal} proposed the notion of principal stratification by the joint potential values of the intermediate latent variable, $U=\left(S(1), S(0)\right)$. For simplicity, we denote the possible values of $U$,  $\{(0,0),(1,0),(0,1),(1,1)\}$, as ${00, 10, 01, 11}$. 
It is appropriate to condition on the principal stratum $U = (S(1), S(0))$ when defining (principal) causal effects, as $U$ is a pre-treatment characteristic with respect to potential outcomes and is unaffected by the actual treatment assignment. In contrast, conditioning on the observed post-treatment variable $S$ can induce bias, since $S$ may be affected by treatment and thus should not be used (naively) to define a well-posed causal contrast.
The PCEs are the average causal effect within a principal stratum $U$, and are defined as 
\begin{equation*}
    \tau^{u}=\mathbb{E}\left(Y(1)-Y(0) \mid U=u\right), \quad u=00,10,11,01.
\end{equation*}
Note that we use superscripts for principal strata in this article. We aim to estimate the CPCE function which characterizes the heterogeneity of causal effects within latent principal strata. For a principal stratum and covariates, the CPCEs are defined as
\begin{equation*}
    \tau^{u}(X)=\mathbb{E}\left(Y(1)-Y(0) \mid U=u,X\right), \quad u=00,10,11,01.
\end{equation*}
Because $U$ cannot be observed directly (since the potential values $S(1)$ and $S(0)$ cannot be observed simultaneously), to estimate $\tau^u (X)$, identifying assumptions are needed. \citet{jiang2022multiply} provided an in-depth discussion on the identification of PCEs, outlining a set of key assumptions necessary for identification. These are the same assumptions needed in this work.

\begin{assumption}[Consistency]\label{assump:con}
Assume $Y = Y(1)Z + Y(0)(1-Z)$ and $S = S(1)Z + S(0)(1-Z)$.
\end{assumption}
\noindent This assumption is necessary for identification.  Unlike the standard causal inference setup, where consistency is only required of $Y$, here it is required for both $Y$ and $S$.
\begin{assumption}[Treatment Ignorability]
\label{assump:pi}
Assume $(Y(1), Y(0), S(1), S(0)) \perp\!\!\!\perp  Z \mid X$, where $\perp\!\!\!\perp$ denotes independence.
\end{assumption}
\noindent  Ignorability of treatment assignment requires that, conditional on observed covariates $X$, treatment $Z$ is independent of all potential outcomes and intermediate variables, which eliminates confounding between the treatment and both the intermediate variable $S$ and the outcome $Y$. In the context of an RCT, this ignorability assumption holds by design, since treatment is independent of all unobserved patient characteristics. It extends the classical treatment ignorability assumption commonly invoked in observational studies to the setting of principal stratification. In observational studies, this assumption requires a rich set of observed covariates that capture all confounders of treatment, outcome, and the intermediate variable.

\begin{assumption}[Monotonicity]
\label{assump:monotonicity}
Assume $S(1) \geq S(0)$.
\end{assumption}
\noindent Monotonicity stipulates that the treatment has a non-negative effect on the intermediate variable for all units, which explicitly rules out the presence of so-called defiers $(U = 01)$. We develop our identification and estimation results under monotonicity rather than strong monotonicity,
which allows for the possibility of always-takers ($S(0)=1$). If strong monotonicity is known to hold (i.e., $S(0)=0$), such as in one-sided noncompliance settings \citep{sommer1991estimating},
then this assumption holds automatically and the always-taker stratum is absent. Then, the identification formulas for compliers and never-takers in Section~\ref{s:3} remain unchanged.

To identify the CPCE, we need an additional assumption.
\citet{angrist1995identification} and \citet{angrist1996identification} proposed the PCEs among compliers ($\tau^{10}$) under monotonicity and an ER assumption, which sets the PCEs for always-takers and never-takers to zero ($\tau^{11} = \tau^{00} =0$). 
These works inspired \cite{bargagli2022heterogeneous} to propose a formula for the CPCE among compliers ($\tau^{10}(x)$) under the ER assumption.

However, the ER assumption is often implausible and furthermore cannot be imposed when the CPCE for strata $U = 00$ and $U= 11$ are the primary targets of inference. \citet{jiang2022multiply} proposed a framework to identify PCEs under an alternative assumption, principal ignorability.

\begin{assumption}[Principal ignorability]\label{assump:ps}
Assume that $\mathbb{E}\left(Y(1)\mid U=11, X\right)$ is equal to $\mathbb{E}\left(Y(1) \mid U=10, X\right)$, and $\mathbb{E}(Y(0) \mid U=00, X)=E\left(Y(0) \mid U=10, X\right).$
\end{assumption}
\noindent Principal ignorability requires that, conditional on covariates $X$, the mean potential outcomes do not differ across the relevant principal strata. It requires that the covariates $X$ are sufficiently rich to eliminate all differences in potential outcomes between principal strata. \citet{jiang2022multiply} also gave an equivalent version of Assumption \ref{assump:ps} under Assumptions \ref{assump:pi}-\ref{assump:monotonicity}, which is 
$$
\begin{aligned}
& \mathbb{E}\left(Y(1) \mid U=11, Z=1, S=1, X\right)=\mathbb{E}\left(Y(1) \mid U=10, Z=1, S=1, X\right), \\
& \mathbb{E}\left(Y(0) \mid U=00, Z=0, S=0, X\right)=\mathbb{E}\left(Y(0) \mid U=10, Z=0, S=0, X\right).
\end{aligned}
$$
Intuitively, principal ignorability transforms a latent mixture problem, defined over unobserved strata, into an observed mixture problem, which enables identification of principal causal effects from the data. 

\subsection{Contrast 
identification for CPCEs}
We define the mean of the outcome within the observed strata $\{ Z=z, S=s \}$ as $\mu_{z s}(X)=\mathbb{E}(Y \mid Z=z, S=s, X)$. To distinguish between observed strata and principal strata, we use subscripts to denote observed strata and superscripts to denote principal strata. The following theorem provides an identification formula for each CPCE in terms of contrasts of the outcome mean functions.

\begin{theorem}\label{the1}
Let Assumptions~\ref{assump:con}--\ref{assump:ps} hold. Then the CPCEs are identified as
\[
\tau^{11}(x)=\mu_{11}(x)-\mu_{01}(x), \quad
\tau^{10}(x)=\mu_{11}(x)-\mu_{00}(x), \quad
\tau^{00}(x)=\mu_{10}(x)-\mu_{00}(x),
\quad \text{for all } x.
\]
\end{theorem}

\noindent Analogously to the CATE, which can be represented as the difference between the conditional mean outcomes of the treated and control groups under the ignorability assumption, Theorem \ref{the1} shows that under Assumptions \ref{assump:con}–\ref{assump:ps}, the CPCE can be written as a difference in conditional mean outcomes between observed groups. In contrast, the PCE generally cannot be expressed as a simple difference in means between the observed groups.

From the PCE identification formula in \citet{jiang2022multiply} and Theorem \ref{the1}, a clear connection between CPCEs and PCEs can be observed. First, let us introduce some notation.
Let $\pi(x):=\mathbb{P}(Z=1\mid X=x)$ denote the propensity score. For $u\in\{00,10,11\}$, let $e^{u}(x):=\mathbb{P}(U=u\mid X=x)$
denote the principal score, i.e., the conditional probability of belonging to principal stratum $u$ given $X=x$. Define $e^{u}:=\mathbb{E}\!\left[e^{u}(X)\right]$
as the marginal probability of principal stratum $u$ under the distribution of $X$. \cite{jiang2022multiply} showed that the principal scores are identified under Assumptions \ref{assump:pi} and \ref{assump:monotonicity} by
\begin{equation}
    e^{10}(x)=p_1(x)-p_0(x), \quad e^{00}(x)=1-p_1(x), \quad e^{11}(x)=p_0(x),
\end{equation}
where 
 $p_z(x) = P(S = 1|Z = z, X =x)$. 
 Building on this result, we state the corresponding PCE identification formulas below to clarify the relationship between PCEs and CPCEs:
\begin{equation*}
\tau^{u} = \mathbb{E}\left[\frac{e^{u}(X)}{e^{u}} \tau^{u}(X)\right].
\end{equation*}
These equations show that PCEs are weighted averages of CPCEs, where the weights reflect the distribution of covariates $X$ within each principal stratum. This differs from the relationship between the CATE and the ATE, where the ATE is simply the direct (unweighted) average of the CATE over the covariate distribution. The weights appear because we are averaging over the distribution of $X$ within each principal stratum which are unobserved, not over the overall population. 

In the following section, we present several novel estimators for CPCEs, drawing inspiration from recent advances in CATE estimation while addressing the unique challenges posed by principal stratification.

\subsection{T-learner estimators for CPCEs}

\citet{kunzel2019metalearners} refer to the contrast-based approach for estimating the CATE as the T-learner. In this framework, the CATE is estimated as the difference between two separately estimated outcome models, one for the treated group and one for the control group. Inspired by this framework, we can similarly estimate the CPCE by taking the difference between estimates of conditional mean outcomes across the relevant observed groups. To highlight the contrast-based structure of the estimator, we refer to our approach as a T-learner for CPCE, analogous to the T-learner framework for CATE. Based on Theorem~\ref{the1}, we define the T-learner estimators as
\begin{equation}
    \hat{\tau}^{00}(x):= \hat{\mu}_{10}(x)-\hat{\mu}_{00}(x),\quad
        \hat{\tau}^{11}(x):= \hat{\mu}_{11}(x)-\hat{\mu}_{01}(x),\quad
        \hat{\tau}^{10}(x):= \hat{\mu}_{11}(x)-\hat{\mu}_{00}(x),
\end{equation}
where $\hat{\mu}_{zs}(x)$ denotes an estimator of $\mu_{zs}(x)$.
The T-learner is simple and flexible, accommodating any machine learning method to model outcomes under each subgroups. However, it is sensitive to model misspecification: since the CPCE is estimated as the difference of two regressions, errors in either model directly bias the result. Moreover, because it solves two separate problems rather than targeting the CPCE directly, it can be especially biased under group imbalance, where the smaller group tends to be over-smoothed and the larger group under-smoothed.  

To illustrate this limitation, we adapt a toy example from \citet{kennedy2023towards}. The design balances the marginal proportions of the four observed subgroups but induces strong treatment imbalance conditional on $X$. In such settings, the T-learner over-smooths sparse regions and under-smooths dense ones, producing an unnecessarily complex estimate even though the true effect is constant at zero. The details and results of this example can be found in the Supplementary \ref{s:toy_example}. The next section introduces alternative estimators that overcome these issues.  

\section{Multiple robust identification for CPCEs}\label{s:3}
\subsection{Subset identification for CPCEs}
In this section, we present methods for identifying the CPCEs from a subset perspective: by applying the DR-learner to specific observable subsets of the data, we can attain a the doubly robust property despite the presence of three nuisance functions in the setup.

Revisiting Theorem~\ref{the1}, we observe an interesting structural pattern in the identification of the CPCEs. For example, $\tau^{00}(x) = \mu_{10}(x) - \mu_{00}(x)$ compares outcomes within $S=0$, with $\mu_{10}(x)$ from treated units ($Z=1$) and $\mu_{00}(x)$ from controls ($Z=0$). Similarly, $\tau^{10}(x)$ compares outcomes when $S=Z$, and $\tau^{11}(x)$ compares outcomes within $S=1$.
These patterns highlight that each CPCE contrasts outcomes across treatment groups within a shared observed subset, which enables estimation analogous to that of the CATE but within specific subpopulations. The following theorem offers a formal justification for the intuition outlined above and provides an identification result that motivates a DR-learner style estimator applied to certain subsets of the observed data. { To do so, we first define the following `pseudo-outcomes',
\begin{equation}
\label{eq:subset:pseudos}
    \begin{split}
        \varphi_{\tau^{00}}(W)
&=\frac{Z-\pi_{\mathcal{S}_{00}}(X)}
{\pi_{\mathcal{S}_{00}}(X)\{1-\pi_{\mathcal{S}_{00}}(X)\}}
\Bigl\{Y-\mu_{Z0}(X)\Bigr\}
+\mu_{10}(X)-\mu_{00}(X),\\[4pt]
\varphi_{\tau^{10}}(W)
&=\frac{Z-\pi_{\mathcal{S}_{10}}(X)}
{\pi_{\mathcal{S}_{10}}(X)\{1-\pi_{\mathcal{S}_{10}}(X)\}}
\Bigl\{Y-\mu_{ZZ}(X)\Bigr\}
+\mu_{11}(X)-\mu_{00}(X), \text{ and}\\[4pt]
\varphi_{\tau^{11}}(W)
&=\frac{Z-\pi_{\mathcal{S}_{11}}(X)}
{\pi_{\mathcal{S}_{11}}(X)\{1-\pi_{\mathcal{S}_{11}}(X)\}}
\Bigl\{Y-\mu_{Z1}(X)\Bigr\}
+\mu_{11}(X)-\mu_{01}(X),
    \end{split}
\end{equation}
where 
$\pi_{\mathcal{S}_{u}}(x)=\mathbb{E}(Z\mid \mathcal{S}_{u}, X =x )$, and 
 where $\mathcal{S}_{u}$ denotes the conditioning set we will use for subset identification of
$\tau^{u}(X)$, namely
$\mathcal{S}_{00}=\{S=0\},$
$\mathcal{S}_{10}=\{Z=S\},$ 
$ \mathcal{S}_{11}=\{S=1\}.$ }{Recall that $\mu_{zs}(x) = \mathbb{E}(Y| Z=z,S=s, X=x)$; note that the capital Z subscripts in \eqref{eq:subset:pseudos}, e.g.\ $\mu_{Z0}$, denote randomness in which function (e.g.\ $\mu_{10}$ or $\mu_{00}$) is selected.}
\begin{theorem}\label{the:subset}
Under Assumptions \ref{assump:con}--\ref{assump:ps}, the CPCEs can be identified by 
\begin{equation}
\tau^{u}(x)=\mathbb{E}\!\left(\varphi_{\tau^{u}}(W)\mid \mathcal{S}_{u},\,X=x\right),\quad  \text{for all }  x.  
\end{equation}
\end{theorem} 
\noindent Theorem~\ref{the:subset} provides a path, which we will follow later, to estimating the CPCEs by applying a DR-Learner to a defined subset of the data. However, it is important to note that the nuisance functions $\pi_{\mathcal{S}_u}(\cdot)$ in this setting are not the conventional propensity scores. Instead, they correspond to conditional treatment probabilities within an observed stratum (rather than a principal stratum). We refer to the functions $\pi_{\mathcal{S}_u}(\cdot)$ as ``subset propensity scores;" 
they have structured relationships with both propensity scores and principal scores.
For example, $\pi_{\mathcal S_{00}}(x) = \frac{(1 - p_1(x)) \cdot  \pi(x)}{ \pi(x)(1 - p_1(x)) + (1 - p_0(x))(1 -  \pi(x))}$ (see Supplementary \ref{s:subpropensity} for derivation and additional details). Thus, the analyst may either directly estimate the subset propensity score $\pi_{\mathcal S_{u}}(x)$, or alternatively, estimate both the overall propensity score $\pi(x)$ and the principal scores $p_1(x)$ and $p_0(x)$. In practice, we recommend estimating the subset propensity score $\pi_{\mathcal S_{u}}(x)$ directly when the propensity score or principal score is unavailable, since this requires fewer nuisance components. If both of these scores are known, however, it is preferable to use the latter form instead.

It is important to emphasize that the interpretation of subset identification for CPCEs is fundamentally different from that of identifying a CATE within an observed subgroup. Nonetheless, the corresponding subset plug-in limit parameter has analogous double robustness properties. For any stratum $u\in\{00,10,11\}$ and any collection of generic nuisance functions
$\tilde\eta_{sub}:=\{\tilde{\pi}_{\mathcal S_u}(\cdot), \tilde{\mu}_{zs}(\cdot)\}$, define the
plug-in pseudo-outcome $\tilde{\varphi}_{\tau^u}(W)$ by taking the functional form of
$\varphi_{\tau^u}(W)$ given in \eqref{eq:subset:pseudos} and replacing each population nuisance function by its
generic counterpart in $\tilde\eta_{sub}$. For example, when $u=00$, $\tilde{\varphi}_{\tau^{00}}(W)
= \frac{Z-\tilde{\pi}_{\mathcal S_{00}}(X)}{\tilde{\pi}_{\mathcal S_{00}}(X)\{1-\tilde{\pi}_{\mathcal S_{00}}(X)\}}
\Bigl(Y-\tilde{\mu}_{Z0}(X)\Bigr)
+ \tilde{\mu}_{10}(X)-\tilde{\mu}_{00}(X),$
with analogous definitions for $u\in\{10,11\}$.
We then define the corresponding (subset) plug-in limit parameter as
\begin{equation*}
\tilde{\tau}^{u}_{\mathrm{sub}}(x)\;:=\; \mathbb{E}\!\left(\tilde{\varphi}_{\tau^{u}}(W)\,\middle|\,\mathcal{S}_{u},\; X=x\right), \quad \text{for all } x.
\end{equation*}
  
\noindent Here, the nuisance functions $\tilde{\pi}_{\mathcal{S}_{u}}(\cdot)$ and $\tilde{\mu}_{zs}(\cdot)$ are simply fixed functions that are not assumed equal to the true population analogs. Conceptually, we think of these as plug-in estimates obtained from a separate training sample; in this context, our next theorem will provide the bias of the plug-in limit parameter. 
\begin{theorem}\label{robustness_subset} Under Assumption \ref{assump:con}--\ref{assump:ps}, the subset plug-in limit parameter has the following robustness properties:
\begin{itemize}
\item \textbf{(Weak Double Robustness)}  
Suppose, for all $x$, either $\tilde{\pi}_{\mathcal S_{u}} (x)=\pi_{\mathcal S_{u}}(x)$ or
$\tilde{\mu}_{zs}(x)=\mu_{zs}(x)$, then $\tilde{\tau}^{u}_{\mathrm{sub}}(x)=\tau^{u}(x)$.

\item \textbf{(Rate Double Robustness)}  
Suppose, for all $x$, subset overlap holds, i.e.\ there exists $\epsilon>0$ such that
    $\epsilon\le \tilde\pi_{\mathcal S_{u}}(x)\le 1-\epsilon$ almost surely. 
Then, for each $u\in\{00,10,11\}$, $\tilde\tau^{u}_{\mathrm{sub}}(x)-\tau^{u}(x)$ equals
\begin{equation*}\label{eq:subset:bias}
O_p\!\Bigl(
  |\tilde\pi_{\mathcal S_{u}}(x)-\pi_{\mathcal S_{u}}(x)|\cdot  \max_{s,z \in \{0,1\}}|\tilde\mu_{zs}(x)-\mu_{zs}(x)|
\Bigr).
\end{equation*}

\end{itemize}
\end{theorem}

\noindent Theorem~\ref{robustness_subset} shows that the plug-in limit parameter are (weakly) doubly robust and then quantifies the order of magnitude of the bias. 
Later, in Section~\ref{s:4}, we develop estimators that regress plug-in pseudo-outcomes on covariates; their robustness properties are linked to Theorem~\ref{robustness_subset}.

The bias is rate doubly robust, converging as does the product of the subset propensity score and outcome regression errors, and improving whenever both are well estimated. The detailed forms and proofs are provided in Supplementary \ref{s:proofrobsub}.
The downside to using a subset-based method is that indeed, it relies only on a subset of the available data, rather than all of it.  

\subsection{EIF identification}
The DR-learner  estimator for the CATE is motivated by the EIF for the ATE \citep{kennedy2023towards}. Thus our pseudo-outcomes used for subset identification from \eqref{eq:subset:pseudos} 
are also influence function-motivated, 
but they are constructed and regressed within a specific subset of the data.

This naturally raises the question of whether one can instead employ the EIF over the entire dataset, thereby making {use of the full data sample rather than using just a subset}.

\citet{jiang2022multiply} introduced the EIF for PCEs using a set of key quantities, and showed that setting to zero and solving an estimate of the centered EIF yields estimators with triple robustness and local efficiency. Before presenting our identification strategy for CPCEs, we first review these quantities from \citet{jiang2022multiply}, as they play a central role in the identification and estimation of both PCEs and CPCEs.
For a function $f \equiv f(Y,S,X)$ and for any $a \in \{0,1\}$,  define $\psi_{a,f(Y,S,X)}$ by
\begin{align*}
\psi_{a, f(Y,S,X)}(w)
&:= \frac{\mathbf{1}(Z=a)}{\mathbb{P}(Z=a\mid X=x)}
    \Bigl\{ f(y,s,x) - \mathbb{E}\!\bigl[f(Y,S,X)\mid X=x,Z=z\bigr]\Bigr\} \\
&\hspace{12em} + \mathbb{E}\!\bigl[f(Y,S,X)\mid X=x,Z=z\bigr].
\end{align*}
\begin{definition}[EIF quantities]
For $w=(x,y,s,z)$, define
\begin{align*}
\phi_{1,10}(w)
&:= \frac{e^{10}(x)}{p_1(x)}\,\psi_{1,YS}(w)
   - \mu_{11}(x)\Biggl\{
       \psi_{0,S}(w)
       - \frac{p_0(x)}{p_1(x)}\,\psi_{1,S}(w)
     \Biggr\}, \\[1.25ex]
\phi_{0,10}(w)
&:= \frac{e^{10}(x)}{1-p_0(x)}\,\psi_{0,Y(1-S)}(w)
   - \mu_{00}(x)\Biggl\{
       \psi_{1,1-S}(w)
       - \frac{1-p_1(x)}{1-p_0(x)}\,\psi_{0,1-S}(w)
     \Biggr\}, \\[1.25ex]
\phi_{1,11}(w)
&:= \frac{e^{11}(x)}{p_1(x)}\,\psi_{1,YS}(w)
   + \mu_{11}(x)\Biggl\{
       \psi_{0,S}(w)
       - \frac{p_0(x)}{p_1(x)}\,\psi_{1,S}(w)
     \Biggr\}, \\[1.25ex]
\phi_{0,11}(w)
&:= \psi_{0,YS}(w), \qquad
\phi_{1,00}(w)
:= \psi_{1,Y(1-S)}(w), \quad \text{and} \\[1.25ex]
\phi_{0,00}(w)
&:= \frac{e^{00}(x)}{1-p_0(x)}\,\psi_{0,Y(1-S)}(w)
   + \mu_{00}(x)\Biggl\{
       \psi_{1,1-S}(w)
       - \frac{1-p_1(x)}{1-p_0(x)}\,\psi_{0,1-S}(w)
     \Biggr\}.
\end{align*}
Also define
\[
g^{00}(w):=\psi_{1,1-S}(w),\qquad
g^{10}(w):=\psi_{1,S}(w)-\psi_{0,S}(w),\quad \text{and}\quad
g^{11}(w):=\psi_{0,S}(w).
\]
\end{definition}

\noindent In \citet{jiang2022multiply}, the PCE $\tau^u$ is identified as the root of the EIF-based equation
\begin{equation*}
    \mathbb{E}\left[\frac{\phi_{1,u}(W) - \phi_{0,u}(W) - \tau^u g^u(W)}{e^u}\right] =0. 
\end{equation*}
We obtain an EIF-based identification for $\tau^{u}(x)$ by replacing $\tau^{u}$ with $\tau^{u}(x)$ and projecting the EIF estimating equation onto $X$. The following theorem states the result.

\begin{theorem}\label{the:eif}
Under Assumptions \ref{assump:con}--\ref{assump:ps}, the CPCEs can be identified via
\begin{equation}\label{eq:tau_eif}
\tau^{u}_{\mathrm{eif}}(x)
:= \mathbb{E}\!\left[
\frac{\phi_{1,u}(W)-\phi_{0,u}(W)}
     {\mathbb{E}\!\left(g^{u}(W)\mid X\right)}
\ \Bigg|\ X=x\right],
\quad \text{for all }  x.
\end{equation}
\end{theorem}
For each stratum $u$, the EIF denominator satisfies $\mathbb{E}\!\left[g^{u}(W)\mid X=x\right]=e^{u}(x),$
thereby linking the EIF representation to the corresponding principal score. Note this representation leverages the full dataset instead of subset.
We now examine the robustness properties of the EIF plug-in limit parameter.

Let $\tilde\eta={\tilde\pi(\cdot),\tilde p_z(\cdot),\tilde\mu_{zs}(\cdot)}$ denote an arbitrary collection of nuisance functions. Replacing $\eta$ by $\tilde\eta$, defines plug-in EIF components $\tilde\phi_{1,u}(W)$, $\tilde\phi_{0,u}(W)$, and $\tilde g^{u}(W)$, and the pseudo-outcome
$\frac{\tilde\phi_{1,u}(W)-\tilde\phi_{0,u}(W)}
        {\mathbb{E}\!\left[\tilde g^{u}(W)\mid X\right]}$.
We define the EIF plug-in limit parameter as
\[
\tilde{\tau}^{u}_{\mathrm{eif}}(x)
:= \mathbb{E}\!\left[
\frac{\tilde{\phi}_{1,u}(W) - \tilde{\phi}_{0,u}(W)}
     {\mathbb{E}(\tilde g^{u}(W)\mid X)}
\, \Bigg|\, X=x\right], \quad \text{for all } x.  
\]
As in the subset setting, we treat $\tilde\pi(\cdot)$, $\tilde p_z(\cdot)$, and $\tilde\mu_{zs}(\cdot)$ as fixed plug-in functions (e.g., learned on an independent sample). The next theorem characterizes the bias of $\tilde{\tau}^{u}_{\mathrm{eif}}(x)$.

\begin{theorem}\label{robustness_eif}
Under Assumption \ref{assump:con}--\ref{assump:ps}, the EIF plug-in limit parameter has the following robustness properties:

\begin{itemize}
    \item \textbf{(Weak multiple robustness).}  
 Suppose, for all $x$, $\tilde\pi(x)=\pi(x)$
and $\tilde p_z(x)=p_z(x)$,   
or $\tilde\mu_{zs}(x)=\mu_{zs}(x)$, 
then, for each $u\in\{00,10,11\}$, $\tilde\tau^{u}_{\mathrm{eif}}(x)=\tau^{u}(x)$.

    \item \textbf{(Rate multiple robustness).}  
   Suppose, for all $x$, an overlap condition holds: there exists $\epsilon>0$ such that, almost surely, $\epsilon \le \mathbb{E}\bigl[\tilde g^{u}(W) \mid X =x\bigr]$, $
\epsilon \le \tilde\pi(x) \le 1-\epsilon$, $\epsilon \le \tilde p_1(x)$, $
\tilde p_0(x) \le 1-\epsilon$, $\tilde p_1(x) - \tilde p_0(x) \ge \epsilon.$
    Then, for each $u\in\{00,10,11\}$, $\tilde{\tau}_{\mathrm{eif}}^{u}(x)-\tau^{u}(x)$ equals
    \[
O_p\!\left(
\Bigl(|\tilde\pi-\pi|
+\max_{z\in\{0,1\}}|\tilde p_{z}-p_{z}|\Bigr)
\cdot
 \max_{s,z \in \{0,1\}}|\tilde\mu_{zs}-\mu_{zs}|
\right).
\]
\end{itemize}
\end{theorem}
Theorem \ref{robustness_eif} shows  a type of  multiple robustness result for the EIF plug-in limit parameter. Even though three nuisances appear, this is not ``triple robustness'' in the usual sense of ``any two of three'' sufficing for correctness. Instead, the EIF plug-in limit parameter is robust under two paths to validity: one of the paths requires simultaneous correctness of both $\pi(x)$ and $p_z(x)$, whereas the other allows $\mu_{zs}(x)$ alone to be correct.  This robustness structure mirrors that of the subset plug-in limit parameter, since the subset propensity score can be expressed in terms of   $\pi(x)$ and $p_z(x)$. However, the EIF plug-in limit parameter is weaker than subset , as it cannot rely on correct specification of the subset propensity score as a single nuisance function.

While the EIF-based estimator is theoretically appealing, we have found that it can be fragile in practice.
Unlike the subset-identified approach, it requires three nonparametric regressions: one for the nuisance functions, 
one for the pseudo-outcome $\frac{\phi_{1,u}(W)-\phi_{0,u}(W)}
     {\mathbb{E}\!\left(g^{u}(W)\mid X\right)}$ and another for the denominator pseudo-outcome $g^{u}(W)$.
The resulting ratio structure makes the estimator highly sensitive to errors in the
estimated denominator. In finite samples, even small inaccuracies in
estimated denominator can be magnified after inversion, leading to
instability and inflated variance, especially under weak overlap, where
$\mathbb{E}(g^{u}(W)\mid X)$ may be close to zero. Moreover,
the regression used to estimate the denominator typically needs to be trained on data
independent of the data used to form the pseudo-outcomes (e.g., via sample
splitting), which effectively reduces the usable sample size per fold.
Finally, to avoid instability from estimating $\mathbb{E}(g^{u}(W)\mid X)$, one could
substitute a plug-in nuisance-based quantity $\tilde e^{u}(X)$ in the denominator.
This seemingly benign change alters the estimating equation: the key cancellation
(equivalently, Neyman orthogonality) that removes first-order bias under partial
nuisance misspecification is no longer guaranteed. Consequently, the resulting
estimator may lose the robustness properties enjoyed by the EIF-based estimator.

In light of these drawbacks, it is natural to ask whether alternative estimators can be developed to address these issues. Motivated by the classical technique of one-step estimation, we propose to introduce a preliminary estimator that aims to combine the robustness benefits of EIF-based approaches with improved stability and reliability in finite samples.

\subsection{{One-step identification}}

We next consider a one-step estimator obtained by updating a preliminary estimate $\check{\tau}^u(x)$ with the EIF residual computed at $\check{\tau}^u(x)$
\begin{theorem}\label{one-step}
Under Assumptions \ref{assump:con}–\ref{assump:ps}, CPCEs can be identified from a one-step perspective. 
Let $\check{\tau}^{u}(X)$ denote any preliminary estimator of $\tau^u(X)$. Then,
\[
\tau^{u}(x) \;=\; \mathbb{E}\!\left(\check{\tau}^{u}(X)+\frac{\phi_{1,u}(W)-\phi_{0,u}(W)-\check{\tau}^{u}(X)\,g^{u}(W)}{e^u(X)}\mid X =x\right), \quad \text{for all }x. 
\]
\end{theorem}

Theorem \ref{one-step} provides a one-step identification formula for the CPCEs.
Intuitively, the one-step functional ``repairs'' a preliminary plug-in estimator
by adding an augmentation term that projects the remaining discrepancy onto
functions of $X$ in a manner that is orthogonal to first-order nuisance
estimation errors. This construction avoids directly
smoothing unstable ratios, improves finite-sample stability, and, as we show
below, inherits the same weak and rate multiple-robustness properties as the
EIF plug-in limit parameter for suitable choices of the preliminary estimator. In
practice, the preliminary estimator can be any reasonable estimator,
including a T-learner, the subset-based estimator, or the EIF-based estimator
discussed above.

Among these estimators, the T-learner is especially attractive: it is straightforward to implement, relies only on the outcome regression functions that have already been estimated, and therefore offers computational savings. Moreover, as we show later, 
when using T-learner as the preliminary estimator, 
the one-step limit parameter inherits the same robustness structure as the EIF plug-in limit parameter. The following theorem establishes the robustness properties of the one-step plug-in limit parameter.

Define the plug-in pseudo-outcome
$\tilde{\zeta}_{\tau^u}(W)
:= \check{\tau}^{u}(X)
+ \frac{\tilde \phi_{1,u}(W)-\tilde \phi_{0,u}(W)-\check \tau^{u}(X)\tilde g^{u}(W)}
{\tilde e^{u}(X)} .$
Then the one-step plug-in limit parameter is defined by
\[
\tilde{\tau}^{u}_{\mathrm{one}}(x)
:= \mathbb{E}\!\left[
  \tilde{\zeta}_{\tau^u}(W)
  \,\Bigg|\, X = x
\right], \quad \text{for all } x. 
\]
Here, the quantities $\tilde{\phi}_{z,u}(w)$ and $\tilde g^{u}(w)$ are defined as in the
EIF population function above. The term $\tilde e^{u}(x)$ denotes the estimated
probability of belonging to principal stratum $u$; for example, for
$u=00$ we have $\tilde e^{00}(x) = 1 - \tilde p_{1}(x)$.

\begin{theorem}\label{robustness_one}
Under Assumption \ref{assump:con}--\ref{assump:ps}, the one-step plug-in limit parameter has the following robustness properties:

\begin{itemize}
  \item \textbf{(Weak multiple robustness).}  
Suppose that, for all $x$, at least one of the following holds:
(a) $\tilde\pi(x)=\pi(x)$ and $\tilde p_z(x)=p_z(x)$;
(b) $\tilde\mu_{zs}(x)=\mu_{zs}(x)$ and $\tilde p_z(x)=p_z(x)$;
(c) $\tilde\mu_{zs}(x)=\mu_{zs}(x)$ and $\check\tau^u(x)=\tau^u(x)$.
Then, for each $u\in\{00,10,11\}$, $\tilde\tau^u_{\mathrm{one}}(x)=\tau^u(x)$.

  \item \textbf{(Rate multiple robustness).}  
  Suppose, for all $x$, an overlap condition holds: there exists $\epsilon > 0$ such that,
  almost surely, $  \epsilon \le \tilde\pi(x) \le 1-\epsilon,\quad
    \epsilon \le \tilde p_1(x),\,\tilde p_0(x) \le 1-\epsilon,\quad
    \tilde p_1(x) - \tilde p_0(x) \ge \epsilon$.
  Then, for each $u\in\{00,10,11\}$, $\tilde{\tau}_{\mathrm{one}}^{u}(x)-\tau^{u}(x)$ equals
\begin{equation}\label{eqonebias}
O_p\!\Biggl(
\Bigl(|\tilde\pi-\pi|
+\max_{z\in\{0,1\}}|\tilde p_{z}-p_{z}|\Bigr)
\cdot
\max_{s,z \in \{0,1\}}|\tilde\mu_{zs}-\mu_{zs}|
+\;|\check \tau^u - \tau^u|\,\cdot \max_{z\in\{0,1\}}|\tilde p_{z}-p_{z}|
\Biggr).
\end{equation}
\end{itemize}
\end{theorem}

{Theorem~\ref{robustness_one} establishes a multiple robustness property for the one-step plug-in limit parameter. Although the conditions may appear weaker at first glance, the result is in fact comparable to the robustness guarantees for the other plug-in limit parameters.
From \eqref{eqonebias}, it is important to note that if $\check \tau^u(x)$ is consistent, then the last term vanishes. In this case, the one-step plug-in limit parameter attains the same robustness structure as the EIF plug-in limit parameter. More interestingly, when $\check{\tau}^u(x)$ is chosen to be the T-learner, the last term can be expressed as a product involving the outcome regression error and the principal score error; consequently, the one-step plug-in limit parameter again attains the same robustness structure as the EIF plug-in limit parameter.}

So far, we have focused on identification of the CPCEs and properties of the
corresponding plug-in limit parameters. In the next section, we turn to
estimation: we describe how to implement the proposed estimators in practice
using flexible regression methods, and we study their statistical properties
when the conditional expectations are estimated via smoothing
regression.

\section{Estimation and error bounds under smoothness}\label{s:4}
In this section, we describe the subset, EIF, and one-step estimators and derive error bounds for the linear-smoother regression step under standard smoothness assumptions.

\subsection{Estimation Procedure}
Two-stage estimation procedures have been shown to yield robust performance and
have become a standard approach for estimating the CATE
\citep{kennedy2024semiparametric,foster2023orthogonal,knaus2022double}.
As discussed above, both the subset estimator and the one-step estimator can be
implemented within an analogous two-stage framework, requiring two sets of
regressions. In the first stage, we estimate the necessary nuisance functions,
including the propensity score, principal scores, and stratum-specific outcome
regressions. In the second stage, we construct a plug-in pseudo-outcome and regress it on the covariates $X$ to obtain an estimate of the
CPCE. We propose this two-stage strategy for the subset and one-step estimators. 
As emphasized by \citet{chernozhukov2018double}, sample splitting is a crucial component of two-stage estimation, as it prevents dependence between errors from first-stage nuisance estimation and the second-stage regression for the target. To improve the efficiency of sample splitting, cross-fitting is commonly used.

In contrast, the EIF-based estimator typically requires an additional regression,
leading to a three-stage procedure. In particular, beyond estimating the
nuisance functions and regressing the pseudo-outcome on $X$, one must
also estimate the conditional expectation of the denominator component,
$\mathbb{E}(\tilde g^{u}(W) \mid X)$. 
Hence, implementation the EIF-based estimator may require  a three-way
sample split: one subsample to estimate the nuisance functions, a second to
estimate the conditional expectation of the denominator component
$\tilde g^{u}$, and a third to regress the resulting
pseudo-outcome on $X$. In contrast to the two-stage setting, extending this
three-way splitting scheme to a convenient $K$-fold cross-fitting procedure is
less straightforward and efficiency is lost due to its inability to use most of the data for any of the three stages. Plot \ref{fig:procedure} in the Supplementary \ref{s:procedure} provides further detail and intuition.

For all multi-stage estimators, nuisance-function estimation and the regressions of the pseudo-outcome (and, when applicable, the denominator) on $X$ are highly flexible.
In particular, modern machine learning methods, such as random forests \citep{breiman2001random}, gradient boosting \citep{friedman2001greedy}, lasso regression \citep{tibshirani1996regression} or other nonparametric learners may be readily applied to this stage. The use of such flexible estimators allows one to capture complex, high-dimensional relationships without restrictive parametric assumptions.

\subsection{Background on stability of smoothing operators}

As introduced in the estimation procedure, the regression step is highly flexible, yielding model-free guarantees that make the approach broadly applicable in practice. We will use a general framework developed by \cite{kennedy2023towards} for analyzing the error of two-stage regression estimators. For simplicity, and following \citet{kennedy2023towards}, we present our theoretical results under equal sample splitting. 
Specifically, for the subset and one-step estimators we assume a sample of size $2n$, with $n$ observations used for nuisance estimation and $n$ observations used to fit the second-stage regression operator, denoted $\widehat{\mathbb{E}}_n(\cdot\mid X=x)$. 
When the second-stage regression is restricted to a subset $\mathcal S_u$ we  denote the operator by $\widehat{\mathbb{E}}_n(\cdot\mid \mathcal S_u, X=x)$. 
For the EIF estimator, we assume a sample of size $3n$, consisting of three independent subsamples of size $n$ for nuisance estimation, denominator regression, and second-stage regression, respectively. 
We later comment on extensions to unequal splitting; these require only minor notational changes and do not affect the conclusions.

Under equal sample splitting, \citet{kennedy2023towards} introduce a notion of stability for a broad class of estimators, including generic linear smoothers; we adapt their framework to our setting.
For completeness, we present the formal definitions of stability and the associated distance metric below.

\begin{definition}[Stability; \citet{kennedy2023towards}]\label{def:stability}
Let $D_1^n$ (training) and $D_2^n$ (estimation) be independent samples of size $n$.
Let $\widetilde f(\cdot)=\widetilde f(\cdot;D_1^n)$ be an estimator
of a function $f:\mathcal W\to\mathbb R$ trained on $D_1^n$. 
Define the plug-in bias function $\widetilde b(x)
:=\mathbb E\!\left[\widetilde f(W)-f(W)\mid D_1^n,\;X=x\right].$
Let $\widehat{\mathbb E}_n(\,\cdot\mid X=x)$ denote a generic second-stage regression operator
fit on the estimation sample $D_2^n$, mapping responses $\{R_i\}_{i=1}^n$ to a prediction at $X=x$. (In our applications, $R_i$ will be a pseudo-outcome
such as $\varphi_{\tau^u}(W_i)$ or $\zeta_{\tau^u}(W_i)$, possibly restricted to a subset of observations.)
For a chosen distance metric $d(\cdot,\cdot)$ on functions, we say that $\widehat{\mathbb E}_n$ is \emph{stable at $X=x$ with respect to a distance $d$ on functions $f(W)$}
if, whenever $d(\widetilde f,f)\xrightarrow{p}0$, it holds that
\[
\frac{
\widehat{\mathbb E}_n\{\widetilde f(W)\mid X=x\}
-
\widehat{\mathbb E}_n\{f(W)\mid X=x\}
-
\widehat{\mathbb E}_n\{\widetilde b(X)\mid X=x\}
}{
\sqrt{
\mathbb E\!\left[
\Bigl\{
\widehat{\mathbb E}_n\{f(W)\mid X=x\}
-
\mathbb E\{f(W)\mid X=x\}
\Bigr\}^2
\right]
}
}
\;\xrightarrow{p}\;0 .
\]
\end{definition}

\begin{theorem}[Stability of linear smoothers; \citet{kennedy2023towards}]\label{thm:linear_stable}
Suppose the second-stage regression operator is a \emph{linear smoother}, i.e.,
for any responses $\{R_i\}_{i=1}^n$ and covariates $\{X_i\}_{i=1}^n$ from $D_2^n$, $\widehat{\mathbb E}_n(R\mid X=x)
=
\sum_{i=1}^n w_i(x;X^n)\,R_i $, $
X^n:=(X_1,\ldots,X_n).$
Then $\widehat{\mathbb E}_n$ is stable in the sense of Definition~\ref{def:stability}
with respect to the distance $d(\widetilde f,f)=\|\widehat f-f\|_{w^2}$, where
\[
\|\widetilde f-f\|_{w^2}^2
:=
\sum_{i=1}^n
\left\{
\frac{w_i(x;X^n)^2}{\sum_{j=1}^n w_j(x;X^n)^2}
\right\}
\,
\mathbb E\!\left[
\{\widetilde f(W)-f(W)\}^2 \mid X=X_i
\right].
\]

\end{theorem}
This result applies to a broad class of linear smoothers, such as kernel regression, local polynomial regression, and spline/series smoothers. It holds under a mild boundedness condition on the conditional standard deviation, namely $1/\|\sigma\|_{w^2}=O_{\mathbb P}(1)$, where $\sigma(x)^2:=\mathrm{Var}\{f(W)\mid X=x\}$. The norm $|\cdot|_{w^2}$ that defines the distance $d(\widetilde f,f)$ for linear smoothers is a natural weighted and conditional version of an $L_2(\mathbb{P})$ norm.

\subsection{Error for estimators under smoothness}
In this section, we study the errors of our three proposed estimators (subset, one-step, and EIF) when the second-stage regression is smoothed. Each estimator admits a decomposition into an oracle error term and a smoothed plug-in bias for the limiting parameter, analogous to the results of \citet{kennedy2023towards}.

We define the subset estimator and the one-step estimator, respectively, as 
\[
\widehat{\tau}_{\mathrm{sub}}^{u}(x)
:= \widehat{\mathbb{E}}_{n}\!\left[\tilde{\varphi}_{\tau^{u}}(W)
\;\middle|\; \mathcal{S}_{u},\, X=x \right] \text{ and }
\widehat{\tau}_{\mathrm{one}}^{u}(x)
:= \widehat{\mathbb{E}}_{n}\!\left[
\tilde{\zeta}_{\tau^u}(W)
\;\middle|\; X=x
\right],
\]
where $\tilde{\varphi}_{\tau^{u}}$ and $\tilde{\zeta}_{\tau^u}$ are learned
on the first-stage training sample by plugging in nuisance estimates.
For the theorem below, we define
$
 \zeta_{\tau^u}(W)
:= \tau^{u}(X)
  + \frac{{\phi}_{1,u}(W)-{\phi}_{0,u}(W)-\tau^{u}(X) g^{u}(W)}
          {e^{u}(X)}.$

\begin{theorem}\label{thm:smooth_sub_one}
Fix $u\in\{00,10,11\}$ and a target point $x$. As $n\to\infty$, assume: (i) the second--stage regression operator $\widehat{\mathbb{E}}_n$ is stable with respect to a
distance $d$, (ii) $d(\tilde{\varphi}_{\tau^u}, \varphi_{\tau^u}) \xrightarrow{p} 0$ and
$d(\tilde \zeta_{\tau^u}, \zeta_{\tau^u})\xrightarrow{p} 0$. Define the oracle risks
$R_{\mathrm{sub}}^{u,*}(x)^2:=\mathbb{E}\!\Big[\{\widehat{\mathbb{E}}_n\!\bigl(\varphi_{\tau^u}(W)\mid \mathcal S_u,\,X=x\bigr)-\tau^u(x)\}^2\Big]$
and
$\mbox{$R_{\mathrm{one}}^{u,*}(x)^2:=\mathbb{E}\!\left[\left\{\widehat{\mathbb{E}}_n\!\bigl(\zeta_u(W)\mid X=x\bigr)-\tau^u(x)\right\}^2\right]$}$.
\\
Then for the subset estimator,  $\widehat{\tau}^{u}_{\mathrm{sub}}(x)-\tau^u(x)$ can be decomposed as 
\begin{align*}
\underbrace{\widehat{\mathbb{E}}_n\!\bigl\{\varphi_{\tau^u}(W)\mid\mathcal{S}_u,\,X=x\bigr\}-\tau^u(x)}_{\text{oracle error}}
+
\underbrace{\widehat{\mathbb{E}}_n\!\bigl\{\tilde\tau^u_{\mathrm{sub}}(X)-\tau^u(X)\mid\mathcal{S}_u,\,X=x\bigr\}}_{\text{smoothed plug-in bias}}
+o_p\!\bigl(R_{sub}^{u,*}(x)\bigr)
\end{align*}
and for the one-step estimator, $\widehat{\tau}^{u}_{\mathrm{one}}(x)-\tau^{u}(x)$ can be decomposed as
\[
\underbrace{
\widehat{\mathbb{E}}_n\!\bigl\{\zeta_u(W)\mid X=x\bigr\}
-
\tau^{u}(x)}_{\text{oracle error}}+
\underbrace{\widehat{\mathbb{E}}_n\!\left[
  \tilde{\tau}^{u}_{\mathrm{one}}(X)-\tau^{u}(X)
\mid X=x
\right]}_{\text{smoothed plug-in bias}}
+o_p\!\bigl(R_{one}^{u,*}(x)\bigr).
\]
\end{theorem}
Theorem \ref{thm:smooth_sub_one} shows that the smoothed subset and one-step  estimators admit a
decomposition into an oracle estimation error term and a smoothing-bias term
associated with the plug-in limit parameters, under the mild condition that the
pseudo outcome are consistent. {The corollaries below illustrate how the result applies when the nuisance functions, the preliminary estimator, and the CPCE are learned via smoothness- or sparsity-based procedures. Let $p$ denote the covariate dimension.} 

\begin{corollary}\label{col:sub_one}
Suppose the assumptions of Theorem~\ref{thm:smooth_sub_one} hold and
$\widehat{\mathbb E}_n$ is a minimax-optimal linear smoother with
$\sum_{i=1}^n |w_i(x;X^n)| = O_p(1)$.
Assume further that the regressions $\mu_{zs}$ are $\alpha_{\mu}$-smooth and satisfy
$\|\tilde \mu_{zs}-\mu_{zs}\|_{w,2} =O_p\!\left(n^{-1/(2+p/\alpha_{\mu})}\right)$, and that
the CPCE $\tau^u$ is $\gamma$-smooth.
Then the following rates hold:
\begin{enumerate}
\item \textbf{(Subset estimator).}
If the subset propensity score $\pi_{\mathcal S_u}$ is $\alpha_u$-smooth and
$\|\tilde\pi_{\mathcal S_u}-\pi_{\mathcal S_u}\|_{w,2}
=O_p\!\left(n^{-1/(2+p/\alpha_u)}\right)$, then $\widehat{\tau}^u_{\mathrm{sub}}(x)-\tau^u(x)$ equals
\[
O_{p}\!\left(
n^{\frac{-1}{2+p/\gamma}}
+
n^{-\left(\frac{1}{2+p / \alpha_u}+\frac{1}{2+p / \alpha_{\mu}}\right)}
\right).
\]

\item \textbf{(One-step estimator).}
If the preliminary
estimator $\check\tau^u$ is $\alpha_{\gamma}$-smooth and 
$\|\check\tau^u-\tau^u\|_{2}=O_p\!\left(n^{-1/(2+p/\alpha_{\gamma})}\right)$ . And if the propensity score $\pi$ is $\alpha_{\pi}$-smooth and
$\|\tilde\pi-\pi\|_{w,2}=O_p\!\left(n^{-1/(2+p/\alpha_{\pi})}\right)$, the principal score
$p_z$ is $\alpha_p$-smooth and
$\|\tilde p_z-p_z\|_{w,2}=O_p\!\left(n^{-1/(2+p/\alpha_{p})}\right)$, 
then $    \hat\tau_{\mathrm{one}}^{u}(x)-\tau^{u}(x)$ equals
\begin{equation}\label{one_r}
O_p\!\Bigg(
n^{\frac{-1}{2+p/\gamma}}
\;+\;
n^{-\left(\frac{1}{2+p/\alpha_\pi}+\frac{1}{2+p/\alpha_{\mu}}\right)}
\;+\;
n^{-\left(\frac{1}{2+p/\alpha_p}+\frac{1}{2+p/\alpha_{\mu}}\right)}
\;+\;
n^{-\left(\frac{1}{2+p/\alpha_{\gamma}}+\frac{1}{2+p/\alpha_{p}}\right)}
\Bigg).
\end{equation}
\end{enumerate}
\end{corollary}

Corollary~\ref{col:sub_one} shows that both subset estimator and one-step estimator can attain the optimal smoothness rate even when the nuisance functions are less smooth. That is, it establishes the double robustness of the subset estimator and a form of multiple robustness of the one-step estimator. The condition for oracle efficiency of the subset estimator matches the oracle efficiency condition for the DR-learner in \citet{kennedy2023towards}. 
When the preliminary estimator of the one-step estimator $\check{\tau}^u$ is chosen as the $T$-learner, the final term in~\eqref{one_r} is absorbed into the third term. As a result, the one-step estimator exhibits a ``double robustness'' structure with respect to the three nuisance components: it is consistent if either the outcome mean models are correctly specified, or both the propensity score and the principal score are correctly specified. Moreover, under this choice of $\check{\tau}^u$, the one-step estimator satisfies the same oracle-efficiency condition as the subset estimator, namely
\[
\sqrt{r\,\alpha_\mu}\ \ge\
\frac{p/2}{\sqrt{\,1+\frac{p}{\gamma}\Bigl(1+\frac{p}{2\bar{h}(r)}\Bigr)\,}},
\qquad
\bar{h}(r):=\left(\frac{r^{-1}+\alpha_\mu^{-1}}{2}\right)^{-1},
\]
where for the subset, $r$ is $\alpha_u$ and $\bar{h}(r)$ is the harmonic mean of $(\alpha_u,\alpha_\mu)$. 
For the one-step estimator ($\check \tau^u$ is T-learner),  $r$ is $\max\{\alpha_\pi,\alpha_p\}$ and $\bar{h}(r)$ is the harmonic mean of $(\max\{\alpha_\pi,\alpha_p\},\alpha_\mu)$. In particular, the oracle-efficiency condition is driven by the smoother of the principal score and propensity score components.

Moreover, this decomposition suggests a simple, approximate {\it inference strategy}. Because the additional plug-in term is asymptotically negligible relative to the oracle error under the stated conditions, the pointwise sampling variability of $\widehat{\tau}_{sub}^u(x)$ or $\widehat{\tau}_{one}^u(x)$ is asymptotically the same as that of the oracle regression (i.e., the second-stage regression applied to the true pseudo-outcome). Therefore, we may construct pointwise confidence intervals using the same standard-error formula as for the oracle estimator. In practice, since the true pseudo-outcome is unobserved, we compute the standard error from the second-stage regression applied to the plug-in pseudo-outcomes.

Unlike the subset and one-step estimators, the EIF result has a slightly different structure; see details in Supplementary \ref{s:error_eif}. In particular, the leading term involves a smoothed, weighted population-function bias rather than a (pure) smoothed bias. However, this weight is close to one, since it is given by the ratio between the population denominator and its estimator. Moreover, our results show that the oracle error contains two terms because the EIF estimator involves two smoothed regressions. This not only makes the oracle error difficult to evaluate directly, since the $\tau^u(x)$ is unknown, but also makes the oracle error larger than that of the one-step estimator, which requires only a single second-stage smoothing step. Consequently, when the preliminary estimator for the one-step procedure is the $T$-learner, the plug-in errors for the one-step and EIF estimators have the same structure, but the EIF oracle error is larger due to the additional smoothing step. The corresponding smoothness result for the EIF estimator and the oracle efficient condition are deferred to Supplementary~\ref{col:eif}.

\section{Simulation}\label{s:5}
In the simulation section, we conduct two studies. The first examines robustness properties using parametric regression methods, while the second evaluated performance under flexible machine learning methods.

We draw $X=(X_1,\ldots,X_4)\in\mathbb{R}^4$ with independent coordinates, where $X_i \sim \mathrm{Unif}(0,1)$, and generate an additive outcome noise term $\varepsilon \sim \mathcal{N}(0, .2^2)$. We also specify a baseline outcome function $b(\cdot)$.  For  $\pi(\cdot)$ and $p_z(\cdot)$ to be discussed shortly, 
we let
\begin{equation}\label{eq:d}
\begin{split}
   & Z \mid X \overset{\text{i.i.d.}}{\sim}\mathrm{Ber}\{\pi(X)\}, 
 \quad S \mid (Z=z, X) \sim \text{Ber}\{p_z(X)\},  \text{ and} \\[0.5em]
   & Y = \mu_{zs}(X) + \epsilon 
     = b(X) + 0.5 Z X_1 + 0.5\, S X_1 - 0.5\, Z S X_1 +\epsilon .
\end{split}
\end{equation}
From~\eqref{eq:d},  the true conditional principal causal effects are identical across strata, with $\tau_{11}(X) = \tau_{10}(X) = \tau_{00}(X) = 0.5X_1$. We used the T-learner as the preliminary estimator for the one-step estimator. To assess robustness under different data-generating mechanisms, we consider four settings: (i) all nuisance functions $\pi(\cdot)$, $p_z(\cdot)$, $\mu_{zs}(\cdot)$ are correctly specified; (ii) $\pi(\cdot)$ and $p_z(\cdot)$ are misspecified but $\mu_{zs}(\cdot)$ are correct; (iii) $\mu_{zs}(\cdot)$ are misspecified but $\pi(\cdot)$ and $p_z(\cdot)$  are correct; and (iv) all nuisance functions are misspecified. The remaining details for the data generating process (namely, $b$, $\pi$, $p_z$) are in Supplementary \ref{app:sim}. All nuisance functions and the second-stage regression were estimated using parametric models. In scenarios where the true nuisance functions are linear (and propensity/other probability models follow a logistic regression with a linear predictor), the nuisance models are correctly specified; when the true functions are nonlinear, these parametric specifications are misspecified.

{Figure~\ref{fig:robustness for tau10} reports the log-RMSE for estimating $\tau^{10}(x)$ across sample sizes under four nuisance-specification regimes, providing empirical support for our theoretical robustness results. In all panels, RMSE decreases with $n$, indicating improved accuracy as the sample size grows.  Since the T-learner relies essentially on differences of two outcome regressions, it attains its theoretically best performance when $\mu_{zs}(\cdot)$ is correctly specified, but becomes inconsistent when $\mu_{zs}(\cdot)$ is misspecified. When all models are correctly specified, the subset and one-step estimators attain the lowest RMSE, matching the T-learner, while the EIF has larger RMSE, consistent with our theory that its oracle error is larger. When $\pi(\cdot)$ and $p_z(\cdot)$ are misspecified but $\mu_{zs}(\cdot)$ is correctly specified, or when $\mu_{zs}(\cdot)$ is misspecified but $\pi(\cdot)$ and $p_z(\cdot)$ are correctly specified, the subset, one-step, and EIF estimators remain consistent, though the EIF performs worst at small sample sizes. Finally, when all nuisance models are misspecified, the subset and one-step estimators continue to deliver the best performance.
}

\begin{figure}
    \centering
\includegraphics[width=0.75\linewidth]{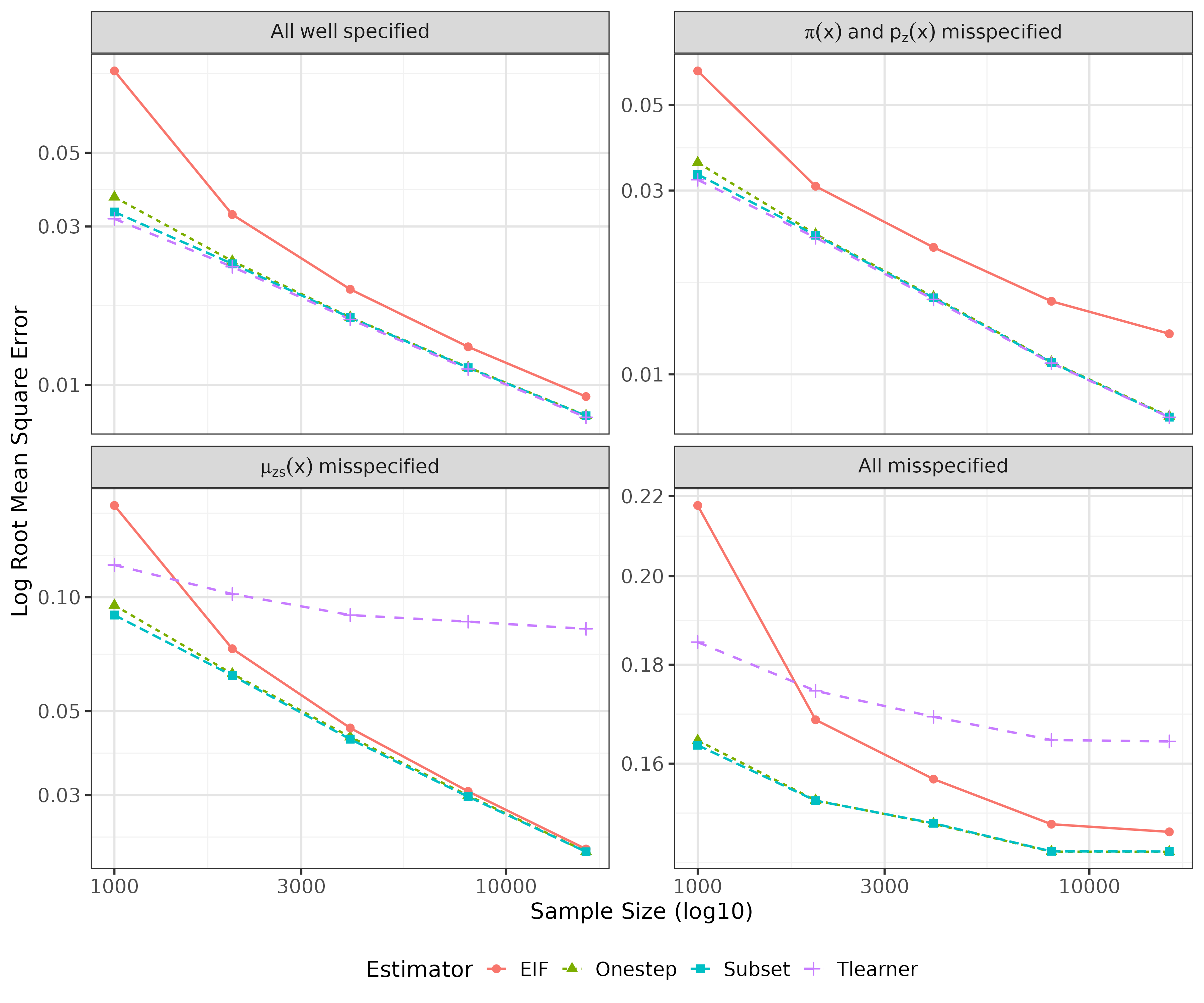}
    \caption{Root mean squared error (RMSE, log scale) of $\tau^{10}(X)$ estimators across sample sizes (1,000–16,000). Results compare the T-learner, subset, EIF, and one-step estimators.}
    \label{fig:robustness for tau10}
\end{figure}
In a second simulation study, we evaluate finite-sample performance under flexible, data-adaptive modeling by replacing the previously used parametric models with generalized additive models (GAMs).
In this setting, the nuisance functions are nonlinear in $X$ but they are additive and smooth, so the GAM working models are correctly specified. See Supplementary~\ref{app:sim} for details of the data-generating process. Figure~\ref{fig:gam tau00} reports the RMSE for estimators of $\tau^{00}(X)$ under this setting.

\begin{figure}
    \centering    \includegraphics[width=0.75\linewidth]{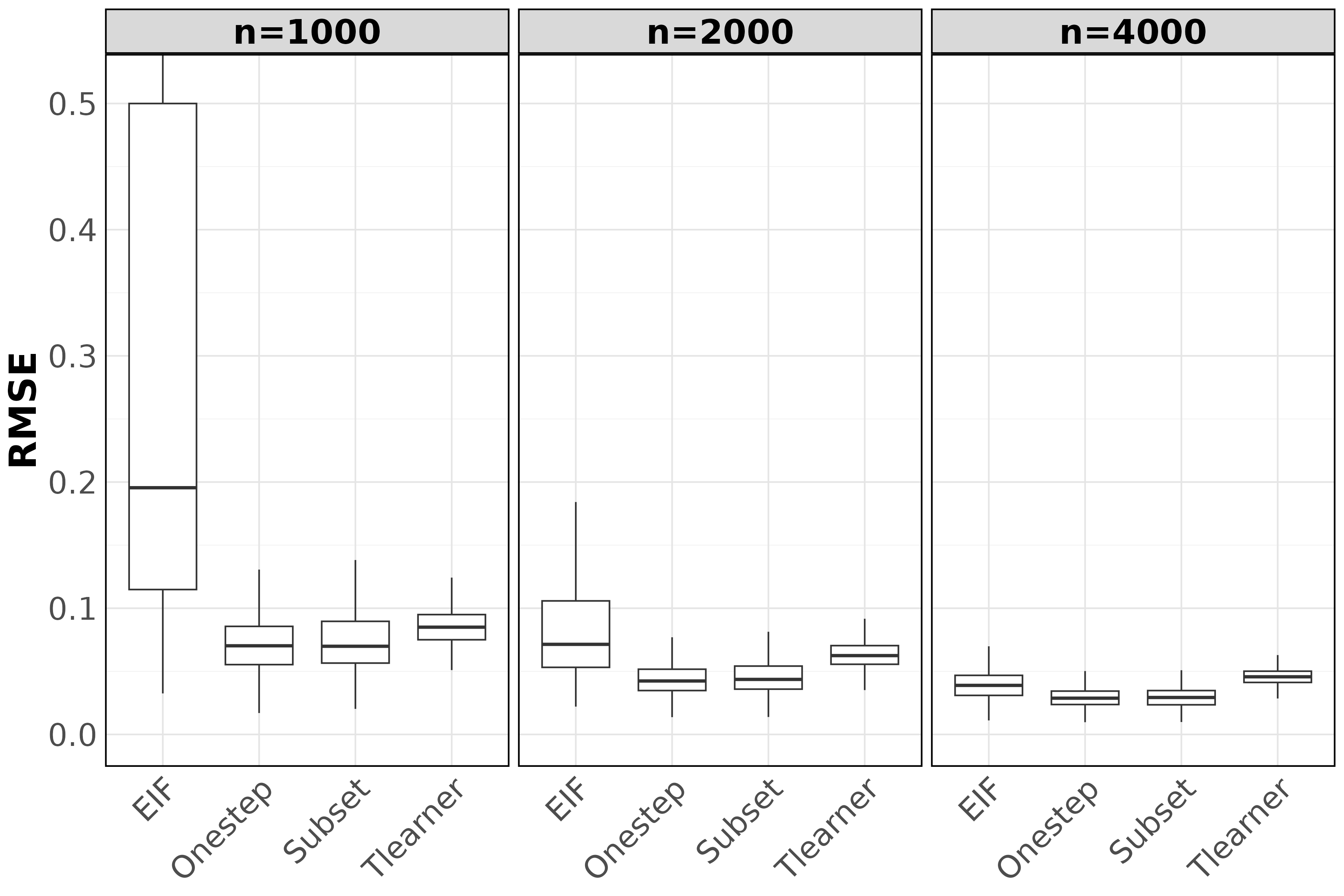}
    \caption{Boxplots display the distribution of RMSE for four estimators (EIF, One-step, Subset, and Tlearner) estimated by GAM model across three sample sizes ($n=1000$, $n=2000$, $n=4000$)}
    \label{fig:gam tau00}
\end{figure}
From Figure~\ref{fig:gam tau00}, the subset and one-step estimators attain the lowest RMSE at every sample size, and their performance is nearly indistinguishable. And the T-learner is uniformly less accurate than the subset and one-step estimators. The EIF estimator’s RMSE is large and highly variable at $n=1000$, but drops sharply and becomes much more concentrated as $n$ increases. This indicates that EIF is sample size sensitive, it can perform poorly in small samples but improves substantially with larger $n$, eventually becoming comparable to the other methods.

In addition, we conduct a follow-up simulation in Study 2 in which the target estimand $\tau^u(x)$ is generated from a nonlinear, nonparametric function of $X$, rather than the earlier linear form $0.5X_1$. The result in Supplementary \ref{app:sim} show a similar pattern. The third study further examines the finite-sample performance of the subset and one-step estimators under overlap violations and unbalanced subset sizes. Figure~\ref{fig:gam_tau11_s3} indicates that the subset estimator performs worse, whereas the one-step estimator is more stable.

\section{Analysis of health care hotspotting data}\label{s:6}
We illustrate our methods by estimating the CPCE among compliers in the “Health Care Hotspotting” RCT. The trial randomized superutilizer patients with high healthcare needs to either standard care or an intensive follow-up program after hospital discharge. Although the original trial reported null overall effects \citep{finkelstein2020health}, a secondary analysis focusing on high engagers found significant intervention benefits \citep{yang2023hospital},  indicating heterogeneity. Our goal is to help elucidate whether the heterogeneity in effect seen is a result of heterogeneity in who engages or is a reflection of true heterogeneity in the effect of the intervention.

Among those assigned to the intervention, some may later decline to participate, resulting in one-sided noncompliance. Accordingly, let $S$ denote a binary indicator of engagement (high vs. low), defined following \cite{clark2024transportability} using ``engagement" metrics recorded during trial follow-up. The primary outcome is 30-day hospital readmission, measured as a binary indicator. Because control-assigned patients could not access the intervention, the (strong) monotonicity assumption holds. A detailed description of the dataset and covariates is provided in Supplementary~\ref{app:hot}.

The ATE, estimated using the \texttt{tmle} package \citep{gruber2012tmle}, is estimated as -0.0246 (95\% CI: -0.0872, 0.0379), indicating no overall benefit and aligning with the null finding reported by \citet{finkelstein2020health}. In contrast, the complier average treatment effect, estimated following \citet{jiang2022multiply}, was -0.0775 (95\% CI: -0.1629, -0.0039), suggesting a modest but statistically significant reduction in 30-day readmissions among compliers. Finally, CATE estimates from the DR-learner reveal substantial treatment effect heterogeneity (see Supplementary~\ref{app:hot} for details). This 
particularly interesting scenario suggests that averaging can mask important structure in who benefits from treatment and through which principal strata the effect operates. Consequently, examining CPCEs provides a more refined explanation of the data.

We focus our empirical analysis on the subset and one-step estimators, since the EIF estimator displays weaker finite-sample performance in this application. The propensity score is estimated via a logistic regression model. Following \citet{yang2023hospital}, we estimate the principal scores using gradient boosting. Outcome regression nuisance functions are fit using the Super-Learner \citep{van2007super}. Because the readmission outcome is binary, we apply Hájek normalization \citep{hajek2003conditional} to the pseudo-outcomes to improve numerical stability and to yield bounded effect estimates. For the second-stage regression, we use generalized random forests \citep{athey2019grf}.

\begin{figure}[htbp]
\centering

    \includegraphics[width=0.8\linewidth]{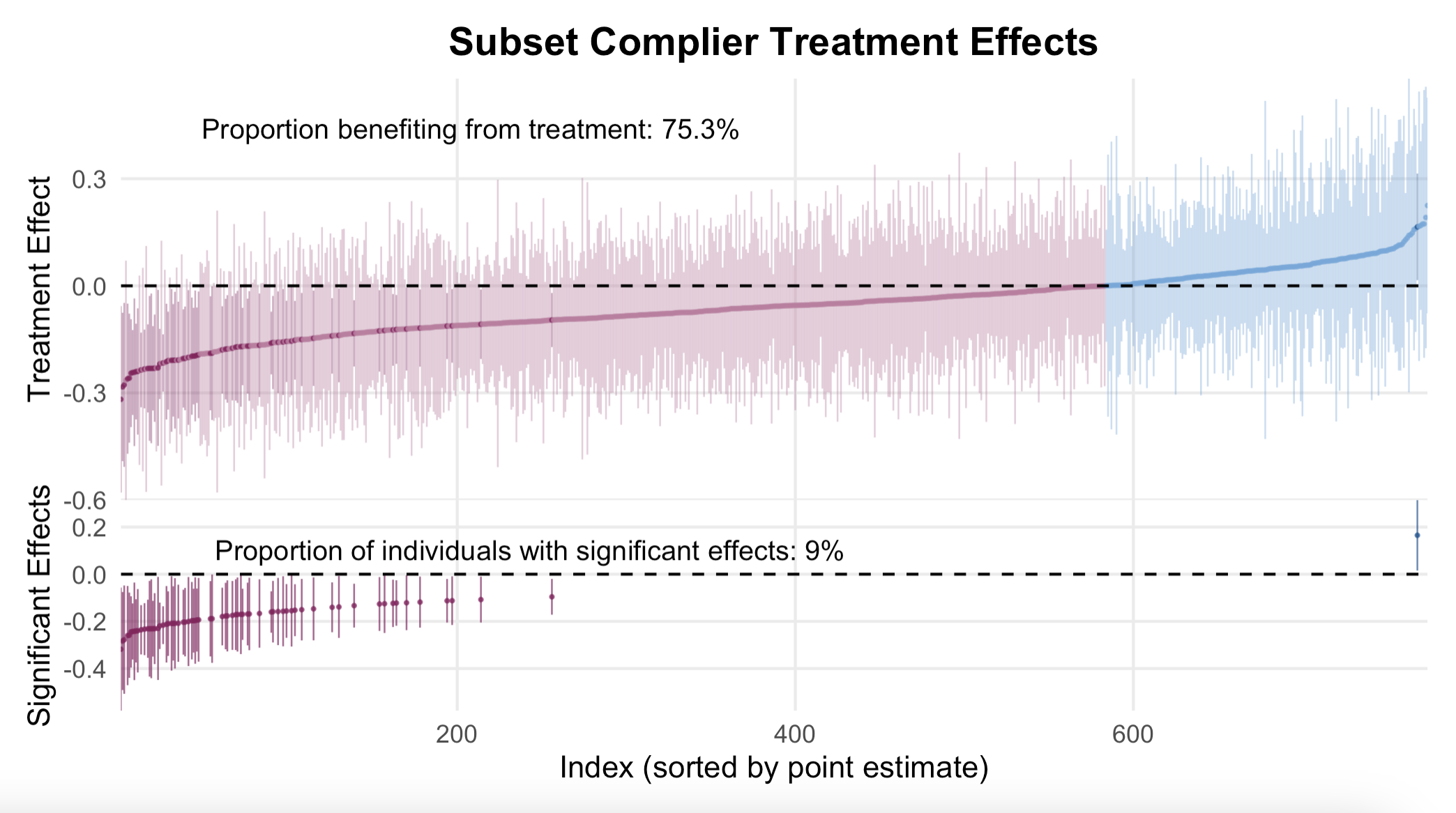}

\caption{Estimated complier-specific treatment effects, ordered by the point estimates, with 95\% confidence intervals computed by \texttt{grf}. The top subpanel displays unit-level estimates for all compliers with corresponding confidence intervals, while the bottom subpanel highlights individuals whose effects are statistically distinguishable from zero (i.e., whose confidence intervals exclude 0). Red points denote estimated beneficial effects ($\tau^{10}(x)<0$), whereas blue points indicate increased readmission risk ($\tau^{10}(x)>0$).
}
    \label{fig:hot}
\end{figure}

Figure~\ref{fig:hot} displays the estimated complier CPCEs and corresponding 95\% confidence intervals from the subset estimator. The one-step estimator yields a very similar pattern; see Supplementary~\ref{app:hot} for results. The figure shows substantial across-individual variation in unit-level CPCEs, providing evidence of treatment effect heterogeneity. Moreover, the distribution is shifted toward negative values, suggesting an overall beneficial effect. 75.3\% of unit-level effects are negative, indicating that most compliers are estimated to benefit from treatment. About 8.91\% of unit-level estimates are significantly negative, whereas only 0.13\% are significantly positive. Overall, effects among compliers appear modest, with most individuals experiencing small reductions in 30-day readmissions, motivating follow-up analyses to identify baseline covariates that drive this heterogeneity.

{From the generalized random forest fit in the second-stage regression, we assessed heterogeneity drivers using the forest’s variable-importance measure and found the highest-ranked covariates to be prior inpatient admissions in the past 180 days, duration of the initial (index) hospital stay, and sex.} 
We find that the intervention may be more effective for female compliers, whereas male compliers do not appear to benefit, which is consistent with the findings of \citet{yang2023hospital}. While \citet{yang2023hospital} reported that higher education levels were negatively associated with engagement, our results suggest that, among compliers, there is limited evidence that education is a major driver of treatment effect heterogeneity.

In summary, our analysis of the Hotspotting RCT reveals heterogeneity among compliers, with CPCEs generally modest and skewed negative, suggesting reductions in 30-day readmissions. Hospitalization history and sex were key modifiers, with female compliers benefiting more, while education showed no effect once compliance was established. These findings provide a nuanced assessment of treatment effects, offering decision makers clearer insight into which subgroups may benefit most from targeted interventions.

\section{Discussion}
In this paper, we develop a new framework for identifying, estimating, and conducting inference on heterogeneous principal causal effects under principal ignorability. The target parameters are CPCEs within the never-taker, complier, and always-taker strata, enabling treatment effect heterogeneity to be studied within principal strata rather than only through average principal effects, and thereby helping decision makers distinguish two sources of subgroup differences: (i) heterogeneity in engagement (i.e., variation in $p_z(x)$) and (ii) heterogeneity in within-stratum causal effects (i.e., variation in $\tau^u(x)$). 

We propose four estimators. A baseline approach adapts the T-learner via observed outcome regressions. We then develop three cross-fitted, machine learning compatible estimators that are robust to nuisance estimation: (i) a subset-based estimator obtained by applying a DR-learner to the observed subset; (ii) an EIF-based estimator derived from the efficient influence function for the PCEs; and (iii) a one-step estimator that refines a preliminary estimator using an influence function correction. Under smoothness conditions, we establish pointwise theory showing that the subset estimator is doubly robust, remaining consistent if either $\pi_{\mathcal S_u}(x)$ or $\mu_{zs}(x)$ are consistently estimated. The EIF and one-step estimators are multiply robust, remaining consistent if either $\mu_{zs}(x)$ are consistently estimated or $\pi(x)$, and $ p_z(x)$ are consistently estimated. Our pointwise inference procedures are built on these results.

Our work raises several open questions. First, although CPCEs separate heterogeneity in principal causal effects into variation arising from the principal-score components $p_z(x)$ and from the conditional effect function $\tau^u(x)$, it remains unclear how this decomposition can be used to derive principled policy-learning rules under principal stratification. Second, our pointwise theory is established under smoothness conditions, and extending it to weaker assumptions and more general settings is an important direction for future work. Third, it would be of interest to extend the framework of \citet{lu2026principal} to accommodate heterogeneous effects when the post-treatment variable is continuous rather than binary. These questions highlight several promising directions for future research.

\bibliography{ref}

\newpage
\section{Supplementary Material}
In this supplementary material, we often present derivations for a single principal stratum as an illustrative example; the derivations for the remaining strata are analogous. Supplementary~\ref{s:toy_example} provides a toy example illustrating the limitations of the $T$-learner. Supplementary~\ref{s:proofth2.1} proves Theorem~\ref{the1}, showing that the target function $\tau^{u}(x)$ can be expressed as a difference of observed outcome mean models. Supplementary~\ref{s:proofth3.1} proves Theorem~\ref{the:subset}, showing that our target can be identified from a subset perspective. Supplementary~\ref{s:subpropensity} shows a second way to express the subset propensity score function. Supplementary~\ref{s:proofrobsub} proves Theorem~\ref{robustness_subset}, which establishes weak robustness and rate robustness for the subset plug-in limit parameter. Supplementary~\ref{s:proof5} proves Theorem~\ref{the:eif}, CPCEs can be identified from the corresponding component of the EIF for the principal causal effect (PCE).  Supplementary~\ref{s:proof6} proves the multiple robustness of EIF identification of EIF plug-in limit parameter. Supplementary~\ref{s:proof7} proves Theorem~\ref{one-step}, showing that our target can be identified in a one-step perspective. Supplementary~\ref{s:proof8} proves the multiple robustness of one-step identification of one-step plug-in limit parameter. Supplementary~\ref{s:proof9} contains proofs of the error decompositions for the subset and one-step estimators. The detailed proof of the EIF estimator decomposition theorem is given in Supplementary~\ref{s:error_eif}. Supplementary~\ref{s:proof10} provides the proof of the bias decomposition for the smoothed EIF estimator. The corresponding corollaries for these theorems are collected in Supplementary~\ref{s:corollary}. Supplementary~\ref{s:procedure} describes the details of the estimation procedure. Supplementary~\ref{app:sim} and Supplementary~\ref{app:hot} provide additional results for the simulation studies and the Hotspotting analysis, respectively.

\subsection{Toy Example}\label{s:toy_example}
We present a toy example, modified from \citet{kennedy2023towards}, to illustrate the limitations of the T-learner. Consider the following data-generating process: covariates are drawn as \( X \sim \mathcal{U}(-1, 1) \), and the treatment assignment mechanism (propensity score) and principal score are defined as $ \pi(x) = 0.25 + 0.5\times \text{sign}(x)$, and 
$ p_z(x) = (0.35 + 0.2\times\text{sign}(x))\times1_{Z=1}+(0.65 - 0.2\times\text{sign}(x))\times1_{Z=0}$. This setup ensures that the marginal proportions across the four observed subgroups remain approximately balanced, yet the treatment assignment becomes highly imbalanced conditional on $X$. We consider $\mu_{00}(x)= \mu_{10}(x)=\mu_{01}(x)=\mu_{11}(x)$ are equal to the piecewise polynomial function defined on
page 10 of \cite{gyorfi2006distribution}:
$$\mu_{zs}(x)= \begin{cases}(x+2)^2 / 2 & \text { if }-1 \leq x \leq-0.5 \\ x / 2+0.875 & \text { if }-0.5 \leq x<0 \\ -5(x-0.2)^2+1.075 & \text { if } 0<x \leq 0.5 \\ x+0.125 & \text { if } 0.5 \leq x<1\end{cases}.$$

As a result, although each subgroup remains adequately represented in the overall population, the distribution of units is heavily skewed within specific regions of the covariate space. This imbalance is problematic for the T-learner, which fits separate outcome models and may over-smooth where data are sparse and under-smooth where data are dense. Such instability can lead to biased estimation of the CPCE and spurious heterogeneity in the estimated treatment effects.

\begin{figure}[]
    \centering
    \begin{subfigure}{0.45\textwidth}
        \centering
        \includegraphics[width=\textwidth]{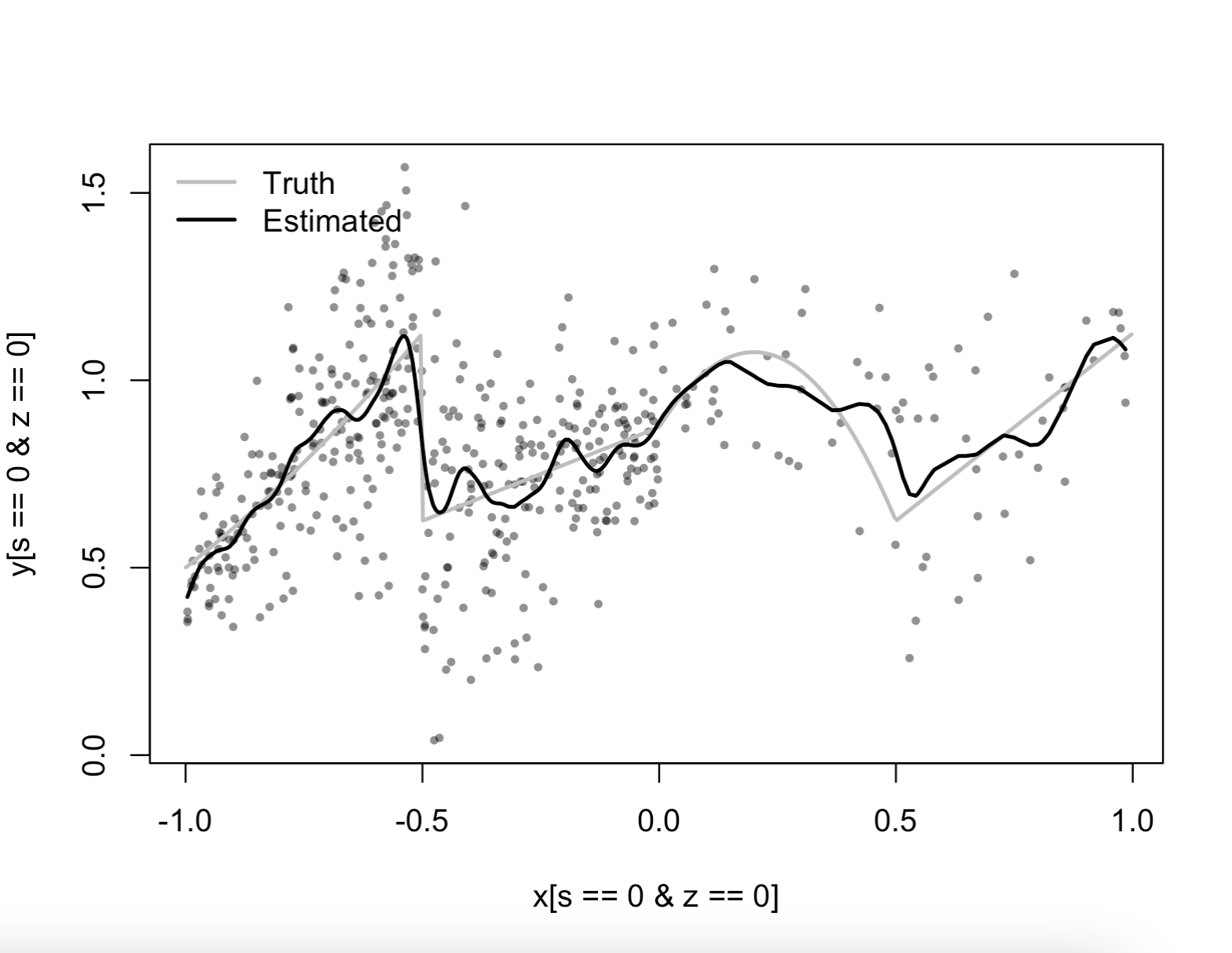}
        \caption{$Z=0$ and $S =0$ and the estimated function $\mu_{00}$}
        \label{fig:mu00}
    \end{subfigure}
    \hfill
    \begin{subfigure}{0.45\textwidth}
        \centering
        \includegraphics[width=\textwidth]{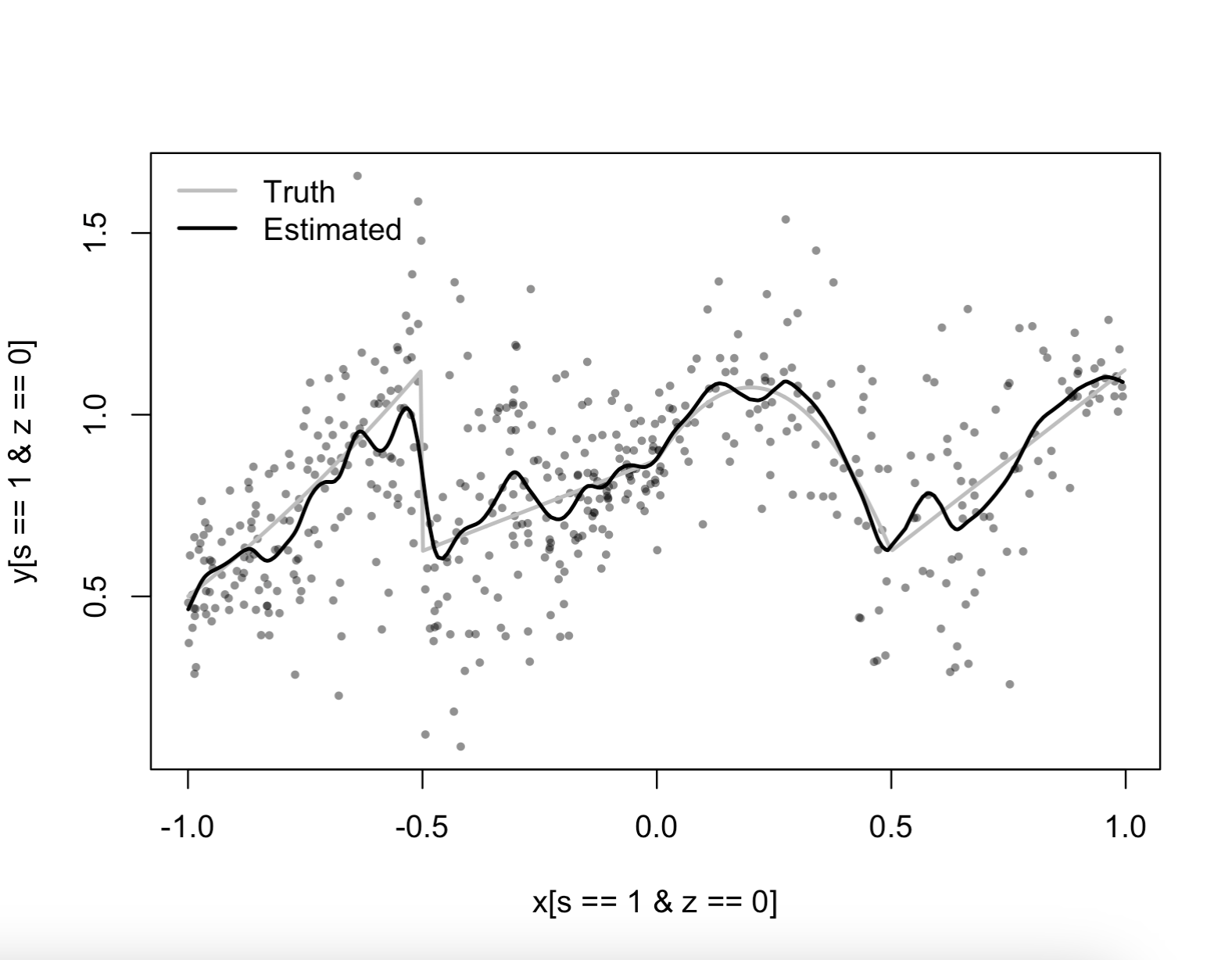}
        \caption{$Z=0$ and $S =1$ and the estimated function $\mu_{01}$}
        \label{fig:mu01}
    \end{subfigure}
    
    \vspace{1em}
    
    \begin{subfigure}{0.45\textwidth}
        \centering
        \includegraphics[width=\textwidth]{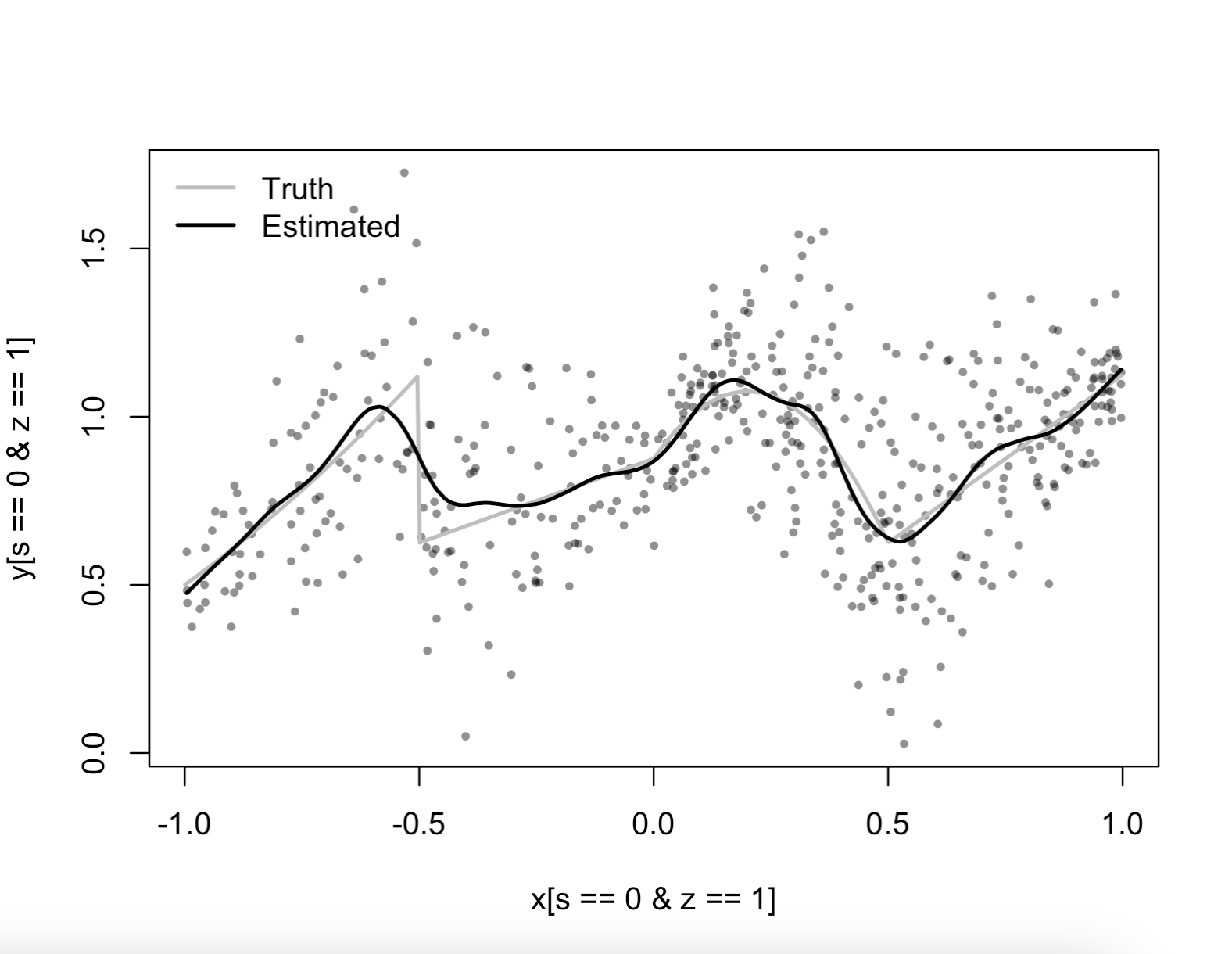}
        \caption{$Z=1$ and $S =0$ and the estimated function $\mu_{10}$}
        \label{fig:mu10}
    \end{subfigure}
    \hfill
    \begin{subfigure}{0.45\textwidth}
        \centering
        \includegraphics[width=\textwidth]{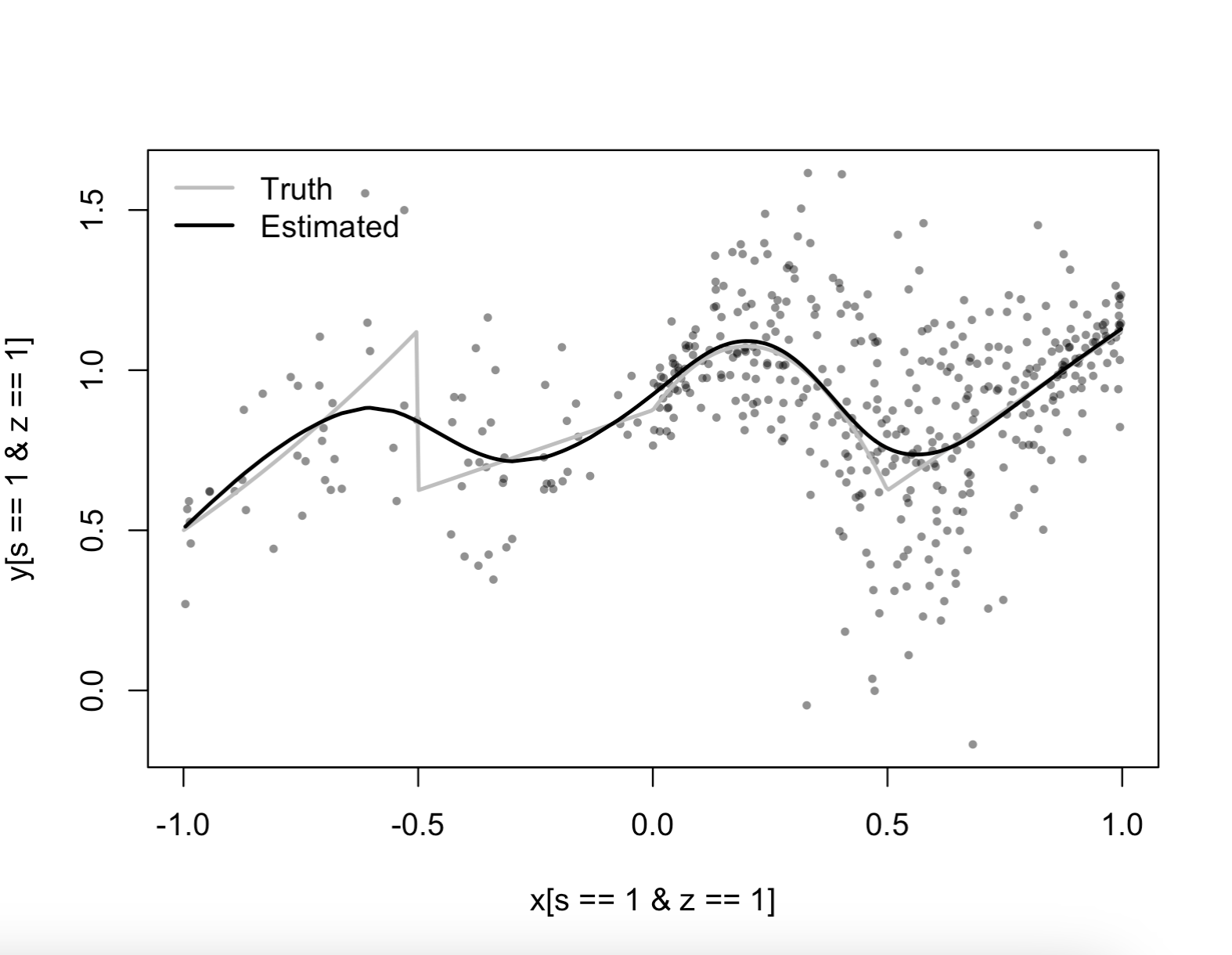}
        \caption{$Z=1$ and $S =1$ and the estimated function $\mu_{11}$}
        \label{fig:mu11}
    \end{subfigure}
    \caption{Plot of simulated subset data of simple example and estimated functions of outcome models}
    \label{fig:example1}
\end{figure}
Figure \ref{fig:example1} shows $n = 2000$ simulated data points, with roughly one-quarter in each subgroup: $Z=0, S=0$ (top left), $Z=0, S=1$ (top right), $Z=1, S=0$ (bottom left), and $Z=1, S=1$ (bottom right). The figure also plots the estimated functions $\hat{\mu}_{00}$, $\hat{\mu}_{01}$, $\hat{\mu}_{10}$, and $\hat{\mu}_{11}$, obtained using the default settings of \texttt{smoothing.spline()} in R.

Based on the data-generating process, control observations ($Z=0$) lie mainly on the left of the distribution, while treated observations ($Z=1$) lie mostly on the right. Within the control group, the $S=0$ subgroup is more concentrated on the left than $S=1$, and within the treated group, the $S=0$ subgroup is less concentrated on the right than $S=1$. As a result, $\hat{\mu}_{10}$ and $\hat{\mu}_{11}$ are over-smoothed in regions with sparse data on the left, while $\hat{\mu}_{00}$ and $\hat{\mu}_{01}$ are under-smoothed in denser regions on the right. In particular, $\hat{\mu}_{00}$ appears sharper than $\hat{\mu}_{01}$ on the left, and $\hat{\mu}_{10}$ appears smoother than $\hat{\mu}_{11}$ on the right.

Hence, the T-learner can be a poor and unnecessarily complex estimator of the true difference, which in this case is constant and equal to zero, and it also suffers from the limitations discussed above. In the next section, we present alternative estimators that improve upon the T-learner in several respects and avoid these issues.

\begin{figure}[]
    \centering
    \begin{subfigure}{0.3\textwidth}
        \centering
        \includegraphics[width=\textwidth]{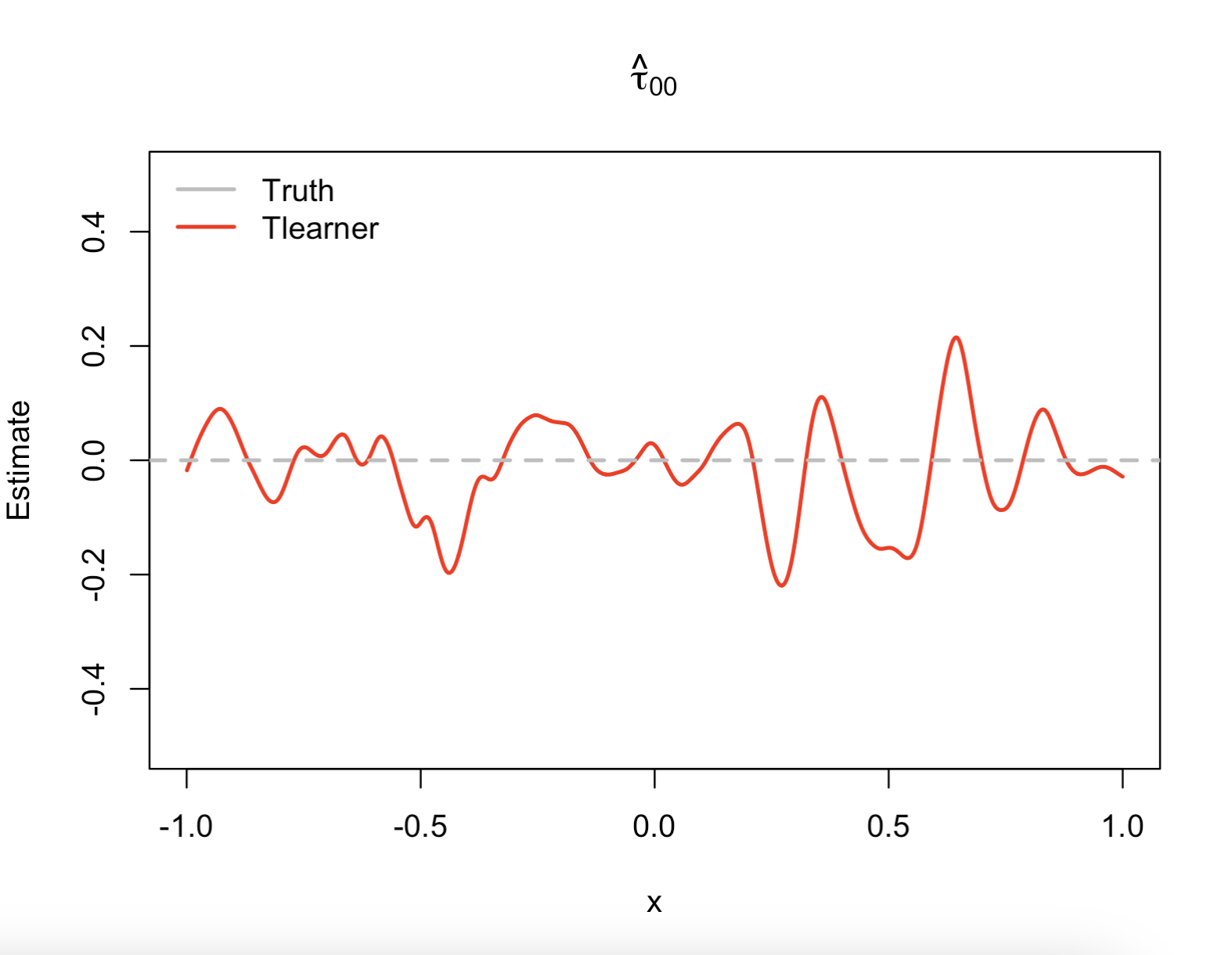}
        \caption{Estimated treatment effect $\hat{\tau}_{00}$}
        \label{fig:left}
    \end{subfigure}
    \hfill
    \begin{subfigure}{0.3\textwidth}
        \centering
        \includegraphics[width=\textwidth]{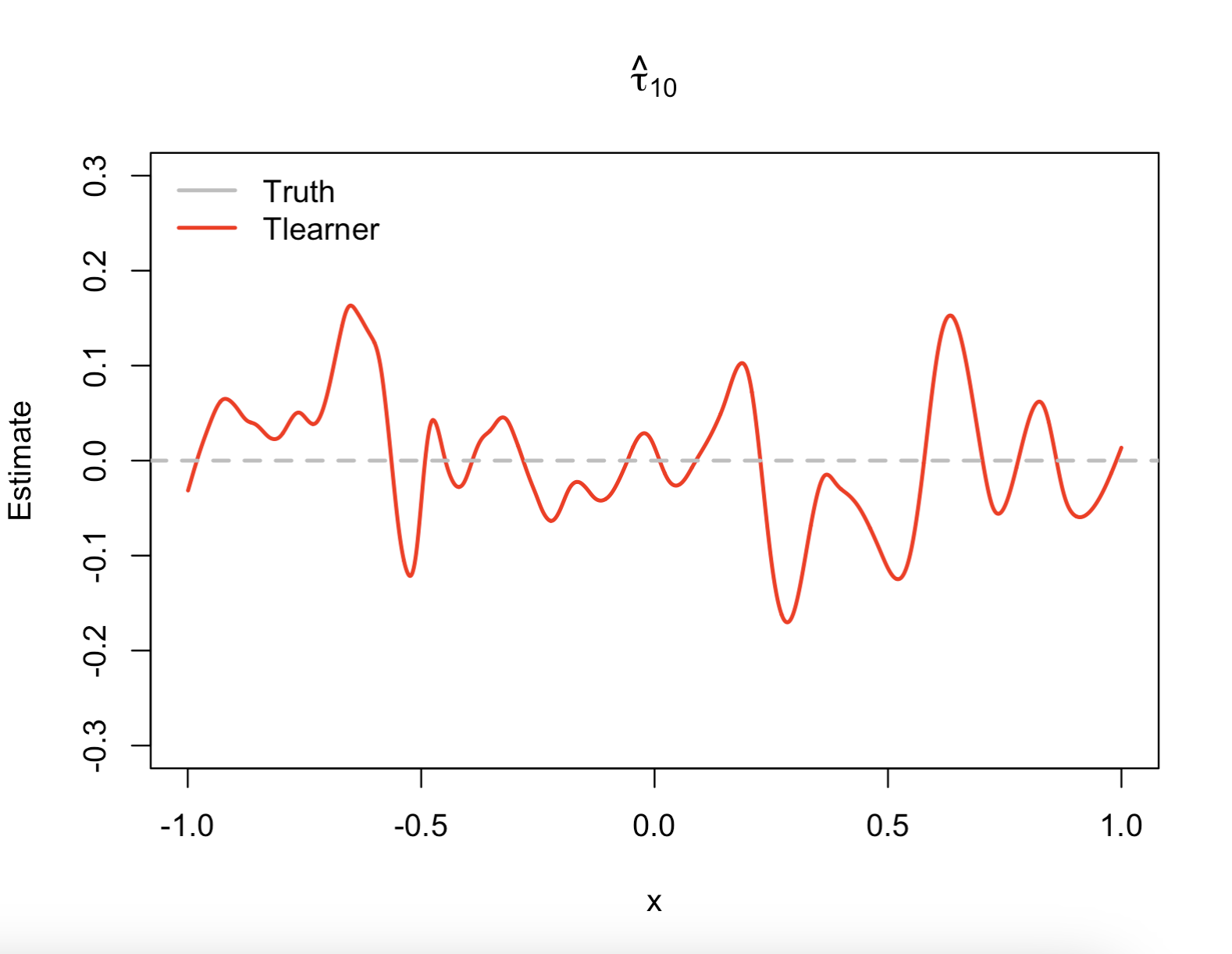}
        \caption{Estimated treatment effect $\hat{\tau}_{10}$}
        \label{fig:middle}
    \end{subfigure}
    \hfill
    \begin{subfigure}{0.3\textwidth}
        \centering
        \includegraphics[width=\textwidth]{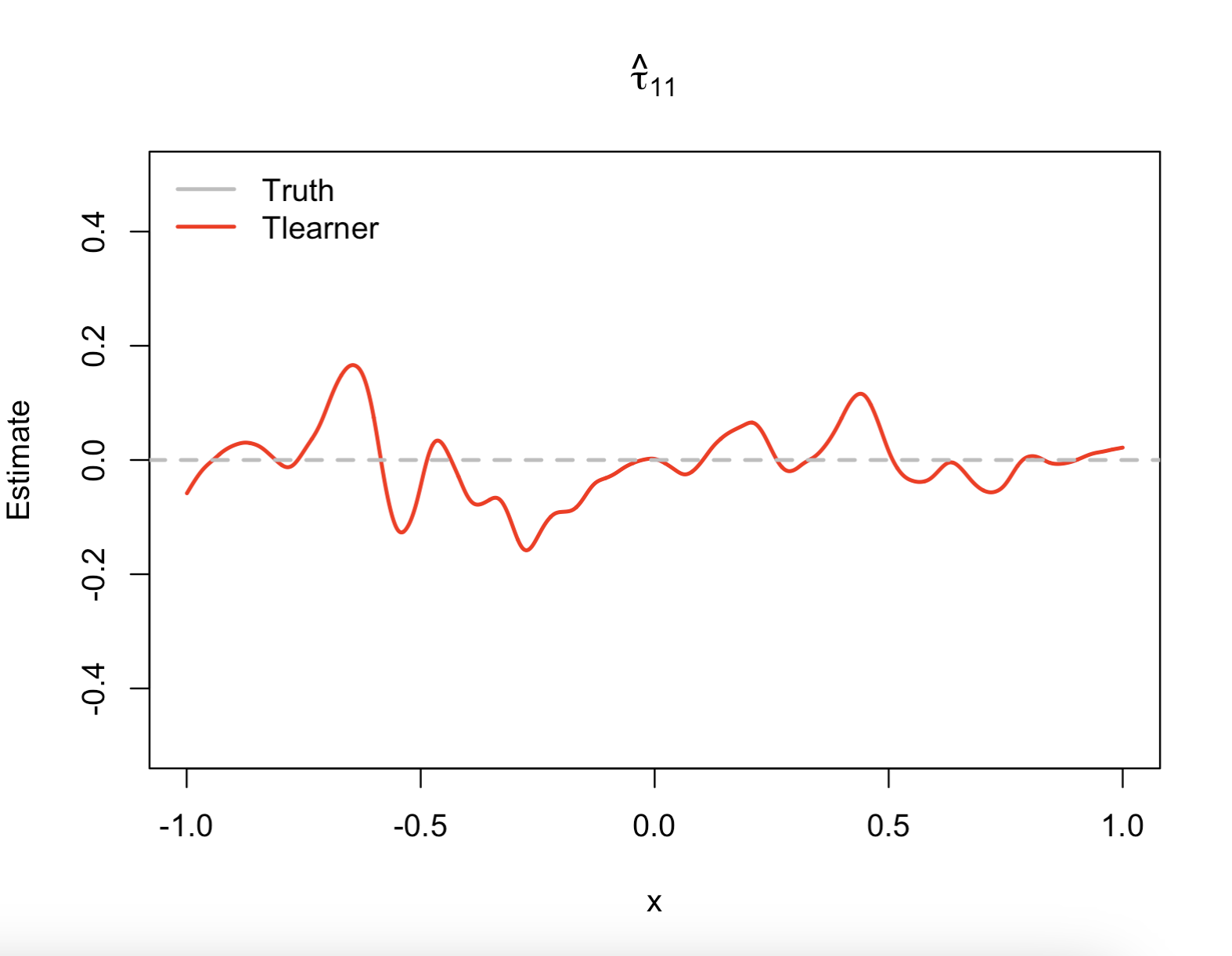}
        \caption{Estimated treatment effect $\hat{\tau}_{11}$}
        \label{fig:right}
    \end{subfigure}
    \caption{Estimated CPCE and the true CPCE in the simple example}
    \label{fig:example2}
\end{figure}

\subsection{Proof of Theorem \ref{the1}}\label{s:proofth2.1}
Theorem \ref{the1} shows that the target function $\tau^{u}(x)$ can be expressed as a difference of observed outcome mean models. Under the monotonicity assumption (Assumption~\ref{assump:monotonicity}) and the principal ignorability assumption (Assumption~\ref{assump:ps}), we obtain a useful Lemma \ref{le0} that links the potential outcomes within each principal stratum to the corresponding observed outcome cells.

\begin{lemma}\label{le0}
Suppose Assumptions~\ref{assump:monotonicity} (monotonicity) and \ref{assump:ps} (principal ignorability) hold. Then, for all $x$, each observed $(Z,S)$ cell is associated with a unique principal stratum or a pair of strata that share the same conditional mean potential outcomes. Specifically, we have
\begin{align*}
\mathbb{E}\{Y(1)\mid Z=1,S=1,X=x\}
&= \mathbb{E}\{Y(1)\mid U\in\{10,11\},X=x\},\\[4pt]
\mathbb{E}\{Y(0)\mid Z=0,S=0,X=x\}
&= \mathbb{E}\{Y(0)\mid U\in\{00,10\},X=x\},\\[4pt]
\mathbb{E}\{Y(1)\mid Z=1,S=0,X=x\}
&= \mathbb{E}\{Y(1)\mid U=00,X=x\}, \text{ and} \\[4pt]
\mathbb{E}\{Y(0)\mid Z=0,S=1,X=x\}
&= \mathbb{E}\{Y(0)\mid U=11,X=x\}
\end{align*}
for all $x$.
\end{lemma}
Given Lemma~\ref{le0}, the proof of Theorem~\ref{the1} is straightforward. 
As an illustration, consider $u=10$. Let
$\mu_{zs}(x) := \mathbb{E}\{Y\mid Z=z,S=s,X = x\}$. Then
\begin{align*}
\tau^{10}(x)
&= \mathbb{E}\{Y(1)\mid U=10,X =x\} - \mathbb{E}\{Y(0)\mid U=10,X =x\} \\
&= \mathbb{E}\{Y(1)\mid Z=1,S=1,X =x\} - \mathbb{E}\{Y(0)\mid Z=0,S=0,X =x\}\\
&= \mathbb{E}\{Y\mid Z=1,S=1,X =x\} - \mathbb{E}\{Y\mid Z=0,S=0,X =x\} \\
&= \mu_{11}(x) - \mu_{00}(x).
\end{align*}
The first equation is by definition, $\tau^{10}(x)
= \mathbb{E}\{Y(1)\mid U=10,X=x\} - \mathbb{E}\{Y(0)\mid U=10,X=x\}$. The second equation is by Lemma~\ref{le0}. The third equation is by the consistency assumption.

Now, let's show the proof of Lemma \ref{le0}.
\begin{proof}
We first show that the left part of Lemma~\ref{le0} can be written in the following form, for all $x$:
\begin{align*}
\mathbb{E}\{Y(1)\mid U\in\{10,11\},X=x\}
&= \mathbb{E}\{Y(1)\mid U\in\{10,11\},Z=1,S=1,X=x\},\\[4pt]
\mathbb{E}\{Y(0)\mid U\in\{00,10\},X=x\}
&= \mathbb{E}\{Y(0)\mid U\in\{00,10\},Z=0,S=0,X=x\},\\[4pt]
\mathbb{E}\{Y(1)\mid U=00,X=x\}
&= \mathbb{E}\{Y(1)\mid U=00,Z=1,S=0,X=x\}, \text{ and} \\[4pt]
\mathbb{E}\{Y(0)\mid U=11,X=x\}
&= \mathbb{E}\{Y(0)\mid U=11,Z=0,S=1,X=x\}.
\end{align*}
We prove one of these equalities as an example; the others follow analogously. 
For the complier stratum $U=10$, we rewrite the conditioning in terms of the potential compliance indicators $(S(1),S(0))=(1,0)$:
\[
\mathbb{E}\{Y(1)\mid U=10,X = x\}
= \mathbb{E}\{Y(1)\mid S(1)=1,S(0)=0,X = x\}.
\]
Under Assumption~\ref{assump:pi}, we can further write
\[
\mathbb{E}\{Y(1)\mid S(1)=1,S(0)=0,X = x\}
= \mathbb{E}\{Y(1)\mid S(1)=1,S(0)=0,Z=1,X = x\},
\]
which, in terms of the principal stratum $U=10$, is equivalently
\[
\mathbb{E}\{Y(1)\mid S(1)=1,S(0)=0,Z=1,X = x\}
= \mathbb{E}\{Y(1)\mid U=10,S=1,Z=1,X = x\}.
\]
So we have $\mathbb{E}\{Y(1)\mid U=10,X = x\} = \mathbb{E}\{Y(1)\mid U=10,S=1,Z=1,X = x\}.$

We now turn to these quantities in more detail. Under monotonicity $S(1) \ge S(0)$, we have 
\begin{equation*}
    \begin{array}{c|cc}
 & S=0 & S=1 \\
\hline
Z=0 & U\in\{00,10\} & U=11 \\
Z=1 & U=00 & U\in\{11,10\} \\
\end{array}
\end{equation*}
From Table, we see that the observed cells $(Z,S) = (1,0)$ and $(Z,S) = (0,1)$ each contain only a single principal stratum, namely $U=00$ and $U=11$, respectively. Thus
$\mathbb{P}(U=00 \mid Z=1,S=0,X = x) = 1,
\qquad
\mathbb{P}(U=11 \mid Z=0,S=1,X = x) = 1,$ and hence
\begin{align*}
\mathbb{E}\{Y(1)\mid Z=1,S=0,X = x\}
&= \mathbb{E}\{Y(1)\mid U=00,Z=1,S=0,X = x\} \text{ and} \\[4pt]
\mathbb{E}\{Y(0)\mid Z=0,S=1,X =x\}
&= \mathbb{E}\{Y(0)\mid U=11,Z=0,S=1,X =x\}.
\end{align*}
The remaining observed cells are mixtures: $\{U=10,11\}$ when $(Z,S)=(1,1)$ and $\{U=00,10\}$ when $(Z,S)=(0,0)$. Thus
\begin{align*}
\mathbb{P}(U=10 \mid Z=1,S=1,X =x) + \mathbb{P}(U=11 \mid Z=1,S=1,X =x) &= 1 \text{ and} \\
\mathbb{P}(U=00 \mid Z=0,S=0,X=x) + \mathbb{P}(U=10 \mid Z=0,S=0,X =x) &= 1.
\end{align*}
We show the first identity; the second follows analogously. Because $(Z,S)=(1,1)$ is a mixture of $U=10$ and $U=11$, we can write
\begin{align*}
&\mathbb{E}\!\bigl[ Y(1) \mid Z=1, S=1, X=x\bigr]\\
&= \mathbb{E}\!\bigl[ Y(1) \mid U=11, Z=1, S=1, X=x \bigr]\,
   \mathbb{P}\!\bigl( U=11 \mid Z=1, S=1, X=x\bigr) \\
&\quad+ \mathbb{E}\!\bigl[ Y(1) \mid U=10, Z=1, S=1, X=x \bigr]\,
   \mathbb{P}\!\bigl( U=10 \mid Z=1, S=1, X=x\bigr)
\end{align*}
Under Assumption~\ref{assump:ps}, the mean potential outcomes are equal across the relevant strata, i.e.
\[
\mathbb{E}\{Y(1)\mid U=11,Z=1,S=1,X =x\}
=
\mathbb{E}\{Y(1)\mid U=10,Z=1,S=1,X =x\}.
\]
Therefore,
\begin{align*}
\mathbb{E}\{Y(1)\mid Z=1,S=1,X =x\}
&=\mathbb{E}\{Y(1)\mid U=11,Z=1,S=1,X =x\} \\
&=\mathbb{E}\{Y(1)\mid U=10,Z=1,S=1,X =x\}.
\end{align*}

\noindent Similarly, for $(Z,S)=(0,0)$, Assumption~\ref{assump:ps} implies
\[
\mathbb{E}\{Y(0)\mid U=00,Z=0,S=0,X =x\}
=
\mathbb{E}\{Y(0)\mid U=10,Z=0,S=0,X= x\},
\]
and hence
\begin{align*}
\mathbb{E}\{Y(0)\mid Z=0,S=0,X =x\}
&=\mathbb{E}\{Y(0)\mid U=00,Z=0,S=0,X =x\} \\
&=\mathbb{E}\{Y(0)\mid U=10,Z=0,S=0,X =x\}.
\end{align*}

\end{proof}

\subsection{Proof of Theorem \ref{the:subset}}\label{s:proofth3.1}
Theorem~\ref{the:subset} shows that our target can be identified from a subset perspective. 
We use the complier stratum $U=10$ as an illustrative example; the arguments for $U=00$ and $U=11$
are analogous. 

For $U=10$, define the pseudo-outcome
\[
\varphi_{\tau^{10}}(W)
= \frac{Z-\pi_{\mathcal{S}_{10}}(X)}{\pi_{\mathcal{S}_{10}}(X)\{1-\pi_{\mathcal{S}_{10}}(X)\}}
  \bigl\{Y-\mu_{ZZ}(X)\bigr\}
  + \mu_{11}(X)-\mu_{00}(X),
\]
where $\pi_{\mathcal{S}_{10}}(x)$ is the subset propensity score and
$\mu_{zz}(x)=\mathbb{E}\{Y\mid Z=z,S=z,X =x\}$ for $z\in\{0,1\}$.
We need to show that
\[
\mathbb{E}\bigl\{\varphi_{\tau^{10}}(W) \mid \mathcal{S}_{10}, X =x\bigr\}
= \tau^{10}(x).
\]

\noindent On $\mathcal{S}_{10}$, we have $\mathbb{P}(Z=1\mid \mathcal{S}_{10},X =x)=\pi_{\mathcal{S}_{10}}(x)$, so by the law of total expectation,
\begin{align*}
\mathbb{E}\{\varphi_{\tau^{10}}(W)\mid \mathcal{S}_{10},X =x\}
&= \mathbb{E}\{\varphi_{\tau^{10}}(W)\mid Z=1,S=1,X =x\}\,\pi_{\mathcal{S}_{10}}(x) \\
&\quad + \mathbb{E}\{\varphi_{\tau^{10}}(W)\mid Z=0,S=0,X =x\}\,\{1-\pi_{\mathcal{S}_{10}}(x)\}.
\end{align*}
Substituting the definition of $\varphi_{\tau^{10}}$ gives
\begin{align*}
\mathbb{E}\{\varphi_{\tau^{10}}(W)\mid \mathcal{S}_{10},X =x\}
&= \mathbb{E}\!\left[
Y-\mu_{11}(X)
\;\Big|\; Z=1,S=1,X =x
\right] \\
&\quad-\mathbb{E}\!\left[
Y-\mu_{00}(X)
\;\Big|\; Z=0,S=0,X =x
\right]
+\mu_{11}(x)-\mu_{00}(x).
\end{align*}
Using 
$\mathbb{E}\{Y-\mu_{11}(X)\mid Z=1,S=1,X =x\}=0$ and 
$\mathbb{E}\{Y-\mu_{00}(X)\mid Z=0,S=0,X =x\}=0$, it follows that
\[
\mathbb{E}\bigl\{\varphi_{\tau^{10}}(W) \mid \mathcal{S}_{10},X =x\bigr\}
= \mu_{11}(x)-\mu_{00}(x)
= \tau^{10}(x).
\]

\subsection{Another way to express $\pi_{\mathcal{S}_{u}}(x)$}\label{s:subpropensity}
There is a simple relationship between the subset propensity scores and the usual
propensity and principal scores. In particular, each subset propensity score can be
expressed as a function of the overall propensity score
\(\pi(x) = \mathbb{P}(Z=1 \mid X =x)\) and the principal scores
\(p_z(x) = \mathbb{P}(S=1 \mid Z=z,X =x)\), \(z \in \{0,1\}\).
For the subsets \(\mathcal{S}_{10} = \{(Z,S)=(1,1),(0,0)\}\),
\(\mathcal{S}_{00} = \{(Z,S)=(1,0),(0,0)\}\), and
\(\mathcal{S}_{11} = \{(Z,S)=(1,1),(0,1)\}\), we have
\begin{align*}
\pi_{\mathcal{S}_{10}}(x)
&:= \mathbb{P}\bigl(Z=1 \mid \mathcal{S}_{10},X =x\bigr)
 = \frac{\pi(x)\,p_1(x)}
        {\pi(x)\,p_1(x) + \{1-\pi(x)\}\{1-p_0(x)\}},\\[4pt]
\pi_{\mathcal{S}_{00}}(x)
&:= \mathbb{P}\bigl(Z=1 \mid \mathcal{S}_{00},X =x\bigr)
 = \frac{\{1-p_1(x)\}\,\pi(x)}
        {\pi(x)\{1-p_1(x)\} + \{1-p_0(x)\}\{1-\pi(x)\}},\\[4pt]
\pi_{\mathcal{S}_{11}}(x)
&:= \mathbb{P}\bigl(Z=1 \mid \mathcal{S}_{11},X =x\bigr)
 = \frac{p_1(x)\,\pi(x)}
        {\pi(x)\,p_1(x) + p_0(x)\{1-\pi(x)\}}.
\end{align*}

Here we use $\pi_{\mathcal{S}_{10}}(x)$ as an example; the derivations for the other two subset propensity scores are analogous. Since 
\(\mathcal{S}_{10} = \{(Z,S)=(1,1),(0,0)\}\), we have
\[
\pi_{\mathcal{S}_{10}}(x)
= \mathbb{P}\bigl(Z=1 \mid \mathcal{S}_{10},X =x\bigr)
= \frac{\mathbb{P}(Z=1,S=1 \mid X =x)}
       {\mathbb{P}(Z=1,S=1 \mid X =x) + \mathbb{P}(Z=0,S=0 \mid X =x)}.
\]
By the definition of conditional probability, $\mathbb{P}\bigl(Z=1 \mid \mathcal{S}_{10},X =x\bigr) = \mathbb{P}(S=1 \mid Z=1,X=x)\,\mathbb{P}(Z=1 \mid X=x)$. In terms of the propensity score $\pi(x)=\mathbb{P}(Z=1\mid X =x)$ and the principal scores
$p_1(x)=\mathbb{P}(S=1\mid Z=1,X =x)$ and $p_0(x)=\mathbb{P}(S=1\mid Z=0,X =x)$, it can be written as
\[
\pi_{\mathcal{S}_{10}}(x)
= \frac{p_1(x)\,\pi(x)}
       {\pi(x)p_1(x) + \{1-\pi(x)\}\{1-p_0(x)\}}.
\]

\subsection{Proof of Theorem \ref{robustness_subset}}\label{s:proofrobsub}

Theorem~\ref{robustness_subset} establishes weak robustness and rate robustness for the subset plug-in limit parameter. We begin by proving the weak robustness property.

\begin{enumerate}
    \item \textbf{Weak robustness:} For the weak robustness property, we prove the result for $\tilde{\tau}_{\mathrm{sub}}^{10}(x)$ as an illustrative example; the arguments for the other two strata are analogous. Double robustness means that $\tilde{\tau}_{\mathrm{sub}}^{10}(x)$ is unbiased if either the outcome mean model or the subset propensity score model is correctly specified. 
Under a correctly specified outcome mean model, the conditional expectation of the
pseudo-outcome $\tilde\varphi_{\tau^{10}}$ given $(Z,S,X)$ reduces to the target
contrast. Specifically,
\begin{align*}
&\mathbb{E}\!\left[\tilde\varphi_{\tau^{10}}(W) \mid Z=S,X =x\right]\\
&= \mathbb{E}\!\left[
\frac{Z-\tilde\pi_{\mathcal{S}_{10}}(X)}
     {\tilde\pi_{\mathcal{S}_{10}}(X)\{1-\tilde\pi_{\mathcal{S}_{10}}(X)\}}
\{Y-\tilde\mu_{ZZ}(X)\}
\;\Bigm|\; Z=S,X = x
\right] + \tilde\mu_{11}(x)-\tilde\mu_{00}(x) \\
&= \frac{\pi_{\mathcal{S}_{10}}(x)}{\tilde\pi_{\mathcal{S}_{10}}(x)}
   \,\mathbb{E}\!\left[\,Y-\tilde\mu_{11}(X)\;\Bigm|\; Z=1,S=1,X =x\right] \\
&\quad + \frac{1-\pi_{\mathcal{S}_{10}}(x)}{1-\tilde\pi_{\mathcal{S}_{10}}(x)}
   \,\mathbb{E}\!\left[\,Y-\tilde\mu_{00}(X)\;\Bigm|\; Z=0,S=0,X =x\right] + \tilde\mu_{11}(x)-\tilde\mu_{00}(x).
\end{align*}
The first equality holds by definition of $\tilde\varphi_{\tau^{10}}(W)$. 
The second equality follows by conditioning on $(Z,S)$ under the event $\{Z=S\}$ and using that, on $\{Z=S\}$,
$(Z,S)\in\{(1,1),(0,0)\}$. Specifically, we apply iterated expectation and split by the two cases.
If the outcome model is correctly specified, then
$\tilde\mu_{11}(x)=\mu_{11}(x)$ and $\tilde\mu_{00}(x)=\mu_{00}(x)$, so
$\mathbb{E}[Y-\tilde\mu_{11}(X)\mid Z=1,S=1,X =x]=0$ and
$\mathbb{E}[Y-\tilde\mu_{00}(X)\mid Z=0,S=0,X =x]=0$.
Hence the first two terms vanish and
\[
\mathbb{E}\bigl[\tilde\varphi_{\tau^{10}} \mid Z=S,X =x\bigr]
= \mu_{11}(x)-\mu_{00}(x)
= \tau_{10}(x),
\]
showing that $\tilde\tau^{10}_{\mathrm{sub}}(x)$ is unbiased even if
$\tilde\pi_{\mathcal{S}_{10}}(x)$ is misspecified.

Conversely, suppose the subset propensity score is correctly specified,
$\tilde\pi_{\mathcal{S}_{10}}(x)=\pi_{\mathcal{S}_{10}}(x)$, while the outcome
model $\tilde\mu_{zz}(x)$ may be misspecified. Using the same decomposition, we can
rewrite
\begin{equation*}
\begin{split}
\mathbb{E}\bigl[\tilde\varphi_{\tau^{10}} \mid Z=S,X =x\bigr]
&= 
   \mathbb{E}[Y\mid Z=1,S=1,X =x]-\tilde\mu_{11}(x)\\
&\quad - 
   \mathbb{E}[Y\mid Z=0,S=0,X =x]-\tilde\mu_{00}(x)
   + \tilde\mu_{11}(x)-\tilde\mu_{00}(x).
\end{split}
\end{equation*}
A simple algebraic rearrangement shows that the terms involving
$\tilde\mu_{11}(x)$ and $\tilde\mu_{00}(x)$ cancel, yielding
\[
\mathbb{E}\bigl[\tilde\varphi_{\tau^{10}} \mid Z=S,X =x\bigr]
= \mu_{11}(x)-\mu_{00}(x)
= \tau_{10}(x),
\]
so $\tilde\tau^{10}_{\mathrm{sub}}(x)$ remains unbiased even when the outcome
model is misspecified, provided $\tilde\pi_{\mathcal{S}_{10}}(x)$ is correct.
Together, these two cases establish the weak double robustness of the subset plug-in limit parameter.

\item  \textbf{Rate robustness:} 
The detailed bias bounds for the subset plug-in limit parameters are
\begin{align*}
\tilde\tau^{00}_{\mathrm{sub}}(x)-\tau^{00}(x)
&= O_p\!\left(
    \bigl|\tilde\pi_{\mathcal{S}_{00}}(x)-\pi_{\mathcal{S}_{00}}(x)\bigr|\,
\max_{z \in \{0,1\}}|\tilde\mu_{z0}(x)-\mu_{z0}(x)|
  \right),\\[4pt]
\tilde\tau^{10}_{\mathrm{sub}}(x)-\tau^{10}(x)
&= O_p\!\left(
    \bigl|\tilde\pi_{\mathcal{S}_{10}}(x)-\pi_{\mathcal{S}_{10}}(x)\bigr|\,
\max_{z \in \{0,1\}}|\tilde\mu_{zz}(x)-\mu_{zz}(x)|
  \right), \text{ and} \\[4pt]
\tilde\tau^{11}_{\mathrm{sub}}(x)-\tau^{11}(x)
&= O_p\!\left(
    \bigl|\tilde\pi_{\mathcal{S}_{11}}(x)-\pi_{\mathcal{S}_{11}}(x)\bigr|\,
\max_{z \in \{0,1\}}|\tilde\mu_{z1}(x)-\mu_{z1}(x)|
  \right).
\end{align*}
We use $\tilde\pi_{\mathcal{S}_{10}}$ as an illustrative example; the arguments for the other two biases are analogous.
We use $\tilde\pi_{\mathcal{S}_{10}}$ as an illustrative example; the arguments for
$\tilde\pi_{\mathcal{S}_{00}}$ and $\tilde\pi_{\mathcal{S}_{11}}$ are analogous.
Conditioning on $(X,Z=S)$, we have
\begin{align*}
\mathbb{E}\!\left[\tilde\varphi_{\tau^{10}}\mid Z=S, X = x\right]
&= \tilde\mu_{11}(x)-\tilde\mu_{00}(x) + \frac{\pi_{\mathcal{S}_{10}}(x)}{\tilde\pi_{\mathcal{S}_{10}}(x)}
\bigl(\mu_{11}(x)-\tilde\mu_{11}(x)\bigr) \\
&\quad - \frac{1-\pi_{\mathcal{S}_{10}}(x)}{1-\tilde\pi_{\mathcal{S}_{10}}(x)}
\bigl(\mu_{00}(x)-\tilde\mu_{00}(x)\bigr).
\end{align*}

Subtracting $\mu_{11}(x)-\mu_{00}(x)$ yields the exact identity
\begin{equation}\label{eq:bias_sub10}
    \begin{aligned}
    \tilde{\tau}_{\mathrm{sub}}^{10}(x) - \tau^{10}(x)
    &= \frac{1}{\tilde\pi_{\mathcal{S}_{10}}(x)}
       \bigl(\tilde{\pi}_{\mathcal{S}_{10}}(x)-\pi_{\mathcal{S}_{10}}(x)\bigr)
       \bigl(\tilde{\mu}_{11}(x)-\mu_{11}(x)\bigr)\\
    &\quad+ \frac{1}{1-\tilde\pi_{\mathcal{S}_{10}}(x)}
       \bigl(\tilde{\pi}_{\mathcal{S}_{10}}(x)-\pi_{\mathcal{S}_{10}}(x)\bigr)
       \bigl(\tilde{\mu}_{00}(x)-\mu_{00}(x)\bigr).
    \end{aligned}
\end{equation}

From \eqref{eq:bias_sub10} and the overlap assumption
$\tilde\pi_{\mathcal{S}_{10}}(x),\,1-\tilde\pi_{\mathcal{S}_{10}}(x)\ge \epsilon$, we obtain
\begin{align*}
\bigl|\tilde{\tau}_{\mathrm{sub}}^{10}(x) - \tau^{10}(x)\bigr|
&\le \frac{1}{\epsilon}\,
\bigl|\tilde{\pi}_{\mathcal{S}_{10}}(x)-\pi_{\mathcal{S}_{10}}(x)\bigr| \\
&\quad\times
\Bigl(
  \bigl|\tilde{\mu}_{11}(x)-\mu_{11}(x)\bigr|
  + \bigl|\tilde{\mu}_{00}(x)-\mu_{00}(x)\bigr|
\Bigr)
\qquad \text{a.s.}
\end{align*}

Then 
$$
\tilde\tau^{00}_{\mathrm{sub}}(x)-\tau^{00}(x)= O_p\!\left(
    \bigl|\tilde\pi_{\mathcal{S}_{00}}(x)-\pi_{\mathcal{S}_{00}}(x)\bigr|\,
\max_{z \in \{0,1\}}|\tilde\mu_{z0}(x)-\mu_{z0}(x)|
  \right)
$$
The other bias terms can be calculated similarly:
\begin{equation}\label{eq:bias_sub00}
    \begin{aligned}
    \tilde{\tau}_{\mathrm{sub}}^{00}(x) - \tau^{00}(x)
    &= \frac{1}{\tilde\pi_{\mathcal{S}_{00}}(x)}
       \bigl(\tilde{\pi}_{\mathcal{S}_{00}}(x)-\pi_{\mathcal{S}_{00}}(x)\bigr)
       \bigl(\tilde{\mu}_{10}(x)-\mu_{10}(x)\bigr)\\
    &\quad+ \frac{1}{1-\tilde\pi_{\mathcal{S}_{00}}(x)}
       \bigl(\tilde{\pi}_{\mathcal{S}_{00}}(x)-\pi_{\mathcal{S}_{00}}(x)\bigr)
       \bigl(\tilde{\mu}_{00}(x)-\mu_{00}(x)\bigr).
    \end{aligned}
\end{equation}
\begin{equation}\label{eq:bias_sub11}
    \begin{aligned}
    \tilde{\tau}_{\mathrm{sub}}^{11}(x) - \tau^{11}(x)
    &= \frac{1}{\tilde\pi_{\mathcal{S}_{11}}(x)}
       \bigl(\tilde{\pi}_{\mathcal{S}_{11}}(x)-\pi_{\mathcal{S}_{11}}(x)\bigr)
       \bigl(\tilde{\mu}_{11}(x)-\mu_{11}(x)\bigr)\\
    &\quad+ \frac{1}{1-\tilde\pi_{\mathcal{S}_{11}}(x)}
       \bigl(\tilde{\pi}_{\mathcal{S}_{11}}(x)-\pi_{\mathcal{S}_{11}}(x)\bigr)
       \bigl(\tilde{\mu}_{01}(x)-\mu_{01}(x)\bigr).
    \end{aligned}
\end{equation}

\subsection{Proof of Theorem \ref{the:eif}}\label{s:proof5}
Theorem~\ref{the:eif} shows that the target $\tau^{10}(x)$ can be identified from the corresponding component of the efficient influence function for the principal causal effect (PCE).

We first present the identification proof for the CPCEs. We prove the result for $\tau^{10}(x)$ as an illustrative example; the arguments for the other two strata are analogous. We then derive the corresponding identities for the key component quantities. For convenience, we use $\psi_{a, f(Z,S,Y)} \equiv \psi_{a, f(Z,S,Y)}(W)$.

From \cite{jiang2022multiply}, we know $e^{10}(x) = p_1(x) - p_0(x).$
The EIF components are
\[
\phi_{1,10}(W)
= \frac{e^{10}(X)}{p_1(X)} \psi_{1,YS}
  -\mu_{11}(X)\biggl\{\psi_{0,S}
    -\frac{p_0(X)}{p_1(X)} \psi_{1,S}\biggr\} \text{ and} 
\]
\[
\phi_{0,10}(W)
= \frac{e^{10}(X)}{1-p_0(X)} \psi_{0,Y(1-S)}
  -\mu_{00}(X)\biggl\{\psi_{1,1-S}
    -\frac{1-p_1(X)}{1-p_0(X)} \psi_{0,1-S}\biggr\}.
\]
The corresponding augmented
scores $\psi_{1,S}$, $\psi_{0,S}$, $\psi_{1,1-S}$, $\psi_{0,1-S}$, $\psi_{1, YS}$, $\psi_{0, YS}$, and
$\psi_{1, Y(1-S)}$, $\psi_{0, Y(1-S)}$ have conditional expectations
\begin{equation}\label{eq:den_id}
    \begin{aligned}
        \mathbb{E}[\psi_{1,S}\mid X = x]=p_1(x),& \quad \mathbb{E}[\psi_{0,S}\mid X = x]=p_0(x),\\
         \mathbb{E}[\psi_{1,1-S}\mid X = x]=1-p_1(x), & \quad
\mathbb{E}[\psi_{0,1-S}\mid X = x]=1-p_0(x),\\
\mathbb{E}[\psi_{1, YS}\mid X = x]=\mu_{11}(x)p_1(x),&\quad \mathbb{E}[\psi_{0, YS}\mid X = x]=\mu_{01}(x)p_0(x),\\
\mathbb{E}[\psi_{1, Y(1-S)}\mid X = x]=\mu_{10}(x)\{1-p_1(x)\}, & \quad \mathbb{E}[\psi_{0, Y(1-S)}\mid X = x]=\mu_{00}(x)\{1-p_0(x)\}.
    \end{aligned}
\end{equation}
These identities express the principal scores $p_z(X)$, as well as products involving the outcome means and principal scores, as conditional expectations of components of the EIF.
Using the identities \eqref{eq:den_id}, the expectation is 
\begin{equation*}
    \begin{split}
        \mathbb{E}[\phi_{1,10}(W)\mid X = x]
        &= \frac{e^{10}(x)}{p_1(x)}\,\mathbb{E}[\psi_{1,YS}\mid X = x]\\
         & -\mu_{11}(x)\biggl\{\mathbb{E}[\psi_{0,S}\mid X = x]
           -\frac{p_0(x)}{p_1(x)} \mathbb{E}[\psi_{1,S}\mid X = x]\biggr\}\\
        &= e^{10}(x) \mu_{11}(x).
    \end{split}
\end{equation*}
and
\begin{align*}
\mathbb{E}\!\left[\phi_{0,10}(W)\mid X = x\right]
&= \frac{e^{10}(x)}{1-p_0(x)}\,\mathbb{E}\!\left[\psi_{0,Y(1-S)}\mid X = x\right] \\
&\quad - \mu_{00}(x)\Biggl(
\mathbb{E}\!\left[\psi_{1,1-S}\mid X = x\right]
- \frac{1-p_1(x)}{1-p_0(x)}\,\mathbb{E}\!\left[\psi_{0,1-S}\mid X = x\right]
\Biggr) \\
&= e^{10}(x)\,\mu_{00}(x).
\end{align*}

Moreover,
\[
\mathbb{E}(\psi_{1,S}-\psi_{0,S}\mid X = x)
= p_1(x)-p_0(x)
= e^{10}(x).
\]
Hence
\[
\tau^{10}(X)
= \frac{\mathbb{E}\big[\phi_{1,10}-\phi_{0,10}\mid X = x\big]}
       {\mathbb{E}\big(\psi_{1,S}-\psi_{0,S}\mid X = x\big)},
\]
so $\tau^{10}(x)$ is identified as a ratio of EIF components.

\quad We now turn to the proofs for the component quantities. We present one representative cases as examples; the remaining cases follow by analogous arguments.

First, for $f(Y,S,X)=S$, we have $        \psi_{1,S}(w)
        =\frac{\mathbf{1}(Z=1)\{s-p_z(x)\}}{\mathbb{P}(Z=1 \mid X =x)}+p_z(x).$ After you take expectation, we have 
\begin{equation*}
    \begin{split}
        \mathbb{E}[\psi_{1,S}\mid X = x]
        &= \mathbb{E}\left[
             \frac{\mathbf{1}(Z=1)\{S-p_z(x)\}}
                  {\mathbb{P}(Z=1 \mid X =x)}
             \Bigm| X =x
           \right]+p_z(x)\\
        &= \frac{ \mathbb{P}(Z=1 \mid X =x)}{\mathbb{P}(Z=1 \mid X =x)}
           \mathbb{E}\bigl[S-p_z(x)\mid Z=1, X =x\bigr]\,
           + p_z(x)\\
        &= p_1(x).
    \end{split}
\end{equation*}
The second equation follows from the law of iterated expectations. 

In summary, Theorem~\ref{the:eif} shows that each principal effect
$\tau^{u}(x)$, $u\in\{00,10,11\}$, can be written as the ratio of conditional
expectations of suitable EIF components, constructed from the scores
$\psi_{a, f(Y, S,X)}$ defined above.
\end{enumerate}

\subsection{Proof of Theorem \ref{robustness_eif}}\label{s:proof6}
Theorem \ref{robustness_eif} establishes weak robustness and rate robustness for the EIF plug-in limit parameter. We begin by proving the weak robustness property.

\begin{enumerate}
    \item \textbf{Weak robustness} \\
For the weak robustness property, we prove the result for
$\tilde{\tau}_{\mathrm{eif}}^{00}(x)$ as an illustrative example; the arguments
for the other two strata are analogous. Double robustness here means that
$\tilde{\tau}_{\mathrm{eif}}^{00}(x)$ is unbiased if either the outcome mean
model is correctly specified, or the propensity score and principal score
models are correctly specified.

The EIF-based population functional is
\[
\tilde{\tau}^{00}_{\mathrm{eif}}(x)
= \mathbb{E}\!\left[
  \frac{\tilde{\phi}_{1,00}(W)-\tilde{\phi}_{0,00}(W)}
       {\mathbb{E}(1-\tilde{\psi}_{1,S}\mid X)}
  \,\Bigm|\,X =x
\right],
\]
where
\begin{align*}
\tilde{\phi}_{1,00}(W)
&= \tilde{\psi}_{1, Y(1-S)}(W),\\
\tilde{\phi}_{0,00}(W)
&= \frac{\tilde e^{00}(X)}{1-\tilde p_0(X)}\,\tilde{\psi}_{0, Y(1-S)}(W) + \tilde{\mu}_{00}(X)\Biggl(
\tilde{\psi}_{1, 1-S}(W)
-\frac{1-\tilde p_1(X)}{1-\tilde p_0(X)}\,\tilde{\psi}_{0, 1-S}(W)
\Biggr).
\end{align*}

Thus $\tilde{\tau}^{00}_{\mathrm{eif}}(x)$ is obtained by taking the EIF
representation for $\tau^{00}(x)$ and replacing the nuisance functions with
their plug-in counterparts.

\paragraph{Case (i): outcome model misspecified, propensity/principal scores correct.}
Assume
\[
\begin{aligned}
\tilde{\mu}_{10}(x) &\neq \mu_{10}(x), &
\tilde{\mu}_{00}(x) &\neq \mu_{00}(x), \\
\tilde{\pi}(x)      &=    \pi(x),      &
\tilde{p}_1(x)      &=    p_1(x),      &
\tilde{p}_0(x)      &=    p_0(x).
\end{aligned}
\]
Note that $\tilde{\psi}_{1,S}$ does not involve the outcome model, so
\[
\mathbb{E}[\tilde{\psi}_{1,S}(W)\mid X=x] = p_1(x)
\quad\Rightarrow\quad
\mathbb{E}[1-\tilde{\psi}_{1,S} (W)\mid X=x] = 1-p_1(x) = e^{00}(x).
\]

For $f(Y,S,X)=Y(1-S)$, the corresponding augmented score
$\tilde{\psi}_{1,Y(1-S)}$, $\tilde{\psi}_{0,Y(1-S)}$ have the usual double robustness property with respect
to the outcome regression and the principal score. In particular,
\begin{align*}
\mathbb{E}\!\left[\tilde{\psi}_{1,Y(1-S)}\mid X =x\right]
&= \mathbb{E}\!\left[
\frac{\mathbf{1}(Z=1)}{\mathbb{P}(Z=1\mid X)}
\Bigl\{Y(1-S)-\tilde{\mu}_{z0}(X)\{1-p_z(X)\}\Bigr\}
\Bigm| X =x\right]  \\
&\quad + \tilde{\mu}_{z0}(x)\{1-p_z(x)\}\\
&= \mathbb{E}\!\left[
Y(1-S)-\tilde{\mu}_{z0}(X)\{1-p_z(X)\}\mid Z= 1,X =x\right]\\
&\quad + \tilde{\mu}_{z0}(x)\{1-p_z(x)\}\\
&= \mu_{10}(x)\{1-p_1(x)\}.
\end{align*}
The first equality follows by definition. The second equality follows from a straightforward conditional expectation calculation.

Therefore, substituting the above identities, we obtain
\begin{equation*}
\begin{split}
\mathbb{E}[\tilde{\phi}_{0,00} (W)\mid X = x]
&= \frac{e^{00}(x)}{1-p_0(x)} \mathbb{E}[\tilde{\psi}_{0, Y(1-S)}(W)\mid X = x]\\
&\quad + \tilde{\mu}_{00}(x)\Big\{
      \mathbb{E}[\tilde{\psi}_{1, 1-S}(W)\mid X = x]
      -\frac{1-p_1(x)}{1-p_0(x)} \mathbb{E}[\tilde{\psi}_{0, 1-S}(W)\mid X = x]
    \Big\}\\
&= e^{00}(x)\mu_{00}(x).
\end{split}
\end{equation*}

A symmetric argument gives
\[
\mathbb{E}[\tilde{\phi}_{1,00}(W)\mid X =x] = e^{00}(x)\mu_{10}(x).
\]
Therefore, 
\begin{equation*}
\mathbb{E}\!\left[\frac{\tilde{\phi}_{1,00}-\tilde{\phi}_{0,00}}
                       {\mathbb{E}(1-\tilde{\psi}_{1,S}\mid X)}
                 \Bigm| X =x\right]
= \frac{e^{00}(x)\mu_{10}(X)-e^{00}(X)\mu_{00}(X)}{e^{00}(X)}
= \tau^{00}(X),
\end{equation*}
showing that $\tilde{\tau}^{00}_{\mathrm{EIF}}(X)$ remains unbiased for $\tau^{00}(X)$
even when the outcome model is misspecified, provided the propensity and principal
score models are correct.

\paragraph{Case (ii): propensity/principal scores misspecified, outcome model correct.}
We now consider the complementary scenario, where the propensity and principal
score models may be misspecified, but the outcome regressions are correctly
specified. Specifically, suppose
\[
\begin{aligned}
\tilde{\pi}(x) \neq \pi(x), &\qquad
\tilde{p}_1(x) \neq p_1(x)\ \text{or}\ \tilde{p}_0(x) \neq p_0(x),\\
\tilde{\mu}_{10}(x) = \mu_{10}(x), &\qquad
\tilde{\mu}_{00}(x) = \mu_{00}(x).
\end{aligned}
\]

Under misspecification of the propensity score and principal score, we can similarly calculate compotent expectation.

First, for $f(Y,S,X)=1-S$, we have
\begin{equation*}
\begin{split}
\mathbb{E}[\tilde{\psi}_{1, 1-S}\mid X =x] 
&= \mathbb{E}\!\left[
     \frac{\mathbf{1}(Z=1)\{1-S - (1- \tilde p_1(X))\}}{\tilde{\pi}(X)}
     \,\middle|\; X =x
   \right]+ 1- \tilde p_1(x)\\
&= \frac{\pi(x)}{\tilde{\pi}(x)}\big(\tilde p_1(x)-p_1(x)\big)+1-\tilde p_1(x).
\end{split}
\end{equation*}

Next, for $f(Y,S,X)=Y(1-S)$, we have
\begin{equation*}
\begin{split}
\mathbb{E}[\tilde{\psi}_{1, Y (1-S)}\mid X =x] 
&= \mathbb{E}\!\left[
     \frac{\mathbf{1}(Z=1)\{Y (1-S)-\mu_{10}(X) (1-\tilde p_1(X))\}}{\tilde{\pi}(X)}
     \,\middle|\;X =x
   \right] \\
  & \quad + \mu_{10}(x) (1-\tilde p_1(x))\\
  &= \frac{\pi(x)}{\tilde{\pi}(x)}\mu_{10}(x)\big(\tilde p_1(x) - p_1(x)\big)
   + \mu_{10}(x)(1- \tilde p_1(x))\\
&= \mu_{10}(x)\,\mathbb{E}[\tilde{\psi}_{1, 1-S}\mid X =x].
\end{split}
\end{equation*}

The second equality follows from a direct calculation. Since we have already computed
$\mathbb{E}\!\left[\tilde{\psi}_{1, 1-S}\mid X=x\right]$, the above display implies that the two scores are proportional.

A completely analogous calculation for $z=0$ yields
\[
\mathbb{E}[\tilde{\psi}_{0, Y(1-S)}\mid X =x]
= \mu_{00}(X)\,\mathbb{E}[\tilde{\psi}_{0, 1-S}\mid X =x].
\]

Substituting these identities into the EIF component for $u=00$,
\begin{equation*}
\begin{split}
\mathbb{E}[\tilde{\phi}_{0,00}\mid X =x] 
&= \frac{1-\tilde{p}_{1}(x)}{1-\tilde{p}_0(x)} \mathbb{E}[\tilde{\psi}_{0, Y(1-S)}\mid X =x]\\
&\quad + \mu_{00}(x)\left\{\mathbb{E}[\tilde{\psi}_{1, 1-S}\mid X =x]
      -\frac{1-\tilde{p}_1(X)}{1-\tilde{p}_0(X)} \mathbb{E}[\tilde{\psi}_{0, 1-S}\mid X =x]\right\}\\
&= \mu_{00}(x) \,\mathbb{E}[\tilde{\psi}_{1, 1-S}\mid X =x].
\end{split}
\end{equation*}
Similarly, we have
\[
\mathbb{E}[\tilde{\phi}_{1,00}\mid X] = \mu_{10}(X)\,\mathbb{E}[\tilde{\psi}_{1, 1-S}\mid X],
\]
so that in the ratio defining $\tilde{\tau}^{00}_{\mathrm{eif}}(X)$ the common factor
$\mathbb{E}[\tilde{\psi}_{1, 1-S}\mid X]$ cancels:
\begin{equation*}
\begin{split}
\mathbb{E}\left[\frac{\tilde{\phi}_{1,00}-\tilde{\phi}_{0,00}}
                     {\mathbb{E}[1-\tilde{\psi}_{1,S}\mid X]}
               \;\middle|\; X\right] 
&= \frac{\mu_{10}(X)\,\mathbb{E}[\tilde{\psi}_{1, 1-S}\mid X]
       - \mu_{00}(X)\,\mathbb{E}[\tilde{\psi}_{1, 1-S}\mid X]}
        {\mathbb{E}[1-\tilde{\psi}_{1,S}\mid X]}\\
&= \mu_{10}(X)-\mu_{00}(X) = \tau^{00}(X).
\end{split}
\end{equation*}

Thus, even when the propensity and principal score models $(\tilde\pi,\tilde p_0,\tilde p_1)$
are misspecified, as long as the outcome regressions $(\mu_{10},\mu_{00})$ are correctly
specified, the EIF-based functional $\tilde{\tau}^{00}_{\mathrm{EIF}}(X)$ remains unbiased
for $\tau^{00}(X)$. Together with Case (i), this establishes weak double robustness for the
never-taker stratum $u=00$.

    \item \textbf{Rate robustness} \\
    To study the bias of the EIF-based population functional, we compare the plug-in
version $\tilde{\tau}^{u}_{\mathrm{eif}}(x)$ with the true function
$\tau^{u}(x)$. Denote $e_{\pi}(x) := \tilde \pi(x) - \pi(x)$, $e_{p_z}(x) := \tilde p_z(x) - p_z(x)$,$e_{\mu_{zs}}(x):= \tilde \mu_{zs}(x) - \mu_{zs}(x)$.
Then the EIF population bias admits second–order product expansions of the form
\begin{equation}\label{eq:bias_eif_00}
\begin{aligned}
\tilde{\tau}_{\mathrm{eif}}^{00}(x)-\tau^{00}(x)
&= \frac{1}{\Eg{00}}\Biggl[
\frac{1-\tilde p_1(x)}{1-\tilde\pi(x)}\,e_{\pi}(x)\,e_{\mu_{00}}(x)
+ \frac{1-\tilde p_1(x)}{\tilde\pi(x)}\,e_{\pi}(x)\,e_{\mu_{10}}(x) \\
&\qquad
- \frac{1-\tilde p_1(x)}{1-\tilde p_0(x)}\,e_{p_0}(x)\,e_{\mu_{00}}(x)
- e_{p_1}(x)\,e_{\mu_{00}}(x)
+ R_3^{00}(x)
\Biggr].
\end{aligned}
\end{equation}
\begin{equation}\label{eq:bias_eif_10}
\begin{aligned}
\tilde{\tau}_{\mathrm{eif}}^{10}(x)-\tau^{10}(x)
&= \frac{1}{\Eg{10}}\Biggl[
-\frac{\tilde p_0(x)}{\tilde p_1(x)}\,e_{p_1}(x)\,e_{\mu_{11}}(x)
+ e_{p_0}(x)\,e_{\mu_{11}}(x)
+ e_{p_1}(x)\,e_{\mu_{00}}(x) \\
&\qquad
-\frac{1-\tilde p_1(x)}{1-\tilde p_0(x)}\,e_{p_0}(x)\,e_{\mu_{00}}(x)
+ \frac{\tilde p_1(x)-\tilde p_0(x)}{\tilde\pi(x)}\,e_{\pi}(x)\,e_{\mu_{11}}(x) \\
&\qquad
+ \frac{\tilde p_1(x)-\tilde p_0(x)}{1-\tilde\pi(x)}\,e_{\pi}(x)\,e_{\mu_{00}}(x)
+ R_3^{10}(x)
\Biggr].
\end{aligned}
\end{equation}
\begin{equation}\label{eq:bias_eif_11}
\begin{aligned}
\tilde{\tau}_{\mathrm{eif}}^{11}(x)-\tau^{11}(x)
&= \frac{1}{\Eg{11}}\Biggl[
\frac{\tilde p_0(x)}{\tilde p_1(x)}\,e_{p_1}(x)\,e_{\mu_{11}}(x)
- e_{p_0}(x)\,e_{\mu_{11}}(x) \\
&\qquad
+ \frac{\tilde p_0(x)}{\tilde\pi(x)}\,e_{\pi}(x)\,e_{\mu_{11}}(x)
+ \frac{\tilde p_0(x)}{1-\tilde\pi(x)}\,e_{\pi}(x)\,e_{\mu_{01}}(x)
+ R_3^{11}(x)
\Biggr].
\end{aligned}
\end{equation}

where $R_3^{u}(x)$, $u\in\{00,10,11\}$, collects the remaining third-order
terms involving products of three nuisance errors. For example, for $u =00$, $R^{00}_3(x) = - \frac{e_{\mu_{00}}(x)e_{\pi}(x)e_{p_1}(x)}{\tilde\pi(x)} 
  + \frac{1-\tilde p_1(x)}{1-\tilde p_0(x)}\,
    \frac{e_{\pi}(x)e_{\mu_{00}}(x)e_{p_0}(x)}{(1-\tilde\pi(x))}$.
By the expressions for $u\in\{00,10,11\}$, we can write, for each $u$,
\begin{equation}\label{eq:eif-expansion-generic}
\tilde{\tau}_{\mathrm{eif}}^{u}(x)-\tau^{u}(x)
= \sum_{(a,b)\in\mathcal{I}_u} c^{u}_{a,b}(x)\,e_a(x)\,e_b(x)
  +  \frac{1}{\mathbb{E}\!\left[\tilde g^{u}\mid X = x\right]}R_3^{u}(x),
\end{equation}
where:
\begin{itemize}
  \item the index set $\mathcal{I}_u$ contains the relevant pairs of nuisance
  components, e.g.\ $(a,b)\in\{(\pi,\mu_{00}),(\pi,\mu_{10}),(p_0,\mu_{00}),
  (p_1,\mu_{00})\}$ for $u=00$, and analogous pairs for $u=10,11$;
  \item $c^{u}_{a,b}(x)$ are the deterministic weight functions formed by the
  various ratios of $\tilde\pi(x)$, $\tilde p_z(x)$ and
  $\mathbb{E}[\tilde g^{u}\mid X=x]$ in the displayed expansions;
  \item $R_3^{u}(x)$ collects the remaining third-order products of nuisance
  errors.
\end{itemize}
By assumption, an overlap (or positivity) condition holds for the plug–in
nuisance functions: there exists a constant $\epsilon>0$ such that, almost
surely, $\epsilon \le \mathbb{E}(\tilde g^{u}\mid X =x)$, $\text{for each }u\in\{00,10,11\},$
and $\epsilon \le \tilde\pi(x) \le 1-\epsilon$, $
\epsilon \le \tilde p_1(x)$, $\tilde p_0(x) \le 1-\epsilon$,
$\tilde p_1(x) - \tilde p_0(x) \ge \epsilon$.
Hence all ratios appearing in the
coefficients $c^{u}_{a,b}(x)$ are uniformly bounded in $x$, and there exists a
constant $C<\infty$ such that
\[
\sup_{x}\,\max_{(a,b)\in\mathcal{I}_u} |c^{u}_{a,b}(x)| \le C
\quad\text{for each }u\in\{00,10,11\}.
\]
Hence, there exist a constant $C^{\prime} = \max (C, 1/\epsilon)$.
Taking absolute values in \eqref{eq:eif-expansion-generic}, we obtain
\begin{align*}
|\tilde{\tau}_{\mathrm{eif}}^{u}(x)-\tau^{u}(x)|
&\le \sum_{(a,b)\in\mathcal{I}_u}\,|c^{u}_{a,b}(x)\,e_a(x)\,e_b(x)|
   + \frac{1}{\epsilon}\,|R_3^{u}(x)|\,\\
&\le C^{\prime} \Big(\sum_{(a,b)\in\mathcal{I}_u}\,|e_a(x)\,e_b(x)|\,
   + \,|R_3^{u}(x)|\,\Big).
\end{align*}
Therefore
\[
|\tilde{\tau}_{\mathrm{eif}}^{u}(x)-\tau^{u}(x)|\
\le C^{\prime} \Big( \sum_{(a,b)\in\mathcal{I}_u}
    |e_a|\,|e_b|
  + |R_3^{u}(x)|\Big).
\]

From the explicit expansions for $u=00,10,11$, the pairs $(a,b)$ in
$\mathcal{I}_u$ always involve one outcome regression error
$e_{\mu_{zs}}$ and one of the score errors $e_{\pi}$ or $e_{p_z}$.
Consequently, for each $u\in\{00,10,11\}$ we have
\begin{align*}
|\tilde{\tau}_{\mathrm{eif}}^{u}(x)-\tau^{u}(x)|\
= O_p\!\Bigl(
|\tilde\mu_{zs}-\mu_{zs}|
\bigl\{|\tilde\pi-\pi|+|\tilde p_z-p_z|\bigr\} +|R_3^{u}(x)|
\Bigr) 
\end{align*}
Finally, $R_3^{u}(x)$ consists of products of three nuisance errors (e.g.
$e_{\pi} e_{p_z} e_{\mu_{zs}}$), so these third-order terms are of
smaller order than the second-order products and can be absorbed. Thus, we obtain the rate double robustness bound:
\begin{align*}
|\tilde{\tau}_{\mathrm{eif}}^{u}(x)-\tau^{u}(x)|\
= O_p\!\Bigl(
|\tilde\mu_{zs}-\mu_{zs}|,
\bigl\{|\tilde\pi-\pi|+|\tilde p_z-p_z|\bigr\} 
\Bigr) 
\end{align*}
with the indices of the principal score $p_z$ and outcome regression
$\mu_{zs}$ determined by the stratum $u$.

Now we will show how to obtain the exact bias. The bias can be decomposed into numerator bias and denominator bias.
\begin{equation*}
    \begin{aligned}
         \tilde{\tau}_{\text {eif }}^u(x)-\tau^u(x) &=\frac{\mathbb{E}\left[\tilde{\phi}_{1, u}-\tilde{\phi}_{0, u} \mid X =x\right]}{\mathbb{E}\left[\tilde{g}^u \mid X = x\right]}-\tau^u(x)\\
         & = \frac{\mathbb{E}\left[\tilde{\phi}_{1, u}-\tilde{\phi}_{0, u} \mid X =x\right]}{\mathbb{E}\left[\tilde{g}^u \mid X = x\right]} - \frac{\mathbb{E}\left[{\phi}_{1, u}-{\phi}_{0, u} \mid X =x\right]}{\mathbb{E}\left[\tilde{g}^u \mid X = x\right]} \\
         &+ \frac{\mathbb{E}\left[{\phi}_{1, u}-{\phi}_{0, u} \mid X =x\right]}{\mathbb{E}\left[\tilde{g}^u \mid X = x\right]} - \frac{\tau^u(x)\mathbb{E}\left[\tilde{g}^u \mid X = x\right]}{\mathbb{E}\left[\tilde{g}^u \mid X = x\right]}
    \end{aligned}
\end{equation*}
The first equality holds by definition. The second equality follows by adding and subtracting $\frac{\mathbb{E}\!\left[\phi_{1,u}-\phi_{0,u}\mid X=x\right]}
     {\mathbb{E}\!\left[\tilde g^{u}\mid X=x\right]},$
and then rearranging terms.

From the EIF identification, we know $\mathbb{E}\left[{\phi}_{1, u}-{\phi}_{0, u} \mid X =x\right] = \tau^u(x)\mathbb{E}[g^u|X= x]$. Hence, for the second part, we have
\begin{equation*}
   = \frac{\tau^u(x)\mathbb{E}[g^u|X= x]}{\mathbb{E}\left[\tilde{g}^u \mid X = x\right]} - \frac{\tau^u(x)\mathbb{E}\left[\tilde{g}^u \mid X = x\right]}{\mathbb{E}\left[\tilde{g}^u \mid X = x\right]}    
\end{equation*}
So the bias can be decomposed into numerator bias and denominator bias:
\begin{equation*}
    \begin{aligned}
         \tilde{\tau}_{\text {eif }}^u(x)-\tau^u(x) &= \frac{\mathbb{E}\left[\tilde{\phi}_{1, u}-\tilde{\phi}_{0, u} \mid X =x\right]- \mathbb{E}\left[{\phi}_{1, u}-{\phi}_{0, u} \mid X =x\right]}{\mathbb{E}\left[\tilde{g}^u \mid X = x\right]} \\
         &- \frac{\tau^u(x)(\mathbb{E}\left[\tilde{g}^u \mid X = x\right] - \mathbb{E}\left[g^u \mid X = x\right])}{\mathbb{E}\left[\tilde{g}^u \mid X = x\right]}
    \end{aligned}
\end{equation*}
We use the never-taker stratum $u=00$ as an illustrative example; the derivations for $u=10$ and $u=11$ are entirely analogous.

We first compute the contribution from the $Z=1$ part in numerator bias, i.e.
\(\mathbb{E}\big[(\tilde{\phi}_{1,00}-\phi_{1,00})\mid X =x\big]\).
Recall that \(\tilde{\phi}_{1,00} = \tilde{\psi}_{1, Y(1-S)}\) and
\(\phi_{1,00} = \psi_{1,Y(1-S)}\).
By calculation, we have
\begin{align*}
\mathbb{E}[\tilde{\psi}_{1, Y (1-S)}\mid X =x] 
&= \mathbb{E}\!\left[
     \frac{\mathbf{1}(Z=1)\{Y (1-S)-\tilde \mu_{10}(X) (1-\tilde p_1(X))\}}{\tilde{\pi}(X)}
     \,\Bigm|\;X
   \right] \\
   &\quad + \tilde \mu_{10}(x) \{1-\tilde p_1(x)\}\\
= \frac{\pi(x)}{\tilde \pi(x)}&
   \big(\mu_{10}(x)\{1-p_1(x)\}-\tilde \mu_{10}(x)\{1-\tilde p_1(x)\}\big)
   + \tilde \mu_{10}(x)\{1-\tilde p_1(x)\}.
\end{align*}
Recall \(\mathbb{E}[ \psi_{1,Y(1-S)}\mid X =x] = \mu_{10}(x)\{1- p_1(x)\}\):
\begin{align*}
\mathbb{E}[\tilde{\phi}_{1,00}(W)-\phi_{1,00}(W)\mid X =x]
&= \mathbb{E}[\tilde{\psi}_{1, Y (1-S)}\mid X =x]
   - \mathbb{E}[{\psi}_{Y_1 (1-S_1)}\mid X =x]\\
&= \frac{\tilde \pi(x)-\pi(x)}{\tilde \pi(x)}
   \Big(\tilde \mu_{10}(x)\{1-\tilde p_1(x)\}-\mu_{10}(x)\{1-p_1(x)\}\Big).
\end{align*}
We first expand the numerator bias in terms of the nuisance errors. Recall
$e_{\pi}(x) = \tilde\pi(x)-\pi(x)$, 
$e_{p_z}(x) = \tilde p_z(x)-p_z(x)$,
$e_{\mu_{zs}}(x) = \tilde\mu_{zs}(x)-\mu_{zs}(x)$.
From the previous step,
\[
\mathbb{E}[\tilde{\phi}_{1,00}(W)-\phi_{1,00}(W)\mid X =x]
= \frac{e_{\pi}(x)}{\tilde\pi(x)}\,
  \Big(\tilde \mu_{10}(x)\{1-\tilde p_1(x)\}-\mu_{10}(x)\{1-p_1(x)\}\Big).
\]
Add and subtract $\mu_{10}(x)\{1-\tilde p_1(x)\}$ inside the bracket:
\begin{align*}
\mathbb{E}[\tilde{\phi}_{1,00}-\phi_{1,00}\mid X =x]
&= \frac{e_{\pi}(x)}{\tilde\pi(x)} \Big(
      \tilde \mu_{10}(x)\{1-\tilde p_1(x)\}
      -\mu_{10}(x)\{1-\tilde p_1(x)\} \\
&\qquad\qquad\qquad
      + \mu_{10}(x)\{1-\tilde p_1(x)\}
      - \mu_{10}(x)\{1-p_1(x)\}
    \Big)\\
&= \frac{e_{\pi}(x)}{\tilde\pi(x)}\Big(
      e_{\mu_{10}}(x)\{1-\tilde p_1(x)\}
      - \mu_{10}(x)e_{p_1}(x)
    \Big).
\end{align*}
Equivalently,
\[
\mathbb{E}[\tilde{\phi}_{1,00}(W)-\phi_{1,00}(W)\mid X =x]
= \frac{1-\tilde p_1(x)}{\tilde \pi(x)}\,e_{\pi}(x)\,e_{\mu_{10}}(x)
  - \frac{\mu_{10}(x)}{\tilde \pi(x)}\,e_{\pi}(x)\,e_{p_1}(x).
\]

Recall $\tilde \phi_{0,00}(W)
= \frac{1-\tilde p_1(X)}{1-\tilde p_0(X)} \tilde \psi_{0,Y(1-S)}
 +\tilde \mu_{00}(X)\biggl\{\tilde \psi_{1,1-S}
   -\frac{1-p_1(X)}{1-p_0(X)} \psi_{0,1-S}\biggr\}.$ And $\mathbb{E}[\phi_{0,00}(W)|X=x] = \mu_{00}(x)\{1-p_1(x)\}$
   
So we have 
\begin{align*}
    &\mathbb{E}\bigl[\tilde\phi_{0,00} - \phi_{0,00}\mid X =x\bigr] \\[2pt]
&= \frac{1-\tilde p_1(X)}{1-\tilde p_0(X)}
      \mathbb{E}\bigl[\tilde\psi_{0,Y(1-S)}\mid X =x\bigr]
   + \tilde\mu_{00}(x)\Big(
          \mathbb{E}\bigl[\tilde\psi_{1,1-S}\mid X =x\bigr]\\
         &\hspace{8em}  - \frac{1-\tilde p_1(X)}{1-\tilde p_0(X)}
            \mathbb{E}\bigl[\tilde\psi_{0,1-S}\mid X =x\bigr]
        \Big) - \mu_{00}(x)\{1-p_1(x)\} 
\end{align*}
To expose the same bias structure as in the previous calculations, we apply an
“add–and–subtract” trick: we add and subtract $\mu_{00}(x)\{1-p_0(x)\}$ and the
terms involving $(1-p_z)$ inside the brackets.  Define \[
A(x):=
\mathbb{E}[\tilde\psi_{1,1-S}\mid X =x]-\{1-p_1(x)\}
-\frac{1-\tilde p_1(x)}{1-\tilde p_0(x)}
\Big(\mathbb{E}[\tilde\psi_{0,1-S}\mid X =x]-\{1-p_0(x)\}\Big).
\]This yields
\begin{align*}
&= \frac{1-\tilde p_1(x)}{1-\tilde p_0(x)}
   \Big(\mathbb{E}\bigl[\tilde\psi_{0,Y(1-S)}\mid X =x\bigr]
        - \mu_{00}(x)\{1-p_0(x)\}\Big) \\[2pt]
&\quad
   + \frac{1-\tilde p_1(x)}{1-\tilde p_0(x)}\,\mu_{00}(x)\{1-p_0(x)\}
   + \tilde\mu_{00}(x)A(x)\\[2pt]
&\quad
   + \tilde\mu_{00}(x)\{1-p_1(x)\}
   - \tilde\mu_{00}(x)\,\frac{1-\tilde p_1(x)}{1-\tilde p_0(x)}\{1-p_0(x)\}
   - \mu_{00}(x)\{1-p_1(x)\}.
\end{align*}
Then we add and subtract 
$\mu_{00}(x)A(x)$.

Summarizing, we can decompose this bias into four parts:
\begin{align*}
&\mathbb{E}\bigl[\tilde\phi_{0,00}\mid X =x\bigr] - \mu_{00}(x)\{1-p_1(x)\} \\[2pt]
&= \underbrace{\frac{1-\tilde p_1(x)}{1-\tilde p_0(x)}
   \Big\{
      \mathbb{E}\bigl[\tilde\psi_{0,Y(1-S)}\mid X =x\bigr]
      - \mu_{00}(x)\{1-p_0(x)\}
   \Big\}}_{(I)} \\[4pt]
&\quad + \underbrace{e_{\mu_{00}}(x)A(x)}_{(II)}  + \underbrace{\mu_{00}(X)
  A(x)}_{(III)} \\[4pt]
&\quad + \underbrace{e_{\mu_{00}}(x)\{1-p_1(x)\}
   - \frac{1-\tilde p_1(x)}{1-\tilde p_0(x)}\{1-p_0(x)\}\,e_{\mu_{00}}(x)}_{(IV)}.
\end{align*}

We will calculate IV part first. 
\begin{align*}
(IV)
&= \Big\{(1-p_1(x)) - \frac{1-\tilde p_1(x)}{1-\tilde p_0(x)}\{1-p_0(x)\}\Big\}\,e_{\mu_{00}}(x)\\[4pt]
&= \frac{(1-p_1(x))(1-\tilde p_0(x)) - (1-\tilde p_1(x))(1-p_0(x))}
        {1-\tilde p_0(x)}\,e_{\mu_{00}}(x)\\[4pt]
&= \frac{(1-\tilde p_0(x))\big(\tilde p_1(x)-p_1(x)\big)
        - (1-\tilde p_1(x))\big(\tilde p_0(x)-p_0(x)\big)}
        {1-\tilde p_0(x)}\,e_{\mu_{00}}(x)\\[4pt]
&= e_{p_1}(x)\,e_{\mu_{00}}(x)
   - \frac{1-\tilde p_1(x)}{1-\tilde p_0(x)}\,e_{p_0}(x)\,e_{\mu_{00}}(x).
\end{align*}

Notice II and III have the same part $A(x)$. From the previous, we know:
  \begin{align*}
\mathbb{E}[\tilde \psi_{1 - S_1}\mid X=x] - (1 - p_1(x))
&= \left(\frac{\pi(x)}{\tilde \pi(x)} -1\right)\bigl(\tilde p_1(x) - p_1(x)\bigr)
= - \frac{e_{\pi}(x)\,e_{p_1}(x)}{\tilde \pi(x)},\\
\mathbb{E}[\tilde \psi_{1 - S_0}\mid X=x] - (1 - p_0(x))
&= \left(\frac{1 - \pi(x)}{1 - \tilde \pi(x)} -1\right)\bigl(\tilde p_0(x) - p_0(x)\bigr)
= \frac{e_{\pi}(x)\,e_{p_0}(x)}{1 - \tilde \pi(x)}.
\end{align*}
    Hence,
\begin{align*}
(II)
&= - \frac{e_{\pi}(x)e_{p_1}(x)e_{\mu_{00}}(x)}{\tilde \pi(x)}
   - \frac{1 - \tilde p_1(x)}{1 - \tilde p_0(x)}
     \frac{e_{\pi}(x)e_{\mu_{00}}(x)e_{p_0}(x)}{1 - \tilde \pi(x)},\\
(III)
&= - \frac{\mu_{00}(x)}{\tilde \pi(x)}\,e_{\pi}(x)e_{p_1}(x)
   - \frac{1 - \tilde p_1(x)}{1 - \tilde p_0(x)}\,\mu_{00}(x)
     \frac{e_{\pi}(x)e_{p_0}(x)}{1 - \tilde \pi(x)}.
\end{align*}

Finally, for the $I$ term, we proceed by applying the same algebraic trick
as for $\tilde \psi_{1,Y(1-S)}$. First, we have
\begin{equation*}
\begin{aligned}
        &\mathbb{E}\bigl[\tilde \psi_{0,Y(1-S)}\mid X = x\bigr]\\
   & = \frac{1 - \pi(x)}{1- \tilde \pi(x)}
      \bigl[\mu_{00}(x)\{1 - p_0(x)\} - \tilde \mu_{00}(x)\{1 - \tilde p_0(x)\}\bigr]
      + \tilde \mu_{00}(x)\{1 - \tilde p_0(x)\}.
\end{aligned}
\end{equation*}
Hence, the bias is 
\begin{align*}
&\mathbb{E}\big[\tilde\psi_{0,Y(1-S)}\mid X =x\big]-\mu_{00}(x)\{1-p_0(x)\}\\
&= -\Big(\frac{1-\tilde\pi (x)}{1-\pi(x)}-1\Big)
   \Big(\tilde\mu_{00}(x)\{1-\tilde p_0 (x)\}-\mu_{00}(x)\{1-p_0(x)\}\Big) \\[4pt]
&= -\frac{e_{\pi}(x)}{1-\tilde\pi(x)}
   \Big(\tilde\mu_{00}(x)\{1-\tilde p_0 (x)\}-\mu_{00}(x)\{1-p_0(x)\}\Big) \\[4pt]
&= -\frac{1-\tilde p_0 (x)}{1-\tilde\pi (x)}\,e_{\pi}(x)\, e_{\mu_{00}}(x)
   + \frac{\mu_{00}(x)}{1-\tilde\pi(x)}\,e_{\pi}(x)\, e_{p_0}(x).
\end{align*}
The first equality follows from the previous calculation. The second equality is obtained by factoring terms. The third equality follows by adding and subtracting the term $\mu_{00}(x)\{1-\tilde p_0(x)\}$, and then rearranging.

So the I term is 
\[
I \;=\; -\frac{1-\tilde p_1 (x)}{1-\tilde\pi (x)}\,e_{\pi}(x)\, e_{\mu_{00}}(x)
      + \frac{1-\tilde p_1 (x)}{1-\tilde p_0 (x)}\,
        \frac{\mu_{00}(x)}{1-\tilde\pi(x)}\,e_{\pi}(x)\, e_{p_0}(x).
\]

Combining together, we obtain
\begin{align*}
&\mathbb{E}\bigl[\tilde{\phi}_{0,00} \mid X = x\bigr]
 - \mathbb{E}\bigl[{\phi}_{0,00} \mid X = x\bigr]
= -\frac{1-\tilde p_1(x)}{1-\tilde\pi(x)}\,e_{\pi}(x)e_{\mu_{00}}(x)\\[4pt]
&\quad - \frac{\mu_{00}(x)}{\tilde\pi(x)}\,e_{\pi}(x)e_{p_1}(x)  - \frac{e_{\mu_{00}}(x)e_{\pi}(x)e_{p_1}(x)}{\tilde\pi(x)} - \frac{1-\tilde p_1(x)}{1-\tilde p_0(x)}\,
          \frac{e_{\pi}(x)e_{\mu_{00}}(x)e_{p_0}(x)}{1-\tilde\pi(x)} \\[4pt]
&\quad + e_{p_1}(x)e_{\mu_{00}}(x)
        - \frac{1-\tilde p_1(x)}{1-\tilde p_0(x)}\,e_{\mu_{00}}(x)e_{p_0}(x).
\end{align*}

Hence, the numerator bias is
\begin{align*}
&\mathbb{E}\bigl[\tilde{\phi}_{1,00}-\tilde{\phi}_{0,00} \mid X = x\bigr]
  - \mathbb{E}\bigl[{\phi}_{1,00}-{\phi}_{0,00} \mid X = x\bigr] \\
&\qquad =
  \frac{1-\tilde p_1(x)}{\tilde \pi(x)}\,e_{\pi}(x)\,e_{\mu_{10}}(x)
  - \frac{\mu_{10}(x) - \mu_{00}(x)}{\tilde \pi(x)}\,e_{\pi}(x)\,e_{p_1}(x) \\[4pt]
&\qquad\quad
  + \frac{1-\tilde p_1(x)}{1-\tilde\pi(x)}\,e_{\pi}(x)e_{\mu_{00}}(x)
  - \frac{e_{\mu_{00}}(x)e_{\pi}(x)e_{p_1}(x)}{\tilde\pi(x)} \\[4pt]
&\qquad\quad
  + \frac{1-\tilde p_1(x)}{1-\tilde p_0(x)}\,
    \frac{e_{\pi}(x)e_{\mu_{00}}(x)e_{p_0}(x)}{1-\tilde\pi(x)}
  - e_{p_1}(x)e_{\mu_{00}}(x)
  + \frac{1-\tilde p_1(x)}{1-\tilde p_0(x)}\,e_{\mu_{00}}(x)e_{p_0}(x).
\end{align*}

The denomiator bias is 
$$
\mathbb{E}\left[\tilde{g}^{00} - g^{00} \mid X = x\right] = - \frac{e_{\pi}(x)e_{p_1}(x)}{\tilde \pi(x)}
$$
Again, we are interested in $$\mathbb{E}\left[\tilde{\phi}_{1, u}-\tilde{\phi}_{0, u} \mid X =x\right]- \mathbb{E}\left[{\phi}_{1, u}-{\phi}_{0, u} \mid X =x\right]- \tau^u(x)(\mathbb{E}\left[\tilde{g}^u \mid X = x\right] - \mathbb{E}\left[g^u \mid X = x\right])$$
Notice for $e_{\pi}(x)e_{p_1}(x)$ part, it can be canceled. 
So the EIF population bias is 
\begin{equation*}
    \begin{aligned}
                 &\tilde{\tau}_{\text {eif }}^{00}(x)-\tau^{00}(x) =\frac{1}{\Eg{00}}\Biggl[
\frac{1-\tilde p_1(x)}{1-\tilde\pi(x)}\,e_{\pi}(x)\,e_{\mu_{00}}(x)
+ \frac{1-\tilde p_1(x)}{\tilde\pi(x)}\,e_{\pi}(x)\,e_{\mu_{10}}(x) \\
&\hspace{12em}
- \frac{1-\tilde p_1(x)}{1-\tilde p_0(x)}\,e_{p_0}(x)\,e_{\mu_{00}}(x)
- e_{p_1}(x)\,e_{\mu_{00}}(x)
+ R_3^{00}(x)
\Biggr].
    \end{aligned}
\end{equation*}
where $R^{00}_3(x) = - \frac{e_{\mu_{00}}(x)e_{\pi}(x)e_{p_1}(x)}{\tilde\pi(x)} 
  + \frac{1-\tilde p_1(x)}{1-\tilde p_0(x)}\,
    \frac{e_{\pi}(x)e_{\mu_{00}}(x)e_{p_0}(x)}{(1-\tilde\pi(x))}$.
 \end{enumerate}
\subsection{Proof of Theorem \ref{one-step}} \label{s:proof7}

We now verify that each $\tau^{u}(x)$ admits an one-step style decomposition of the
form
\[
\tau^{u}(x) \;=\; \mathbb{E}\bigl(\check{\tau}^{u}(X) + \phi_{u} \mid X =x\bigr),
\qquad u\in\{00,10,11\},
\]
We illustrate the argument for
$u=10$; the cases $u=00$ and $u=11$ are entirely analogous.

For the complier stratum $u=10$, define
\[
\phi_{10}(W)
:= \frac{\phi_{1,10}(W)-\phi_{0,10}(W)-\check{\tau}^{10}(X)\,(\psi_{1,S}-\psi_{0,S})}
         {e^{10}(X)}.
\]
Then
\begin{align*}
\mathbb{E}\bigl(\check{\tau}^{10}(X) + \phi_{10}\mid X=x\bigr) 
&= \check{\tau}^{10}(x) 
   + \mathbb{E}\!\left(
       \frac{\phi_{1,10}-\phi_{0,10}}{e^{10}(X)}
       \,\Bigm|\;X=x
     \right) \\
&\qquad\quad
   - \check{\tau}^{10}(x)\,
     \mathbb{E}\!\left(
       \frac{\psi_{1,S}-\psi_{0,S}}{e^{10}(X)}
       \,\Bigm|\;X=x
     \right)\\
&= \mathbb{E}\!\left(
     \frac{\phi_{1,10}-\phi_{0,10}}{e^{10}(X)}
     \,\Bigm|\;X=x
   \right) = \tau^{10}(x),
\end{align*}

where the two terms involving $\check{\tau}^{10}(X)$ cancel because
\[
\mathbb{E}\!\left(
  \frac{\psi_{1,S}-\psi_{0,S}}{e^{10}(X)}
  \,\Bigm|\;X =x
\right) = 1,
\]
In the last step we used the EIF identification formula
\[
\tau^{10}(x)
= \frac{\mathbb{E}\bigl(\phi_{1,10}(W)-\phi_{0,10}(W)\mid X =x\bigr)}
       {\mathbb{E}(\psi_{1,S}-\psi_{0,S}\mid X =x)}.
\]

\subsection{Proof of Theorem \ref{robustness_one}}\label{s:proof8}
Theorem \ref{robustness_one} showed that the one-step population function has multiple robustness. 
For the weak robustness property, we prove the result for
$\tilde{\tau}_{\mathrm{eif}}^{00}(x)$ as an illustrative example; the arguments
for the other two strata are analogous.
Recall the one-step based population functional is
\[
\tilde{\tau}^{00}_{\mathrm{one}}(x)
:= \mathbb{E}\!\left[
  \check{\tau}^{00}(X)
  + \frac{\tilde{\phi}_{1,00}(W)-\tilde{\phi}_{0,00}(W)
          - \check{\tau}^{00}(X)\,\tilde g^{00}(W)}
         {\tilde e^{00}(X)}
  \,\Bigg|\, X = x
\right],
\]

\begin{enumerate}
    \item \textbf{Weak robustness} \\
\textbf{Case (i): propensity/principal scores correct}\\
Assume $\tilde{\pi}(X) = \pi(X)$, $\tilde{p}_1(X) = p_1(X)$, $\tilde{p}_0(X) = p_0(X)$, $\tilde{\mu}_{10}(X) \neq \mu_{10}(X)$, $\tilde{\mu}_{00}(X) \neq \mu_{00}(X)$, $\check\tau^{00}(X)\neq \tau^{00}(X)$.\\
From the previous proof in EIF, we already showed that $\mathbb{E}[\tilde{\phi}_{1,00}-\tilde{\phi}_{0,00}\mid X=x] = e^{00}(x)\tau^{00}(x), \quad \mathbb{E}[\tilde g^{00}\mid X = x] = e^{00}(x) = 1 - p_1(x)$. 
Hence, 
\begin{equation*}
    \tilde \tau^{00}(x) = \check \tau^{00}(x) + \frac{e^{00}(x)\tau^{00}(x) - e^{00}(x)\check \tau^{u}(x)}{e^{00}(x)} = \tau^{00}(x)
\end{equation*}

\textbf{Case (ii): outcome model and preliminary estimator correct}\\
Assume $\tilde{\mu}_{10}(X) = \mu_{10}(X)$, $\tilde{\mu}_{00}(X) = \mu_{00}(X)$, $\check\tau^{00}(X) =  \tau^{00}(X)$, $\tilde{\pi}(X) \neq \pi(X)$, $\tilde{p}_1(X) \neq p_1(X)$, $\tilde{p}_0(X) \neq p_0(X)$.\\
From the previous proof in EIF, we already showed that $\mathbb{E}[\tilde{\phi}_{1,00}-\tilde{\phi}_{0,00}\mid X=x] = \tau^{00}(x)\mathbb{E}[\tilde g^u \mid X =x]$ under this scenario. 
Hence,
    \begin{equation*}
    \tilde \tau^{00}(x) = \tau^{00}(x) + \frac{\tau^{00}(x)\mathbb{E}[\tilde g^u \mid X =x] - \tau^{00}(x)\mathbb{E}[\tilde g^u \mid X =x]}{\tilde e^{00}(x)} = \tau^{00}(x)
\end{equation*}

\textbf{Case (iii): outcome model and principal scores correct}\\
Assume $\tilde{\mu}_{10}(X) = \mu_{10}(X)$, $\tilde{\mu}_{00}(X) = \mu_{00}(X)$, $\tilde{p}_1(X) = p_1(X)$, $\tilde{p}_0(X) = p_0(X)$, $\check\tau^{00}(X) \neq \tau^{00}(X)$, $\tilde{\pi}(X) \neq \pi(X)$.\\
From the previous proof in EIF, we already showed that $\mathbb{E}[\tilde{\phi}_{1,00}-\tilde{\phi}_{0,00}\mid X=x] = e^{00}(x)\tau^{00}(x), \quad \mathbb{E}[\tilde g^{00}\mid X = x] = e^{00}(x) = 1 - p_1(x)$. 
Hence, 
\begin{equation*}
    \tilde \tau^{00}(x) = \check \tau^{00}(x) + \frac{e^{00}(x)\tau^{00}(x) - e^{00}(x)\check \tau^{u}(x)}{e^{00}(x)} = \tau^{00}(x)
\end{equation*}
    \item \textbf{Rate robustness}\\  
Follow the error notation from EIF part and denote $e_{\check \tau^{u}}(x) := \check\tau^{u}(x) - \tau^{u}(x).$ The proof proceeds in two steps. This follows the same strategy in the EIF analysis. First, we derive an exact expansion of the
bias of the one-step population function,
in terms of products of the nuisance estimation errors
$e_\pi$, $e_{p_z}$, $e_{\mu_{zs}}$, and $e_{\check\tau^u}$.

In the second step, we control the $L_1$ error by applying the overlap
assumption to each product term in the expansion. The overlap condition uniformly bounds the denominators
appearing in the bias expression, which is the same as in the EIF analysis. Hence we will only show the results in the first step. 

We show the exact expansion of the
bias of the one-step population function first. Then we will use $u =00$ as an illustrative example to show how to obtain the result.
For u = 00
\begin{equation}\label{eq:bias_one_00}
    \begin{aligned}
        \tilde{\tau}^{00}_{\mathrm{one}}(x) - \tau^{00}(x)
&= \frac{1}{\tilde{\pi}(x)}\, e_{\pi}(x)\, e_{\mu_{10}}(x)
  + \frac{1}{1-\tilde{\pi}(x)}\, e_{\pi}(x)\, e_{\mu_{00}}(x) \\[4pt]
&\quad - \frac{1}{1-\tilde{p}_1(x)}\, e_{p_1}(x)\, e_{\mu_{00}}(x)
  + \frac{1}{1-\tilde{p}_0(x)}\, e_{\mu_{00}}(x)\, e_{p_0}(x) \\[4pt]
&\quad - \frac{e_{\check{\tau}^{00}}(x)\, e_{p_1}(x)}{1-\tilde{p}_1(x)}
  + H_3^{00}(x),
    \end{aligned}
\end{equation}
where
\begin{align*}
H_3^{00}(x)
&= \frac{1}{\tilde{\pi}(x)\{1-\tilde{p}_1(x)\}}\,
    e_{\check{\tau}^{00}}(x)\, e_{\pi}(x)\, e_{p_1}(x)
  + \frac{\mathbb{E}\!\left[\tilde{g}^{00}\mid X=x\right]}{1-\tilde{p}_1(x)}\,
    R_3^{00}(x),
\end{align*}
and $R_3^{00}(x)$ is the high–order term in the EIF.

For $u = 10$,
\begin{equation}\label{eq:bias_one_10}
    \begin{aligned}
        \tilde{\tau}^{10}_{\mathrm{one}}(x) - \tau^{10}(x)
&= \frac{1}{1-\tilde{\pi}(x)}\, e_{\pi}(x)\, e_{\mu_{11}}(x)
  + \frac{1}{1-\tilde{\pi}(x)}\, e_{\pi}(x)\, e_{\mu_{00}}(x) \\[4pt]
&\quad - \frac{\tilde{p}_0(x)}
        {\tilde{p}_1(x)\{\tilde{p}_1(x)-\tilde{p}_0(x)\}}\,
        e_{p_1}(x)\, e_{\mu_{11}}(x) \\[4pt]
&\quad - \frac{1-p_1(x)}
        {\{1-\tilde{p}_0(x)\}\{\tilde{p}_1(x)-\tilde{p}_0(x)\}}\,
        e_{p_0}(x)\, e_{\mu_{00}}(x) \\[4pt]
&\quad + \frac{1}{\tilde{p}_1(x)-\tilde{p}_0(x)}\,
        e_{p_0}(x)\, e_{\mu_{11}}(x)
      + \frac{1}{\tilde{p}_1(x)-\tilde{p}_0(x)}\,
        e_{p_1}(x)\, e_{\mu_{00}}(x) \\[4pt]
&\quad + \frac{e_{\check{\tau}^{10}}(x)\,\bigl(e_{p_1}(x)-e_{p_0}(x)\bigr)}
           {\tilde{p}_1(x)-\tilde{p}_0(x)}
  + H_3^{10}(x),
    \end{aligned}
\end{equation}
where
\begin{align*}
H_3^{10}(x)
&= - \frac{1}{\tilde{\pi}(x)\{\tilde{p}_1(x)-\tilde{p}_0(x)\}}\,
      e_{\check{\tau}^{10}}(x)\, e_{p_1}(x)\, e_{\pi}(x) \\[4pt]
&\quad - \frac{1}{\{1-\tilde{\pi}(x)\}\{\tilde{p}_1(x)-\tilde{p}_0(x)\}}\,
      e_{\check{\tau}^{10}}(x)\, e_{p_0}(x)\, e_{\pi}(x) \\[4pt]
&\quad + \frac{\mathbb{E}\!\left[\tilde{g}^{10}\mid X=x\right]}
           {\tilde{p}_1(x)-\tilde{p}_0(x)}\,
      R_3^{10}(x).
\end{align*}

For $u = 11$,
\begin{equation}\label{eq:bias_one_11}
    \begin{aligned}
        \tilde{\tau}^{11}_{\mathrm{one}}(x) - \tau^{11}(x)
&= \frac{1}{\tilde{p}_1(x)}\, e_{p_1}(x)\, e_{\mu_{11}}(x)
  - \frac{1}{\tilde{p}_0(x)}\, e_{p_0}(x)\, e_{\mu_{11}}(x) \\[4pt]
&\quad + \frac{1}{\tilde{\pi}(x)}\, e_{\pi}(x)\, e_{\mu_{11}}(x)
  + \frac{1}{1-\tilde{\pi}(x)}\, e_{\pi}(x)\, e_{\mu_{01}}(x) \\[4pt]
&\quad + \frac{e_{\check{\tau}^{11}}(x)\, e_{p_0}(x)}{\tilde{p}_0(x)}
  + H_3^{11}(x),
    \end{aligned}
\end{equation}
where
\begin{align*}
H_3^{11}(x)
&= \frac{1}{\tilde{p}_0(x)\{1-\tilde{\pi}(x)\}}\,
      e_{\check{\tau}^{11}}(x)\, e_{\pi}(x)\, e_{p_0}(x)
  + \frac{\mathbb{E}\!\left[\tilde{g}^{11}\mid X=x\right]}
         {\tilde{p}_0(x)}\,
    R_3^{11}(x).
\end{align*}

Now, let us show how to derive it.
\begin{align*}
\tilde\tau^{u}_{\mathrm{one}}(x) - \tau^{u}(x)
&:= \mathbb{E}\!\left[
      \check\tau^{u}(X)
      + \frac{\tilde\phi_{1,u} - \tilde\phi_{0,u} - \check\tau^{u}(X)\,\tilde g^{u}}
             {\tilde e^{u}(X)}
      \,\Big|\, X = x
    \right]
    - \tau^{u}(x).
\end{align*}

Recall that
\begin{align*}
\tau^{u}(x)
= \mathbb{E}\!\left[
      \check\tau^{u}(X)
      + \frac{\phi_{1,u} - \phi_{0,u} - \check\tau^{u}(X)\,g^{u}}
             {e^{u}(X)}
      \,\Big|\, X = x
    \right].
\end{align*}

So we have
\begin{align*}
\tilde\tau^{u}_{\mathrm{one}}(x) - \tau^{u}(x)
&= \check\tau^{u}(x)
   + \frac{\mathbb{E}\!\left[\tilde\phi_{1,n}-\tilde\phi_{0,n}\mid X=x\right]
          - \check\tau^{u}(x)\,\mathbb{E}\!\left[\tilde g^{u}\mid X=x\right]}
          {\tilde e^{u}(x)}  \\
&\quad
   - \check\tau^{u}(x)
   + \frac{\mathbb{E}\!\left[\phi_{1,n}-\phi_{0,n}\mid X=x\right]
          - \check\tau^{u}(x)\,\mathbb{E}\!\left[g^{u}\mid X=x\right]}
          {e^{u}(x)} .
\end{align*}

Apply the add--and--subtract trick. The right--hand side becomes
\begin{align*}
\mathrm{RHS}
&=
\underbrace{
\frac{\mathbb{E}\!\left[\tilde\phi_{1,u}-\tilde\phi_{0,u} - (\phi_{1,u}-\phi_{0,u})\mid X=x\right]
-
      \check\tau^{u}(x)\,\mathbb{E}\!\left[\tilde g^u -g^{u}\mid X=x\right]
}{\tilde e^{u}(x)}
}_{\text{(I)}} \\[4pt]
&\quad
+
\underbrace{
\Bigl\{\mathbb{E}\!\left[\phi_{1,u}-\phi_{0,u}\mid X=x\right]
      - \check\tau^{u}(x)\,\mathbb{E}\!\left[g^{u}\mid X=x\right]\Bigr\}
\left(\frac{1}{\tilde e^{u}(x)}-\frac{1}{e^{u}(x)}\right)
}_{\text{(II)}} .
\end{align*}

For the first part (I), recall we already calculated
\[
\mathbb{E}\!\left[\tilde\phi_{1,u}-\tilde\phi_{0,u}\mid X=x\right]
 - \mathbb{E}\!\left[\phi_{1,u}-\phi_{0,u}\mid X=x\right],
\qquad
\mathbb{E}\!\left[\tilde g^{u}-g^{u}\mid X=x\right]
\]
in the EIF section.

For the second part (II), recall the identification
\[
\mathbb{E}\!\left[\phi_{1,u}-\phi_{0,u}\mid X=x\right]
= \tau^{u}(x)\,\mathbb{E}\!\left[g^{u}\mid X=x\right]
= \tau^{u}(x)\,e^{u}(x).
\]

Hence
\begin{align*}
\mathrm{II}
&= \frac{\bigl(\tau^{u}(x)-\check\tau^{u}(x)\bigr)e^{u}(x)}{\tilde e^{u}(x)}
   - \bigl(\tau^{u}(x)-\check\tau^{u}(x)\bigr) \\[4pt]
&= \frac{e_{\check\tau^{u}}(x)\,\bigl(\tilde e^{u}(x)-e^{u}(x)\bigr)}
         {\tilde e^{u}(x)}.
\end{align*}
The second equality is a simply algebra. For example, for $u =00$, $\mathrm{II}= - \frac{e_{\check \tau^{00}}(x)e_{p_1}(x)}{1 - \tilde p_1(x)}$, where $e_{p_1}(x): = \tilde p_1(x) - p_1(x)$.

Combining the results obtained above for $u=00$, we have
\begin{align*}
\tilde \tau^{00}_{\mathrm{one}}(x) - \tau^{00}(x)
&
  = \frac{1}{\tilde{\pi}(x)}\, e_{\pi}(x)\, e_{\mu_{10}}(x)
  + \frac{1}{1-\tilde{\pi}(x)}\, e_{\pi}(x)\, e_{\mu_{00}}(x) \\[4pt]
&\quad
  - \frac{1}{1-\tilde{p}_1(x)}\, e_{p_1}(x)\, e_{\mu_{00}}(x)
  + \frac{1}{1-\tilde{p}_0(x)}\, e_{\mu_{00}}(x)\, e_{p_0}(x) \\[4pt]
&\quad
  - \frac{e_{\check{\tau}^{00}}(x)\, e_{p_1}(x)}{1-\tilde{p}_1(x)}
  + H_3^{00}(x),
\end{align*}
where
\begin{align*}
H_3^{00}(x)
&= \frac{1}{\tilde{\pi}(x)\{1-\tilde{p}_1(x)\}}\,
    e_{\check{\tau}^{00}}(x)\, e_{\pi}(x)\, e_{p_1}(x)
  + \frac{\mathbb{E}\!\left[\tilde{g}^{00}\mid X=x\right]}{1-\tilde{p}_1(x)}\,
    R_3^{00}(x).
\end{align*}
\end{enumerate}

\subsection{Proof of Theorem \ref{thm:smooth_sub_one}}\label{s:proof9}
We reviewed the results from the \cite{kennedy2023towards} in the new notation and introduce it as the lemma first. Define notation $W = (X,Z,S,Y)$.

\begin{lemma}\label{lemmaken}\citep{kennedy2023towards}
    Define  $m_f(x)=\mathbb{E}\{f(W) \mid X=x\}$ the conditional expectation of $f(W)$ given $X$, $\widehat{m}_{\tilde{f}}(x)=\widehat{\mathbb{E}}_n\{\tilde{f}(W) \mid X=x\}$ the regression of $\tilde{f}(W)$ on $X$ in the test samples, $\widehat{m}_f(x)=\widehat{\mathbb{E}}_n\{f(W) \mid X=x\}$ the corresponding regression of $f(W)$ on $X$, and the oracle mean squared error $R_n^*(x)^2=\mathbb{E}\left[\{\widehat{m}_f(x)-m_f(x)\}^2\right]$. If:
1. the regression estimator $\widehat{\mathbb{E}}_n$ is stable with respect to distance $d$, and
2. $d(\tilde{f}, f) \xrightarrow{p} 0$,
then
\begin{equation*}
    \hat m_{\tilde f}(x)-m_f(x)
= \hat m_f(x)-m_f(x)
+ \widehat{\mathbb{E}}_n\{\tilde b(X)\mid X=x\}
+ o_p\!\big(R_n^*(x)\big),
\end{equation*}
where $\big(R_n^*(x)\big)^2 := \mathbb{E}\!\big[(\hat m_f(x)-m_f(x))^2\big]$, and $\tilde{b}(x)=\mathbb{E}\left\{\tilde{f}(W)-f(W) \mid D^t,  X=x\right\}$, $D^t$ is the training sample.
\end{lemma}

Hence for the subset and one-step estimator, it is quite straightforward to apply this lemma \ref{lemmaken}. Here $\hat{\tau}_{sub}^u(x)$ or $\hat{\tau}_{one}^u(x)$ is $\hat m_{\tilde f}(x)$, $\tau^u(x)$ is $m_f(x)$. $\tilde b(x)$ is the bias of plug-in limit parameter $\tilde \tau^u_{one}(x) - \tau^u(x)$ or $\tilde \tau^u_{sub}(x) - \tau^u(x)$.

\subsection{Error for the EIF estimator}\label{s:error_eif}
We define the EIF-based estimator using the regression operator
$\widehat{\mathbb{E}}_n$ as
\[
\widehat{\tau}^{u}_{\mathrm{eif}}(x)
= \widehat{\mathbb{E}}_n\!\left[
\frac{\tilde{\phi}_{1,u}(W) - \tilde{\phi}_{0,u}(W)}
     {\hat m_{\widetilde g^u}(X)}
\,\Big|\, X=x \right],
\]
where $\hat m_{\widetilde g^u}(x)
= \widehat{\mathbb{E}}_n\!\left[\widetilde g^u(W) \mid X=x\right]$ denotes the estimated plug-in denominator function. In general, the
regression method used to estimate $\hat m_{\widetilde g^u}(x)$ does not need to
coincide with the second-stage regression method used to estimate
$\widehat{\tau}^{u}_{\mathrm{eif}}(x)$. For
simplicity, however, we focus our theoretical analysis on the case where the
same regression operator $\widehat{\mathbb{E}}_n$ is used in both steps. Recall that for the EIF estimator, we use a three-way sample splitting procedure as described above.

For the following theorem, define $\tilde{\xi}_{u}
:= \frac{\tilde\phi_{1,u}(W)-\tilde \phi_{0,u}(W)}{\hat m_{\widetilde g^u}(X)}$, $\xi_{u}
:= \frac{\phi_{1,u}(W)-\phi_{0,u}(W)}{\hat m_{\widetilde g^u}(X)}$.
Moreover, let $R_{1}^{*}(x)^{2}
:= \mathbb{E}\Big[
\{
\hat{\mathbb{E}}_{n}\!\left[
\xi_{u}\,\big|\, X=x
\right]
-\mathbb{E}\!\left[
\xi_{u}\,\big|\, X=x
\right]\}^{2}
\Big]$, 
and $R_g^*(x)^2$ 
$:=\mathbb{E}\Big[
\{
\hat{\mathbb{E}}_{n}\!(
g_{u}\,\big|\, X=x)
-\mathbb{E}\!(
g_{u}\,\big|\, X=x)
\}^{2}
\Big]$.

\begin{theorem}\label{smooth_eif}
Assume the following conditions hold: (i)
The second-stage regression operator $\widehat{\mathbb{E}}_{n}$ is a linear
smoother with weights $\{w_i(x)\}_{i=1}^n$, i.e.,
$\widehat{\mathbb{E}}_{n}[f(W)\mid X=x]=\sum_{i=1}^n w_i(x) f(W_i)$, and satisfies
the weight regularity $\sum_{i=1}^n |w_i(x)| = O_p(1)$, and localization conditions  $\sum_{i=1}^n |w_i(x)|\,\mathbf{1}\{|X_i-x|>\delta\}\xrightarrow{p}0
\quad\text{for each }\delta>0$. (ii)
$d(\tilde{\xi}_{u},\xi_{u})\xrightarrow{p}0$ and $d(\tilde g^{u}, g^{u})\xrightarrow{p}0$. (iii) overlap: There exists $\epsilon>0$ such that almost surely, $\epsilon \le \pi(X), \tilde\pi(X) \le 1-\epsilon$, $\tilde p_1(x), p_1(x)\ge \epsilon,\,p_0(x), \tilde p_0(x) \le 1-\epsilon$, $p_1(x)-p_0(x), \tilde p_1(x)-\tilde p_0(x)\ge \epsilon$, for all $x$.
(iv)
The function $x\mapsto \tau^{u}(x)/\hat m_{\widetilde g^{u}}(x)$ is continuous and
uniformly bounded by a constant $M<\infty$. 
Then $\hat{\tau}^u_{\mathrm{eif}}(x) - \tau^u(x)$ equals
\begin{equation*}
\begin{aligned}
&\underbrace{\hat{\mathbb{E}}_{n}\!\left[
    \xi_u
    \,\Big|\, X = x
  \right]
  -
  \mathbb{E}\!\left[
   \xi_u
    \,\Big|\, X = x
  \right]
  - \frac{\tau^u(x)}{\hat{m}_{\tilde g^u}(x)}
    \Big(
      \hat{\mathbb{E}}_{n}[g^u \mid X =x]
      - \mathbb{E}[g^u\mid X =x]
    \Big)}_{\text{oracle error}}\\
    &+ \underbrace{\hat{\mathbb{E}}_n\!\Biggl[
     \frac{\mathbb{E}\!\left[\tilde g^{u}\mid X\right]}{\hat m_{\tilde g^{u}}(X)}
     \bigl(\tilde\tau^{u}_{\mathrm{eif}}(X)-\tau^{u}(X)\bigr)
     \,\Bigm|\,X=x\Biggr]}_{\text{smooth weighted plug-in bias}}
+ o_p\big(R_1^*(x)+R_g^*(x)\big).
\end{aligned}
\end{equation*}
\end{theorem}

Although the assumptions in Theorem~\ref{smooth_eif} differ from those used for the subset and one-step results, they are overall mild and largely technical, primarily ensuring the same oracle error structure and smoothed-bias decomposition. Even though the assumptions required for linear smoothers are stronger than the stability condition, they are still standard for kernel regression, local polynomial regression, many series/sieve methods, and certain tree-based procedures that admit a linear-smoother representation. We conjecture that it may be possible to extend our results to the more general stability assumption. The third assumption is a classical overlap assumption. Assumption (iv) is a mild regularity condition. It is satisfied, for example, when $\tau^{u}(\cdot)$ is continuous and bounded and the denominator regression $\hat m_{\widetilde g^{u}}(\cdot)$ is continuous and bounded away from zero with probability tending to one, which can be enforced via standard truncation of $\hat m_{\widetilde g^{u}}$.

\subsection{Proof of Theorem \ref{smooth_eif}}\label{s:proof10}

Like the EIF plug-in limit parameter bias, the EIF smooth estimator can also be decomposed into the numerator bias and denominator bias. 

\begin{align*}
\hat{\tau}_{\mathrm{eif}}^{u}(x) - \tau^{u}(x)
&= \hat{\mathbb{E}}_{n}\!\left[
      \frac{\tilde{\phi}_{1,u}-\tilde{\phi}_{0,u}}{\hat m_{\tilde g^{u}}(X)}
      \,\Big|\, X = x
    \right]
   - \mathbb{E}\!\left[
      \frac{\phi_{1,u}-\phi_{0,u}}{\hat m_{\tilde g^{u}}(X)}
      \,\Big|\, X = x
    \right] \quad (\mathrm{I}) \\[6pt]
&\quad
   + \mathbb{E}\!\left[
      \frac{\phi_{1,u}-\phi_{0,u}}{\hat m_{\tilde g^{u}}(X)}
      \,\Big|\, X = x
    \right]
   - \tau^{u}(x) \quad (\mathrm{II}) .
\end{align*}

Here the first part (I) is the point wise error of numerator. The second part (II) is the point wise error of denominator. We focus on the second part (II) first. Recall $\mathbb{E}[\phi_{1,u}-\phi_{0,u}\mid X] = \tau^u(X) m_{g^u}(X)$
\begin{equation*}
        II = \frac{\tau^u(X) (m_{g^u}(x)-\hat m_{\tilde g^{u}}(x))}{\hat m_{\tilde g^{u}}(x)} 
\end{equation*}
Recall from \cite{kennedy2023towards}, we have:
\begin{equation*}
    \begin{aligned}
        &\hat m_{\tilde g^{u}}(x) - m_{\tilde g^{u}}(x) := \hat{\mathbb{E}}_n [\tilde g^u \mid X = x] - \mathbb{E}[g^u \mid X = x]\\
        & = \hat{\mathbb{E}}_n [g^u \mid X = x] - \mathbb{E}[g^u \mid X = x] + \hat{\mathbb{E}}_n[\tilde b_g(X) \mid X = x] + o_p(R_g^*(x))
    \end{aligned}
\end{equation*}
where $\tilde b_{g}(x) := \mathbb{E}[\tilde g^u - g^u \mid X =x]$ and $(R_g^*(x))^2 := \mathbb{E}\!\big[(\hat{\mathbb{E}}_n [g^u \mid X = x] - \mathbb{E}[g^u \mid X = x])^2\big]$

Throughout the analysis we assume a sample–splitting scheme such that the
data used to estimate the denominator $\hat m_{\tilde g^u}(X)$ is independent
of the data used to regress the pseudo–outcome.  Hence we can apply the \cite{kennedy2023towards} results.
For the first part (I), 
\begin{align*}
\hat{\mathbb{E}}_{n}\!\left[
  \frac{\tilde\phi_{1,u}-\tilde\phi_{0,u}}{\hat m_{\tilde g^u}(X)}
  \,\Big|\, X = x
\right]
&-
\mathbb{E}\!\left[
  \frac{\phi_{1,u}-\phi_{0,u}}{\hat m_{\tilde g^u}(X)}
  \,\Big|\, X = x
\right]
=
\hat{\mathbb{E}}_{n}\!\left[
  \frac{\phi_{1,u}-\phi_{0,u}}{\hat m_{\tilde g^u}(X)}
  \,\Big|\, X = x
\right]\\
&-
\mathbb{E}\!\left[
  \frac{\phi_{1,u}-\phi_{0,u}}{\hat m_{\tilde g^u}(X)}
  \,\Big|\, X = x
\right]  
+ \hat{\mathbb{E}}_{n}\!\left[\tilde b_1(X)\,\big|\, X = x\right]
+ o_p(R_1^*(x)).
\end{align*}
where
\[
\tilde b_1(x)
:= \mathbb{E}\!\left[
    \frac{\tilde \phi_{1,u}-\tilde \phi_{0,u}}{\hat m_{\tilde g^u}(X)}
    - \frac{\phi_{1,u}- \phi_{0,u}}{\hat m_{\tilde g^u}(X)}
    \,\Bigg|\, X = x
  \right]
\]
and
\[
R_1^*(x)^2
:= \mathbb{E}\Bigg[
    \Big(
      \hat{\mathbb{E}}_{n}\!\left[
        \frac{\phi_{1,u}-\phi_{0,u}}{\hat m_{\tilde g^u}(X)}
        \,\Bigg|\, X = x
      \right]
      -
      \mathbb{E}\!\left[
        \frac{\phi_{1,u}-\phi_{0,u}}{\hat m_{\tilde g^u}(X)}
        \,\Bigg|\, X = x
      \right]
    \Big)^2
  \Bigg].
\]
Combining together, we have
\begin{equation*}
\begin{aligned}
      \hat{\tau}^u_{eif}(x) - \tau^u(x) &= \hat{\mathbb{E}}_{n}\!\left[
  \frac{\phi_{1,u}-\phi_{0,u}}{\hat m_{\tilde g^u}(X)}
  \,\Big|\, X = x
\right]-
\mathbb{E}\!\left[
  \frac{\phi_{1,u}-\phi_{0,u}}{\hat m_{\tilde g^u}(X)}
  \,\Big|\, X = x
\right]   \\
& - \frac{\tau^u(x)}{\hat{m}_{\tilde g^u}(x)}\left(\hat{\mathbb{E}}\left[g^u \mid X =x\right] - \mathbb{E}\left[g^u\mid X =x\right]\right)\\
& +\hat{\mathbb{E}}_n\!\left[\tilde b_1^{u}(X)\mid X=x\right]
  - \frac{\tau^{u}(x)}{\hat m_{\tilde g^{u}}(x)}
    \hat{\mathbb{E}}_n\!\left[\tilde b_g^{u}(X)\mid X=x\right] + o_p(R_1^*(x)+R_g^*(x))\\
\end{aligned}
\end{equation*}
Now we focus on the last line (bias part). We want to prove 
\begin{align*}
    \hat{\mathbb{E}}_n\!\left[\tilde b_1^{u}(X)\mid X=x\right]
  - \frac{\tau^{u}(x)}{\hat m_{\tilde g^{u}}(x)}
    &\hat{\mathbb{E}}_n\!\left[\tilde b_g^{u}(X)\mid X=x\right]=\\ 
    &\hat{\mathbb{E}}_n\!\Biggl[
     \frac{\mathbb{E}\!\left[\tilde g^{u}\mid X\right]}{\hat m_{\tilde g^{u}}(X)}
     \bigl(\tilde\tau^{u}_{\mathrm{eif}}(X)-\tau^{u}(X)\bigr)
     \,\Bigm|\,X=x\Biggr] 
\end{align*}
We use $u =00$ as an illustrative example to prove the second equation. From previous analysis, 
\begin{align*}
\tilde b_1^{00}(x)
&= - \frac{\tau^u(x)}{\hat m_{g^{00}}(x)\,\tilde\pi(x)}\,e_\pi(x)e_{p_1}(x)
  + \frac{\mathbb{E}\!\left[\tilde g^{00}\mid X=x\right]}{\hat m_{\tilde g^{00}}(x)}
    \bigl(\tilde\tau^{00}_{\mathrm{eif}}(x)-\tau^{00}(x)\bigr).
\end{align*} 
Then, combining together, since $\hat{\mathbb{E}}_n$ is a linear smoother, we have
\begin{align*}
&\hat{\mathbb{E}}_n\!\left[\tilde b_1^{00}(X)\mid X=x\right]
  - \frac{\tau^{00}(x)}{\hat m_{\tilde g^{00}}(x)}
    \hat{\mathbb{E}}_n\!\left[\tilde b_g^{00}(X)\mid X=x\right] \\[4pt]
&= \hat{\mathbb{E}}_n\!\Biggl[
     \frac{\mathbb{E}\!\left[\tilde g^{00}\mid X\right]}{\hat m_{\tilde g^{00}}(X)}
     \bigl(\tilde\tau^{00}_{\mathrm{eif}}(X)-\tau^{00}(X)\bigr)
     \,\Bigm|\,X=x\Biggr] \\[4pt]
&\qquad
  + \underbrace{
      \frac{\tau^{00}(x)}{\hat m_{\tilde g^{00}}(x)}
      \hat{\mathbb{E}}_n\!\Bigl[\frac{1}{\tilde\pi(X)}\,e_\pi(X)e_{p_1}(X)\,\Bigm|\,X=x\Bigr]
      - \hat{\mathbb{E}}_n\!\Bigl[
          \frac{\tau^{00}(X)}{\hat m_{\tilde g^{00}}(X)\,\tilde\pi(X)}\,e_\pi(X)e_{p_1}(X)
          \,\Bigm|\,X=x\Bigr]
    }_{(*)}.
\end{align*}
If we can show $(*)\xrightarrow{p} 0$, then we are done. 
We introduced the following lemma first. 
\begin{lemma}[Vanishing linear–smoother remainder term]\label{le_smooth}
Fix a point $x \in \mathbb{R}^d$. Let $(X_i)_{i=1}^n$ be a sample and consider
the linear smoother term
\[
  T_n(x) := \sum_{i=1}^n w_i(x)\,\{c(X_i) - c(x)\}\, s(X_i).
\]

Assume:
\begin{enumerate}
  \item \textbf{Weight regularity.} There exists a constant $C < \infty$ such
  that, for all $n$,
  \begin{equation}
    \sum_{i=1}^n |w_i(x)| \le C
    \quad\text{a.s.} 
  \end{equation}

  \item \textbf{Localization.} For every $\delta > 0$,
  \begin{equation}
    \sum_{i=1}^n |w_i(x)| \,\mathbf{1}\{|X_i - x| > \delta\}
      \;\xrightarrow{p}\; 0. 
  \end{equation}

  \item \textbf{Boundedness.}
  The function $a(\cdot)$ is continuous at $x$ and bounded:
  $|a(u)| \le M_a$ for all $u$.
  The function $b(\cdot)$ is bounded: $|b(u)| \le M_b$ for all $u$.
\end{enumerate}
Then
\[
  T_n(x) \xrightarrow{p} 0.
\]
\end{lemma}

Therefore, Lemma~\ref{le_smooth} applies to the starred term. Here $c(x) = \frac{\tau^{00}(x)}{\hat m_{\tilde g^{00}}(x)}$ and $s(x) = \frac{1}{\tilde\pi(x)}\,e_\pi(x)e_{p_1}(x)$. The (*) will converge to the 0 in probability. 


Last, we will show the proof for the Lemma \ref{le_smooth}.
\begin{proof}
We have
\[
|T_n(x)|
\le \sum_{i=1}^n |w_i(x)|\,|c(X_i)-c(x)|\,|s(X_i)|.
\]
Fix $\varepsilon>0$. We will show $\mathbb{P}(|T_n(x)|>\varepsilon)\to 0$.

For $\delta>0$, define
\[
\begin{aligned}
C_n(\delta)
&:= \sum_{i=1}^n |w_i(x)|\,|c(X_i)-c(x)|\,|s(X_i)|\,
   \mathbf{1}\{|X_i-x|\le\delta\},\\[2mm]
S_n(\delta)
&:= \sum_{i=1}^n |w_i(x)|\,|c(X_i)-c(x)|\,|s(X_i)|\,
   \mathbf{1}\{|X_i-x|>\delta\}.
\end{aligned}
\]
Then
\[
|T_n(x)| \le C_n(\delta) + S_n(\delta).
\]

\noindent By continuity of $a(\cdot)$ at $x$, for any $\eta>0$ there exists $\delta>0$ such that
\[
|u-x|\le\delta \quad\Rightarrow\quad |c(u)-c(x)|\le\eta.
\]
For this $\delta$, using $|s(X_i)|\le M_b$,
\[
C_n(\delta)
\le \sum_{i=1}^n |w_i(x)|\,\eta\,M_b
= \eta M_b \sum_{i=1}^n |w_i(x)|
\le \eta M_b C.
\] Given our fixed $\varepsilon>0$, choose $\eta>0$ small enough that
$\eta M_b C < \varepsilon/2$. Then for this $\delta$,
\begin{equation}
C_n(\delta) \le \varepsilon/2
\quad\text{a.s. for all $n$}.
\label{eq:An-bound}
\end{equation}

\noindent For any $u$,
\[
|c(u)-c(x)| \le |c(u)| + |c(x)| \le 2M_a,
\]
and $|s(u)|\le M_b$, so
\[
|c(u)-c(x)|\,|s(u)| \le 2M_a M_b =: K.
\]
Therefore,
\[
\begin{aligned}
S_n(\delta)
&\le \sum_{i=1}^n |w_i(x)|\,K\,\mathbf{1}\{|X_i-x|>\delta\} \\
&= K \sum_{i=1}^n |w_i(x)|\,\mathbf{1}\{|X_i-x|>\delta\}.
\end{aligned}
\]
By the localization assumption,
\[
\sum_{i=1}^n |w_i(x)|\,\mathbf{1}\{|X_i-x|>\delta\}
\xrightarrow{p} 0,
\]
and hence $S_n(\delta)\xrightarrow{p}0$.
From \eqref{eq:An-bound} and the convergence $S_n(\delta)\xrightarrow{p}0$, for our chosen $\delta$,
\[
|T_n(x)| 
\le C_n(\delta) + S_n(\delta)
\le \varepsilon/2 + S_n(\delta).
\]
Thus
\[
\mathbb{P}(|T_n(x)|>\varepsilon)
\le \mathbb{P}\big(\varepsilon/2 + S_n(\delta)>\varepsilon\big)
= \mathbb{P}\big(S_n(\delta)>\varepsilon/2\big)
\longrightarrow 0.
\]
Therefore $T_n(x)\xrightarrow{p}0$, as claimed.
\end{proof}
\subsection{Proof of Corollary}\label{s:corollary}
In this section, we prove Corollary~\ref{col:sub_one} and then present and prove the analogous corollary for the EIF estimator. We first introduce the following proposition from \citet{kennedy2023towards}.

\begin{proposition}[{\citet{kennedy2023towards}}]\label{prop:kennedy_product_bound}
Suppose $\widehat b(x)=\widehat b_1(x)\,\widehat b_2(x)$ and the second-stage regression operator
$\widehat{\mathbb E}_n$ is a linear smoother with weights $\{w_i(x;X^n)\}_{i=1}^n$ and the weights satisfy $\sum_{i=1}^n \bigl|w_i(x;X^n)\bigr| = O_{\mathbb P}(1).$
Then
\[
\widehat{\mathbb E}_n\!\left\{\widehat b(X)\mid X=x\right\}
=
O_{\mathbb P}\!\left(\,\|\widehat b_1\|_{w,p}\,\|\widehat b_2\|_{w,q}
\right),
\]
where $1/p+1/q=1$ with $p,q>1$, and $\|f\|_{w,p}
:=
\left[
\sum_{i=1}^n
\left\{
\frac{|w_i(x;X^n)|}{\sum_{j=1}^n |w_j(x;X^n)|}
\right\}
|f(X_i)|^p
\right]^{1/p}.$
\end{proposition}

We now turn to the proof of Corollary~\ref{col:sub_one}, beginning with the subset estimator.

From \eqref{eq:bias_sub10}, \eqref{eq:bias_sub00}, \eqref{eq:bias_sub11}, we can find that the plug-in bias can be written as products of two errors. For the subset estimator, the second-stage regression is fit on the subset $\mathcal S_u\subseteq\{1,\ldots,n\}$ with sample size $n_u:=|\mathcal S_u|$. Under the overlap condition, there exists a constant $c>0$ such that
\[
\mathbb{P}\!\left(\frac{n_u}{n}\ge c\right)\to 1 \quad \text{as } n\to\infty,
\]
so the subset sample size remains a non-vanishing fraction of $n$.
Hence, the rate result follows:
$$
\begin{aligned}
\widehat{\tau}^u_{sub}(x)-\tau^u(x) & =\widehat{\mathbb{E}}_n[\tilde \tau^u_{sub}(X) - \tau^u(X) \mid \mathcal{S}_u, X =x]+\widehat{\mathbb{E}}_n\!\bigl\{\varphi_{\tau^u}(W)\mid\mathcal{S}_u,\,X=x\bigr\}-\tau^u(x)\\
& =O_{\mathbb{P}}\left(n^{-\left(\frac{1}{2+p / \alpha_u}+\frac{1}{2+p / \alpha_\mu}\right)}\right)+\widehat{\mathbb{E}}_n\!\bigl\{\varphi_{\tau^u}(W)\mid\mathcal{S}_u,\,X=x\bigr\}-\tau^u(x)\\
& =O_{\mathbb{P}}\left(n^{-\left(\frac{1}{2+p / \alpha_u}+\frac{1}{2+p / \alpha_\mu}\right)}\right))+O_{\mathbb{P}}\left(n^{\frac{-1}{2+p / \gamma}}\right)
\end{aligned}
$$

\noindent where the second equality follows from Proposition \ref{prop:kennedy_product_bound} together with Assumptions in Corollary \ref{col:sub_one}, and the third since $\widehat{\mathbb{E}}_n$ is minimax optimal and the CATE is $\gamma$-smooth. For the oracle efficiency condition, note that

$$
\begin{aligned}
& n^{-\left(\frac{1}{2+p / \alpha_u}+\frac{1}{2+p / \alpha_\mu}\right)} \leq n^{-\frac{1}{2+p / \gamma}} \\
& \Longleftrightarrow \frac{1}{2+p / \alpha_u}+\frac{1}{2+p / \alpha_\mu} \geq \frac{1}{2+p / \gamma} \\
& \Longleftrightarrow \alpha_u \alpha_\mu \geq 
\frac{p^2}{4+\bigl(4+p / \alpha_u+p / \alpha_\mu\bigr) p / \gamma}
=
\frac{p^2 / 4}{1+\bigl(1+p / (4\alpha_u)+p / (4\alpha_\mu)\bigr) p / \gamma}.
\end{aligned}
$$
which yields the result.

From \eqref{eq:bias_one_00}, \eqref{eq:bias_one_10}, \eqref{eq:bias_one_11}, we can find that the plug-in bias can be written as products of two errors. 
Hence, the rate result follows:
$$
\begin{aligned}
\widehat{\tau}^u_{one}(x)-\tau^u(x) & =\widehat{\mathbb{E}}_n[\tilde \tau^u_{one}(X) - \tau^u(X) \mid  X =x]+\widehat{\mathbb{E}}_n\!\bigl\{\zeta_{\tau^u}(W)\mid\,X=x\bigr\}-\tau^u(x)\\
& =O_{\mathbb{P}}\left(n^{-\left(\frac{1}{2+p/\alpha_\pi}+\frac{1}{2+p/\alpha_{\mu}}\right)}
\;+\;
n^{-\left(\frac{1}{2+p/\alpha_p}+\frac{1}{2+p/\alpha_{\mu}}\right)}
\;+\;
n^{-\left(\frac{1}{2+p/\alpha_{\gamma}}+\frac{1}{2+p/\alpha_{p}}\right)}\right)\\
&\quad+\widehat{\mathbb{E}}_n\!\bigl\{\zeta_{\tau^u}(W)\mid\,X=x\bigr\}-\tau^u(x)\\
& =O_{\mathbb{P}}\left(n^{-\left(\frac{1}{2+p/\alpha_\pi}+\frac{1}{2+p/\alpha_{\mu}}\right)}
\;+\;
n^{-\left(\frac{1}{2+p/\alpha_p}+\frac{1}{2+p/\alpha_{\mu}}\right)}
\;+\;
n^{-\left(\frac{1}{2+p/\alpha_{\gamma}}+\frac{1}{2+p/\alpha_{p}}\right)}\right)\\
&\quad+O_{\mathbb{P}}\left(n^{\frac{-1}{2+p / \gamma}}\right)
\end{aligned}
$$
Moreover, if the preliminary estimator is T-learner:
$$
\widehat{\tau}^u_{one}(x)-\tau^u(x)= O_p\!\Bigg(
n^{\frac{-1}{2+p/\gamma}}
\;+\;
n^{-\left(\frac{1}{2+p/\alpha_\pi}+\frac{1}{2+p/\alpha_{\mu}}\right)}
\;+\;
n^{-\left(\frac{1}{2+p/\alpha_p}+\frac{1}{2+p/\alpha_{\mu}}\right)}
\Bigg).
$$
Denote $r = \max\{\alpha_\pi,\alpha_p\}$, we have 
$$
\widehat{\tau}^u_{one}(x)-\tau^u(x)= O_p\!\Bigg(
n^{\frac{-1}{2+p/\gamma}}
\;+\;
n^{-\left(\frac{1}{2+p/r}+\frac{1}{2+p/\alpha_{\mu}}\right)}
\Bigg),
$$
which is the similar structure as subset estimator and have the same efficient condition as subset estimator.  Now, we present and prove the Corollary for the EIF estimator.

\begin{corollary}\label{col:eif}
Suppose the assumptions of Theorem~\ref{smooth_eif} hold.
Assume further that: (i) The propensity score $\pi$ is $\alpha$-smooth and
$\|\tilde\pi-\pi\|_{w,2}=O_p\!\left(n^{-1/(2+p/\alpha)}\right)$. (ii) The principal score $p_z$ is $\delta$-smooth and
$\|\tilde p_z-p_z\|_{w,2}
=O_p\!\left(n^{-1/(2+p/\delta)}\right)$. (iii) The regressions $\mu_{zs}$ are $\beta$-smooth and
$\|\hat\mu_{zs}-\mu_{zs}\|_{w,2}
=O_p\!\left(n^{-1/(2+p/\beta)}\right)$. (iv) The numerator $\xi_u$ is $\gamma_{\xi}$-smooth and the denominator $g^u$ is
$\gamma_{g}$-smooth. Moreover,
$\mathbb{E}\!\left[\tilde g^{u}\mid X\right]/\hat m_{\tilde g^{u}}(X) = O_p(1)$.
Then
\[
\hat\tau_{\mathrm{eif}}^{u}(x)-\tau^{u}(x)
=
O_p\!\Bigg(
n^{\left(-\frac{1}{2+p / \gamma_{\xi}}+ \frac{1}{2+p / \gamma_g}\right)}
\;+\;
n^{-\left(\frac{1}{2+p/\alpha}+\frac{1}{2+p/\beta}\right)}
\;+\;
n^{-\left(\frac{1}{2+p/\delta}+\frac{1}{2+p/\beta}\right)}
\Bigg).
\]
\end{corollary}

Furthermore, if we assume $\gamma_{\xi} = \gamma_{g}=\gamma$, then the EIF estimator is oracle efficient under the same conditions as one-step estimator.

\begin{proof}
    \begin{equation*}
\begin{aligned}
\hat{\tau}^u_{\mathrm{eif}}(x) &- \tau^u(x) = \hat{\mathbb{E}}_n\!\Biggl[
     \frac{\mathbb{E}\!\left[\tilde g^{u}\mid X\right]}{\hat m_{\tilde g^{u}}(X)}
     \bigl(\tilde\tau^{u}_{\mathrm{eif}}(X)-\tau^{u}(X)\bigr)
     \,\Bigm|\,X=x\Biggr]\\
&+\hat{\mathbb{E}}_{n}\!\left[
    \xi_u
    \,\Big|\, X = x
  \right]
  -
  \mathbb{E}\!\left[
   \xi_u
    \,\Big|\, X = x
  \right]
  - \frac{\tau^u(x)}{\hat{m}_{\tilde g^u}(x)}
    \Big(
      \hat{\mathbb{E}}_{n}[g^u \mid X =x]
      - \mathbb{E}[g^u\mid X =x]
    \Big)\\
& = O_p(n^{-\left(\frac{1}{2+p/\alpha}+\frac{1}{2+p/\beta}\right)}
\;+\;
n^{-\left(\frac{1}{2+p/\delta}+\frac{1}{2+p/\beta}\right)})\\
&+\hat{\mathbb{E}}_{n}\!\left[
    \xi_u
    \,\Big|\, X = x
  \right]
  -
  \mathbb{E}\!\left[
   \xi_u
    \,\Big|\, X = x
  \right]
  - \frac{\tau^u(x)}{\hat{m}_{\tilde g^u}(x)}
    \Big(
      \hat{\mathbb{E}}_{n}[g^u \mid X =x]
      - \mathbb{E}[g^u\mid X =x]
    \Big)\\
&= O_p(n^{-\left(\frac{1}{2+p/\alpha}+\frac{1}{2+p/\beta}\right)}
\;+\;
n^{-\left(\frac{1}{2+p/\delta}+\frac{1}{2+p/\beta}\right)}) + O_p\!(
n^{(-\frac{1}{2+p / \gamma_{\xi}}+ \frac{1}{2+p / \gamma_g})})
\end{aligned}
\end{equation*}
The first equality follows since the weight $\mathbb{E}\left[\tilde g^{u}\mid X\right]/\hat m_{\tilde g^{u}}(X)=O_p(1)$. The second equality holds because the weight $\tau^u(x)/\hat m_{\tilde g^{u}}(x)$ is uniformly bounded by a constant $M$.
\end{proof}

\subsection{Procedure}\label{s:procedure}

For subset and one-step estimator, in the first stage, the nuisance functions are estimated from sample $D_1^n$. In the second stage, these estimates are used to construct an estimate of the pseudo-outcome, which is then regressed on $X$ using the another sample $D^2_n$.

 For EIF estimation, the data are randomly partitioned into three independent subsamples $D_1^n, D_2^n,$ and $D_3^n$. 
 In the first stage, the nuisance functions $\pi(x)$, $p_z(x)$, and $\mu_{zs}(x)$ are estimated using $D_a^n$, 
 which are then combined to construct the auxiliary regression $\tilde g^u(x)$ and the pseudo-outcome difference 
$\tilde\phi_{1,u} - \tilde\phi_{0,u}$.  In the second stage, $\tilde g^u$ is regressed on $X$ using $D_2^n$ to estimate 
 $\widehat{\mathbb{E}}_n[\tilde g^u \mid X]$. 
 In the third stage, the EIF estimators are constructed by regressing, $\frac{\tilde\phi_{1,u} - \tilde\phi_{0,u}}{\widehat{\mathbb{E}}_n[\tilde g^u \mid X]}$
on $X$ using the held-out sample $D_3^n$. The following Figure~\ref{fig:procedure} illustrate this procedure in detail. 

\begin{figure}
    \centering
    \includegraphics[width=1.15\linewidth]{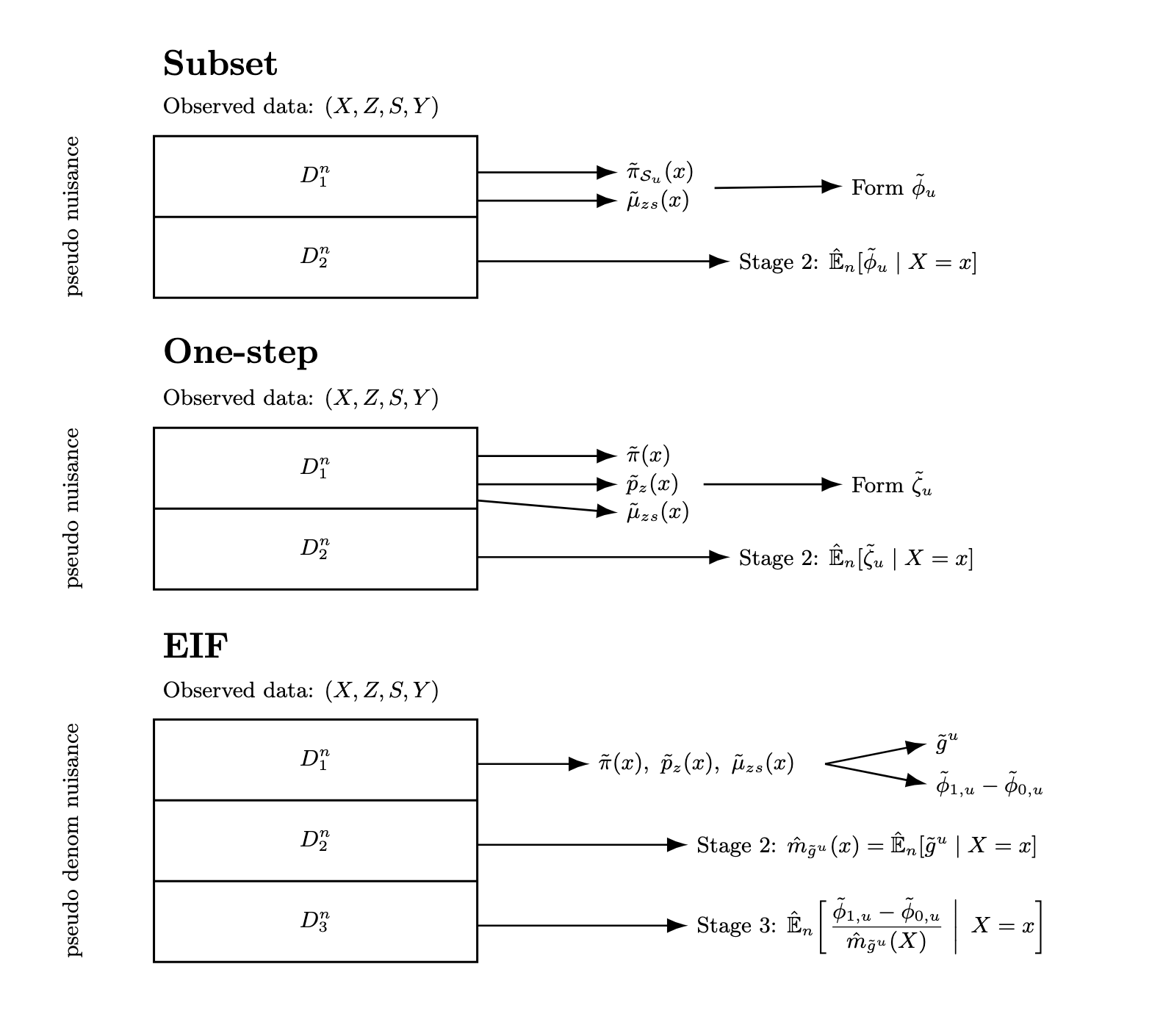}
    \caption{Schematic illustrating the Subset, EIF, and one-step identification}
    \label{fig:procedure}
\end{figure}

\subsection{Simulation}\label{app:sim}

Table~\ref{tab:setting_four_scenarios} summarizes the data-generating process for the first robustness study. 
In this study, we fit parametric working models. Thus, when the underlying true model is \emph{linear}, the corresponding parametric specification is correctly specified; when the underlying true model is \emph{nonlinear}, the parametric specification is misspecified.
Accordingly, the four scenarios correspond to: (i) all models correctly specified, (ii) only the score models $\{\pi(X),p_z(X)\}$ correctly specified, (iii) only the outcome model $b(X)$ correctly specified, and (iv) all models misspecified. Figure \ref{fig:robustness for tau00} and \ref{fig:robustness for tau11} shows the log-RMSE for estimating $\tau^{00}(x)$ and $\tau^{11}(x)$ across sample sizes under four nuisance-specification regimes.

\begin{table}[]
\centering
\rotatebox{90}{%
\makebox[\textwidth]{%
\begin{minipage}{\textwidth}
\centering
\renewcommand{\arraystretch}{1.35}

\begin{tabular*}{\textwidth}{l@{\extracolsep{\fill}}p{0.46\textwidth}p{0.46\textwidth}}
\toprule
& Scenario 1: All correct \; 
& Scenario 2:  $\pi(\cdot)$ and $p_z(\cdot)$ correct only \;  \\
\midrule
$\text{logit}\{\pi(X)\}$
& $-0.4 + 0.4X_1 + 0.4X_2 + 0.4X_3$
& $-0.4 + 0.4X_1 + 0.4X_2 + 0.4X_3$ \\

$\text{logit}\{p_z(X)\}$
& $(2Z - 1)(0.4X_1 - 0.4X_2 + 0.4X_4 + 0.8)$
& $(2Z - 1)(0.4X_1 - 0.4X_2 + 0.4X_4 + 0.8)$ \\

$b(X)$
& $X_1 - 0.4X_2 + X_3 + 0.5X_4$
& $\sin(2\pi X_1X_2) + (X_3 - 0.5)^2 + (X_4 - 0.5)^2$ \\
\midrule
& Scenario 3: $b(\cdot)$ correct only \; 
& Scenario 4: None correct \;  \\
\midrule
$\text{logit}\{\pi(X)\}$
& $0.5\sin(2\pi X_1X_2) + 0.5(X_3 - 0.5)^2$
& $0.5\sin(2\pi X_1X_2) + 0.5(X_3 - 0.5)^2$ \\

$\text{logit}\{p_z(X)\}$
& $(2Z - 1)(0.8 + (X_1-0.5)^2 + 0.6\sin(2\pi X_3X_4))$
& $(2Z - 1)(0.8 + (X_1-0.5)^2 + 0.6\sin(2\pi X_3X_4))$ \\

$b(X)$
& $X_1 - 0.4X_2 + X_3 + 0.5X_4$
& $\sin(2\pi X_1X_2) + (X_3 - 0.5)^2 + (X_4 - 0.5)^2$ \\
\bottomrule
\end{tabular*}

\caption{The data generating process for Study I.}
\label{tab:setting_four_scenarios}
\end{minipage}}}
\end{table}

\begin{figure}
    \centering
    \includegraphics[width=0.9\linewidth]{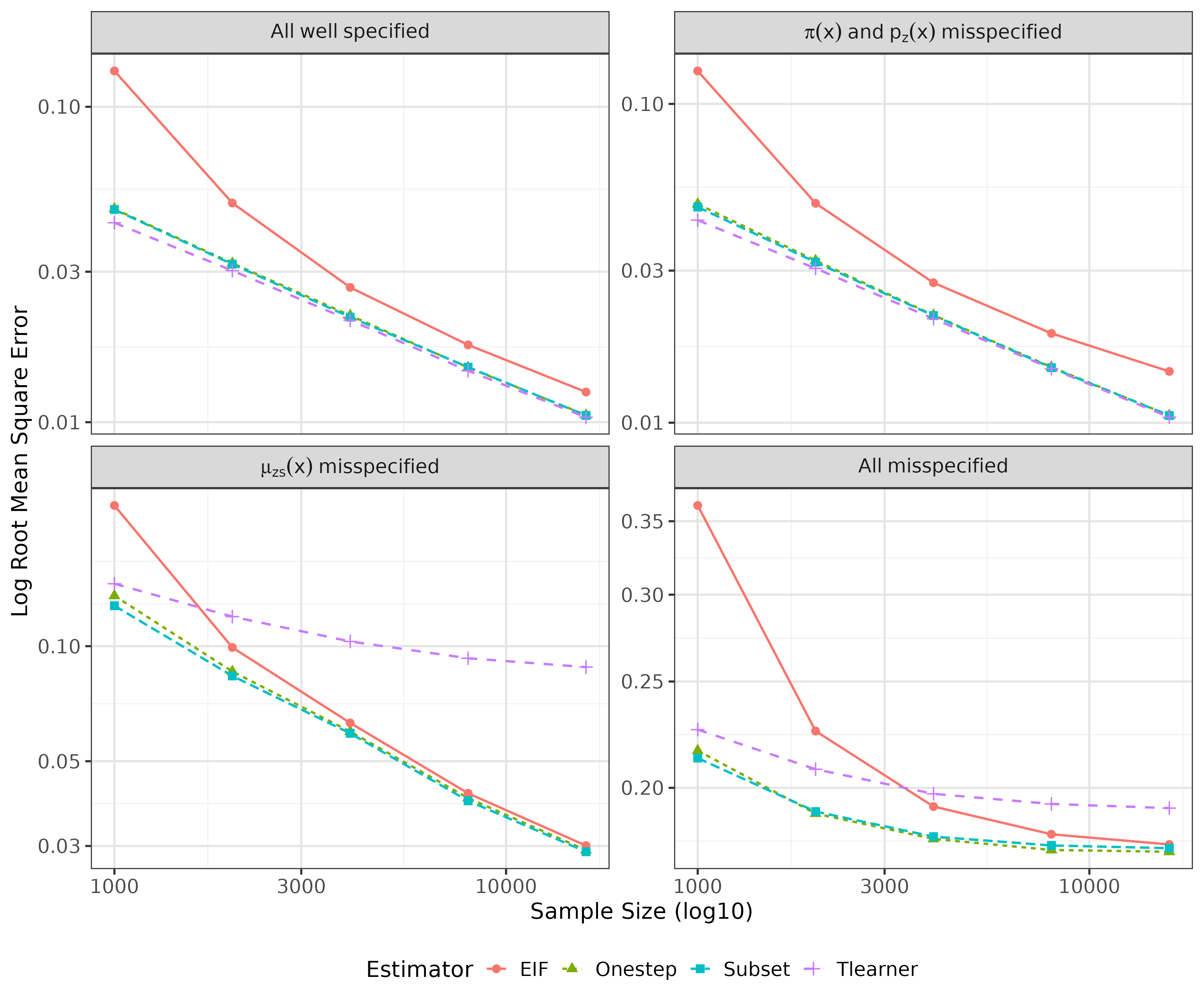}
    \caption{Root mean squared error (RMSE, log scale) of $\tau^{00}(X)$ estimators across sample sizes (1,000–16,000). Results compare the T-learner, subset, EIF, and one-step estimators.}
    \label{fig:robustness for tau00}
\end{figure}

\begin{figure}
    \centering
    \includegraphics[width=0.9\linewidth]{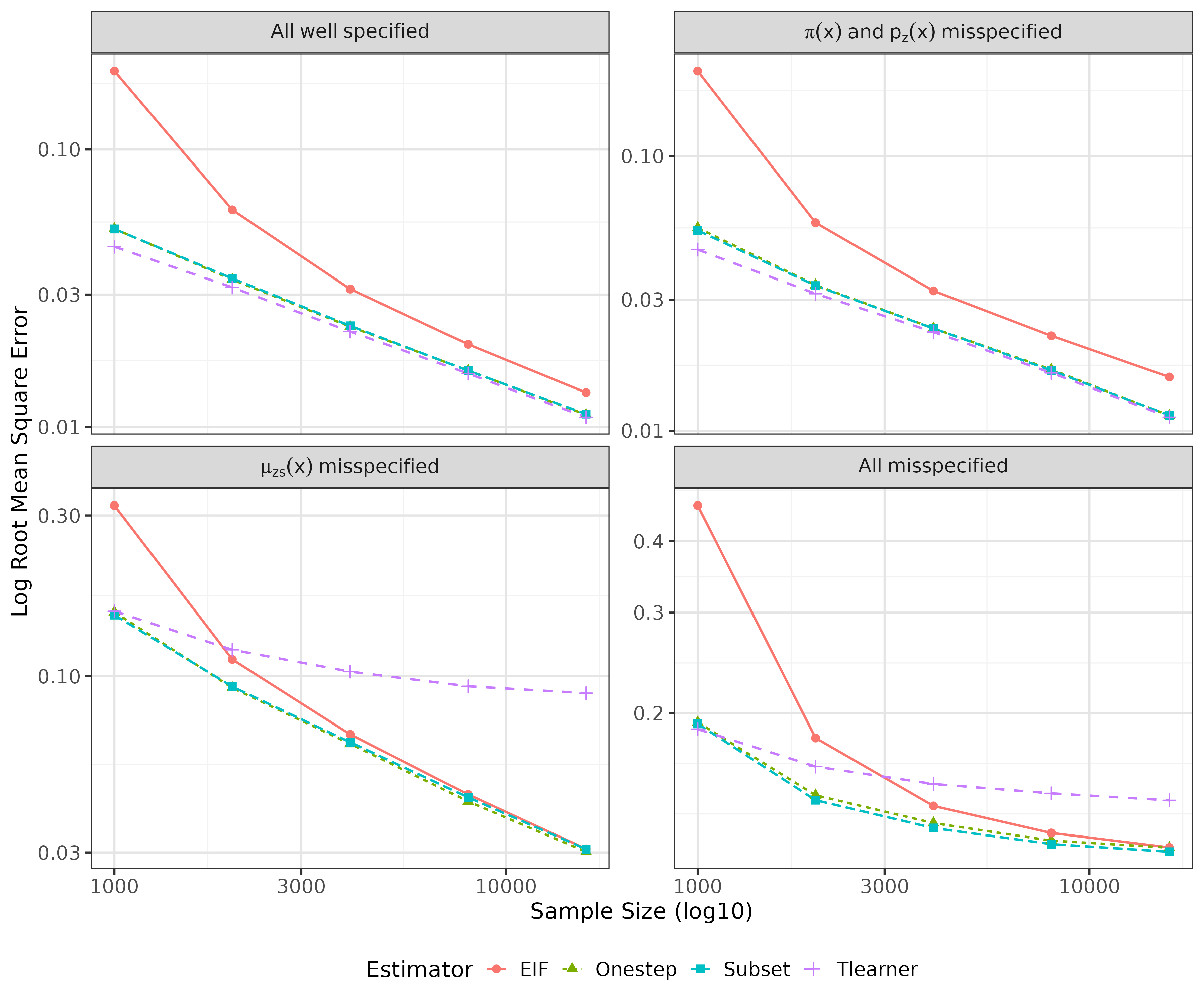}
    \caption{Root mean squared error (RMSE, log scale) of $\tau^{11}(X)$ estimators across sample sizes (1,000–16,000). Results compare the T-learner, subset, EIF, and one-step estimators.}
    \label{fig:robustness for tau11}
\end{figure}

For the second study, we consider a scenario in which all nuisance functions are generated from nonparametric models, but the GAM working models are correctly specified (i.e., the true nuisance functions lie in the additive smooth-function class used for estimation). Figure \ref{fig:gam tau10} and \ref{fig:gam tau11} shows the log-RMSE for estimating $\tau^{10}(x)$ and $\tau^{11}(x)$ under non-parametric regimes estimated by GAM model.

\begin{table}[]
\centering
\label{tab:dgp_gam_nuisance}
\begin{tabular}{p{3.2cm} p{10.8cm}}
\hline
Quantity & Definition / components \\
\hline
$b(X)$
& $\sin(2\pi X_1)+0.8\log(1+3X_2)+1.5\,(X_3-0.3)_+ + (X_4-0.5)^2$ \\

$\text{logit}\{\pi(X)\}$
& $\sin(2\pi X_1)+(X_2-0.5)$ \\

$\text{logit}\{p_1(X)\}$
& $0.8+0.3\log(1+X_1)+0.3(X_2-0.5)-0.25(X_3-0.5)^2$ \\

$\text{logit}\{p_0(X)\}$
& $-0.8-0.3\log(1+X_1)-0.3(X_2-0.5)+0.25(X_3-0.5)^2$ \\
\hline
\end{tabular}
\caption{The data generating process for Study II.}
\end{table}

\begin{figure}
    \centering
    \includegraphics[width=0.8\linewidth]{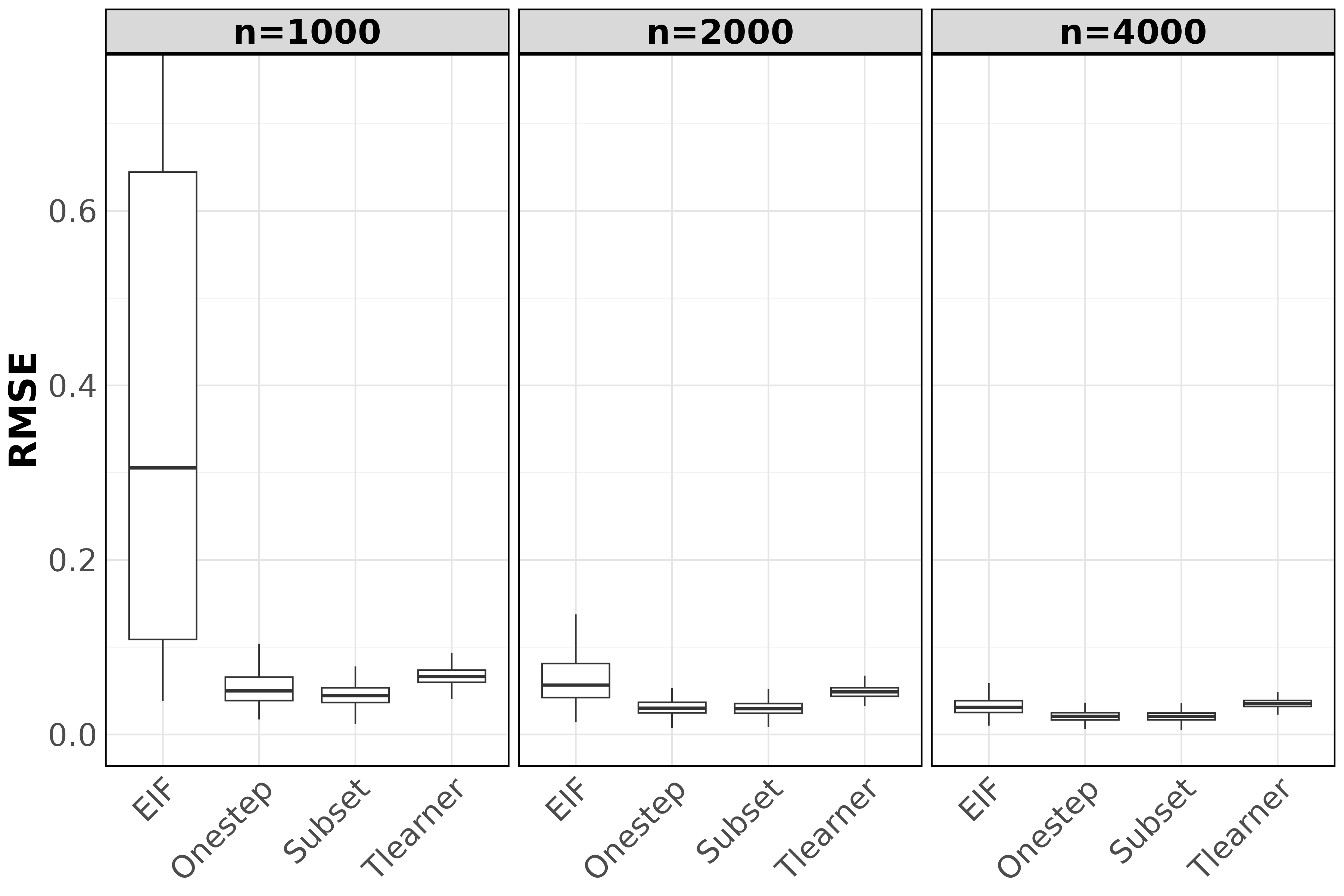}
    \caption{Boxplots display the distribution of RMSE for four estimators (EIF, One-step, Subset, and Tlearner) estimated by GAM model across three sample sizes ($n=1000$, $n=2000$, $n=4000$)}
    \label{fig:gam tau10}
\end{figure}
\begin{figure}
    \centering
    \includegraphics[width=0.8\linewidth]{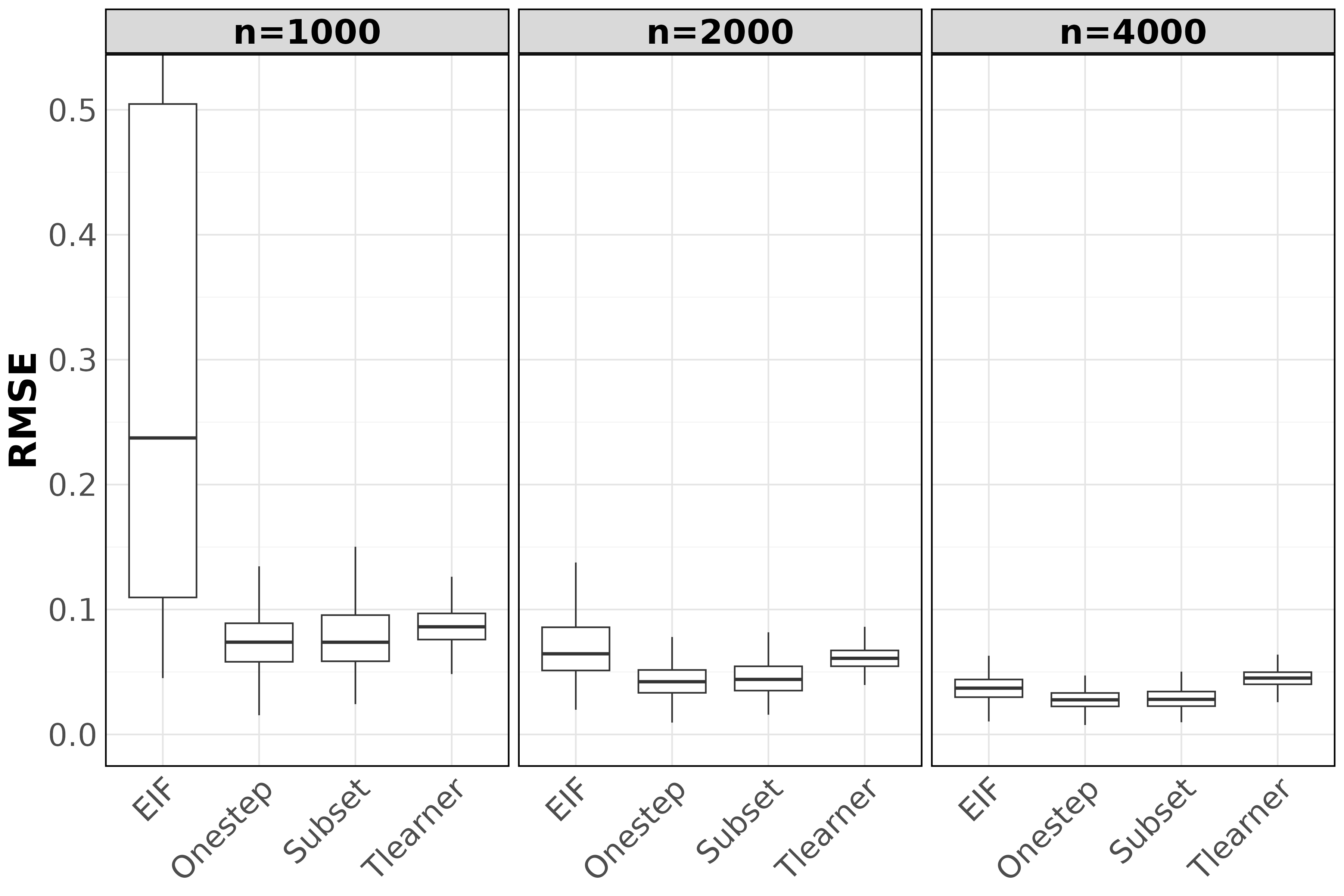}
    \caption{Boxplots display the distribution of RMSE for four estimators (EIF, One-step, Subset, and Tlearner) estimated by GAM model across three sample sizes ($n=1000$, $n=2000$, $n=4000$)}
    \label{fig:gam tau11}
\end{figure}

Since the target estimand $\tau^u(x)$ is linear in the baseline setup (e.g., $\tau^u(x)=0.5X_1$), we also consider a  nonparametric specification in the follow-up second simulation. In particular, we set $\tau_{00}(x)=0.5\sin(2\pi X_1)+0.3(X_2-0.5)^2$, $\tau_{10}(x)=0.4(X_1-0.5)^2+0.6\cos(2\pi X_2)$, $\tau_{11}(x)=-0.5\sin(2\pi X_1)+0.5(X_3-0.5)^2$. Figures~\ref{fig:gam_tau10_non}, \ref{fig:gam_tau00_non}, and \ref{fig:gam_tau11_non} show results similar to those in the linear target-estimand setting.

\begin{figure}[!ht]
    \centering
    \includegraphics[width=0.8\linewidth]{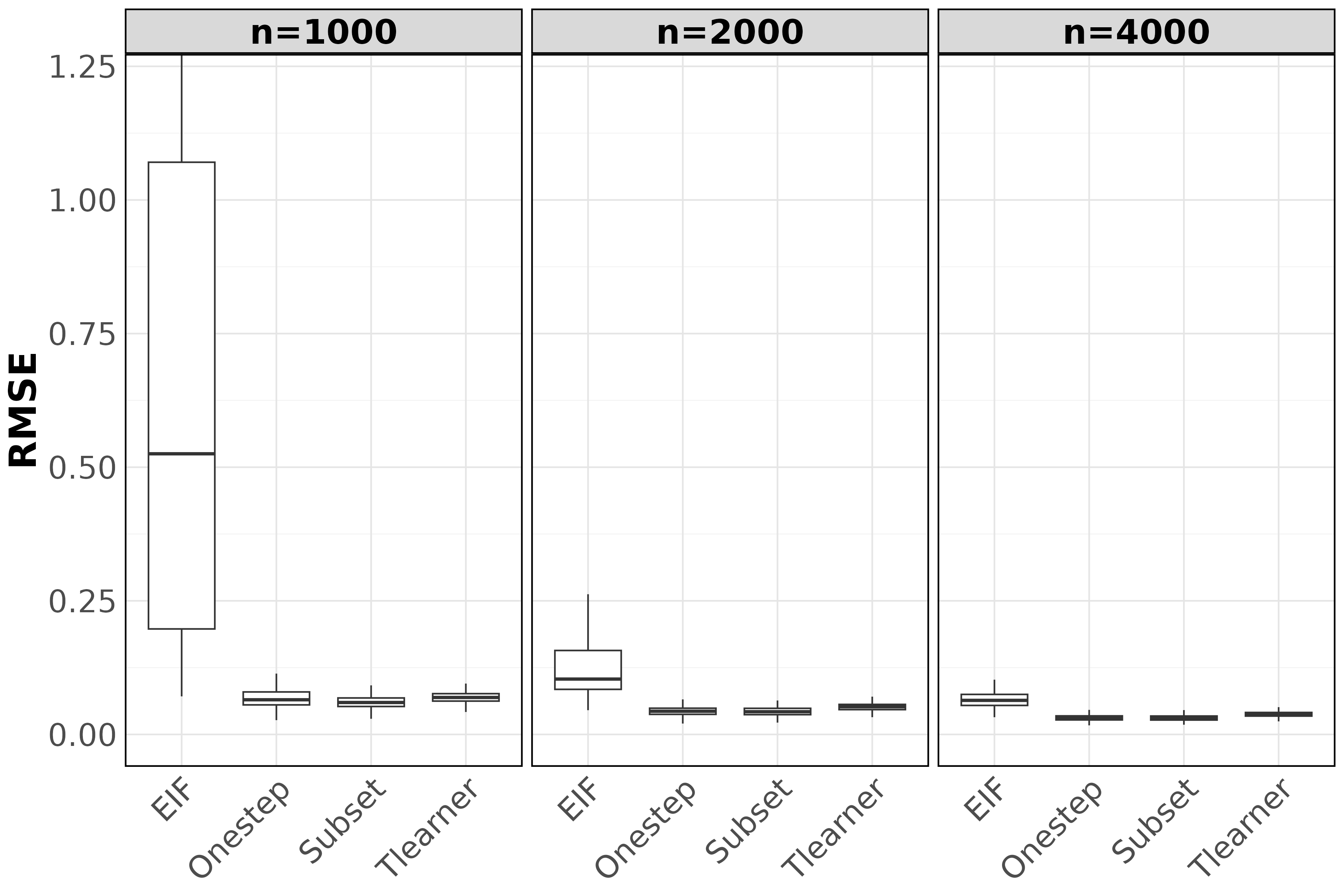}
    \caption{Boxplots show the distribution of RMSE for four estimators (EIF, one-step, subset, and T-learner), using GAM estimation, across three sample sizes when $\tau^u$ is also nonparametric ($n=1000$, $n=2000$, $n=4000$).}
    \label{fig:gam_tau10_non}
\end{figure}

\begin{figure}[!ht]
    \centering
    \includegraphics[width=0.8\linewidth]{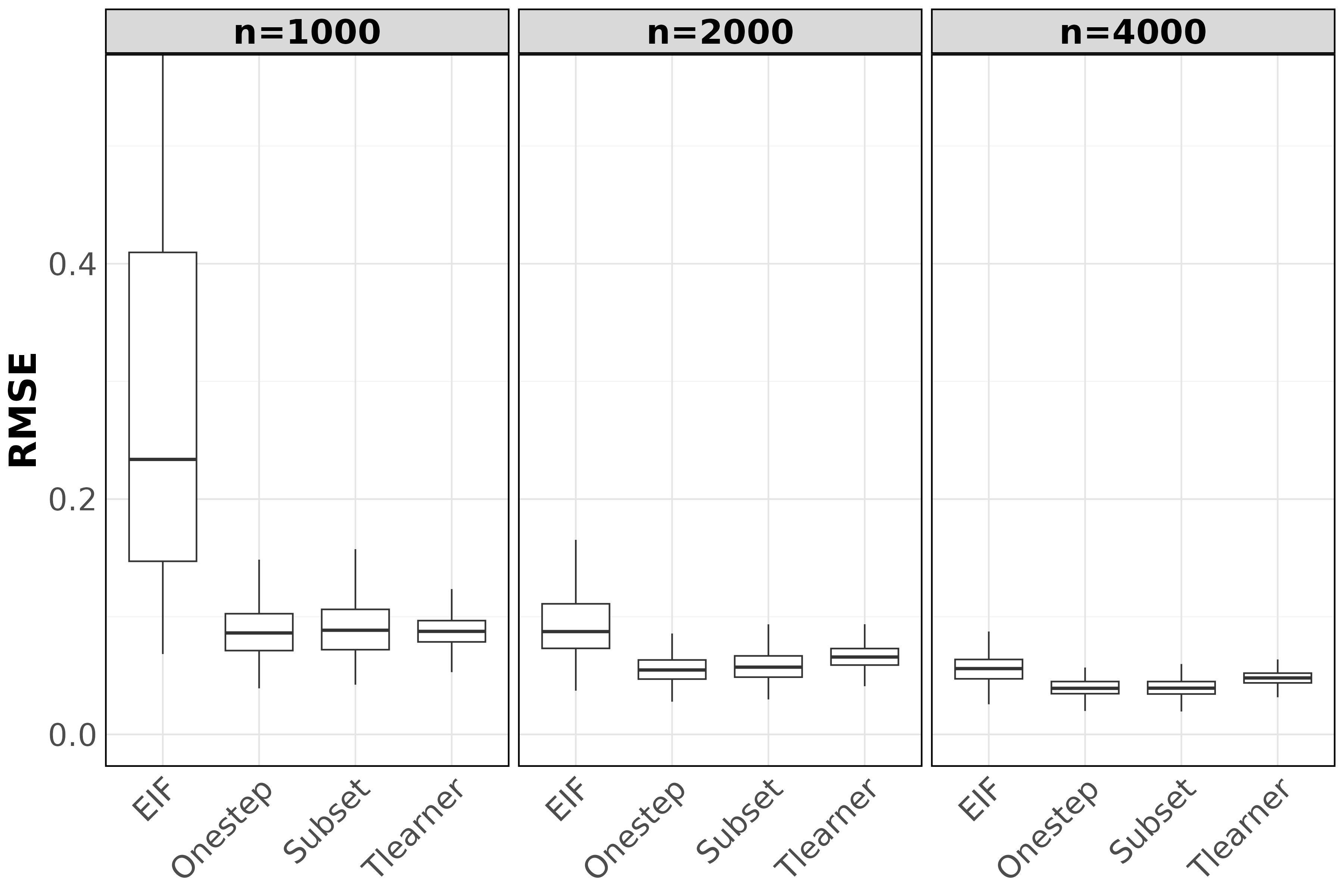}
    \caption{Boxplots show the distribution of RMSE for four estimators (EIF, one-step, subset, and T-learner), using GAM estimation, across three sample sizes when $\tau^u$ is also nonparametric ($n=1000$, $n=2000$, $n=4000$).}
    \label{fig:gam_tau00_non}
\end{figure}

\begin{figure}[!ht]
    \centering
    \includegraphics[width=0.8\linewidth]{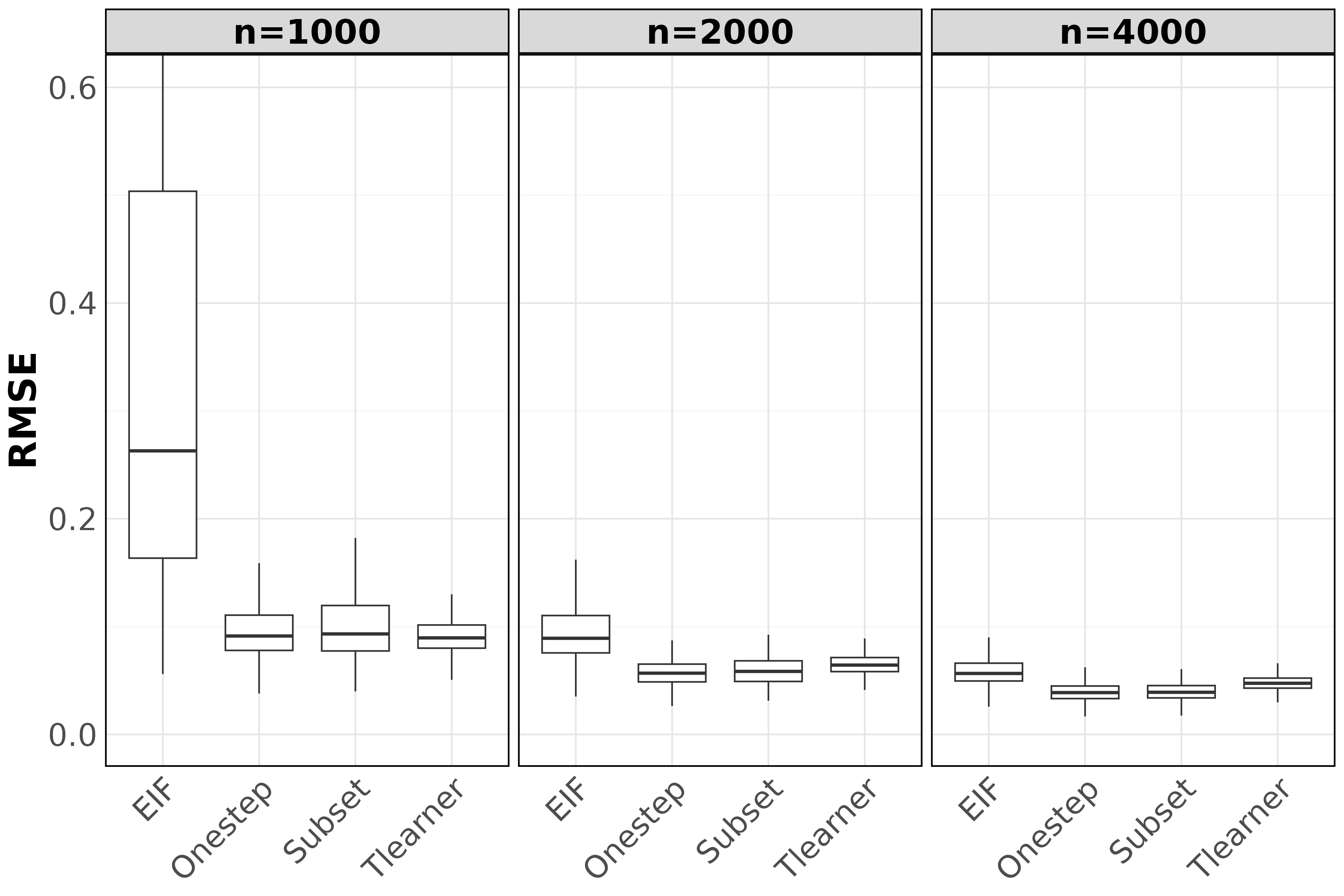}
    \caption{Boxplots show the distribution of RMSE for four estimators (EIF, one-step, subset, and T-learner), using GAM estimation, across three sample sizes when $\tau^u$ is also nonparametric ($n=1000$, $n=2000$, $n=4000$).}
    \label{fig:gam_tau11_non}
\end{figure}

For Study III, we examine the performance of the subset and one-step estimators under overlap violations and unbalanced subset sizes. The true target estimand remains linear, $\tau^u(x)=0.5X_1$, and the outcome mean model is also linear, $b(X)=X_1-0.4X_2+X_3+0.5X_4$. All nuisance functions and the second-stage regression are specified as in the previous GAM setting.

\begin{table}[!ht]
\centering

\begin{tabular}{p{3.2cm} p{10.8cm}}
\hline
Quantity & Definition / components \\
\hline
$\text{logit}\{\pi(X)\}$
& $3.2\,\sin\!\bigl(2\pi X_1X_2\bigr) + 2.6\,(X_3-0.5)^2$ \\

$ \text{logit} \{p_1^{\mathrm{temp}}(X)\}$
& $0.8 + 1.2\,(X_1-0.5)^2 + 0.9\,\sin\!\bigl(2\pi X_3X_4\bigr) - 0.6\,(\pi(X)-0.5)$ \\

$\text{logit}\{p_0^{\mathrm{temp}}(X)\}$
& $ -0.8 - 1.2\,(X_1-0.5)^2 - 0.9\,\sin\!\bigl(2\pi X_3X_4\bigr) - 1.5\,(\pi(X)-0.5)$ \\[2pt]

$\Delta(X)$
& $\Delta(X)=\bigl[p_1^{\mathrm{temp}}(X)-p_0^{\mathrm{temp}}(X)\bigr]$ \\

$m(X)$
& $m(X)=\{p_1^{\mathrm{temp}}(X)+p_0^{\mathrm{temp}}(X)\}/2$ \\[2pt]

$p_1(X)$, $p_0(X)$
& $p_1(X)=\bigl[m(X)+\Delta(X)/2\bigr]$,\quad
  $p_0(X)=\bigl[m(X)-\Delta(X)/2\bigr]$ \\
\hline
\multicolumn{2}{l}{\footnotesize Note: $\pi(x)$, $p_1(x)$, $p_0(x)$ truncated to $[0.01,0.99]$.}\\
\hline
\end{tabular}
\caption{Study III: score-model construction with overlap violations and constrained difference between $p_1(X)$ and $p_0(X)$.}
\label{tab:score_study3}
\end{table}

Table~\ref{tab:score_study3} reports the score-model specifications for $\pi(x)$, $p_1(x)$, and $p_0(x)$. Under this scenario, the observed subset with $S=1$ is highly unbalanced across treatment arms: there are 69 units with $(Z=0,S=1)$ versus 535 units with $(Z=1,S=1)$ (i.e., $69{:}535$). Consequently, we expect the subset estimator for $\tau^{11}(x)$ to perform worse than the one-step estimator. Figure \ref{fig:gam_tau11_s3} proves our thoughts.

\begin{figure}[!ht]
    \centering   \includegraphics[width=0.8\linewidth]{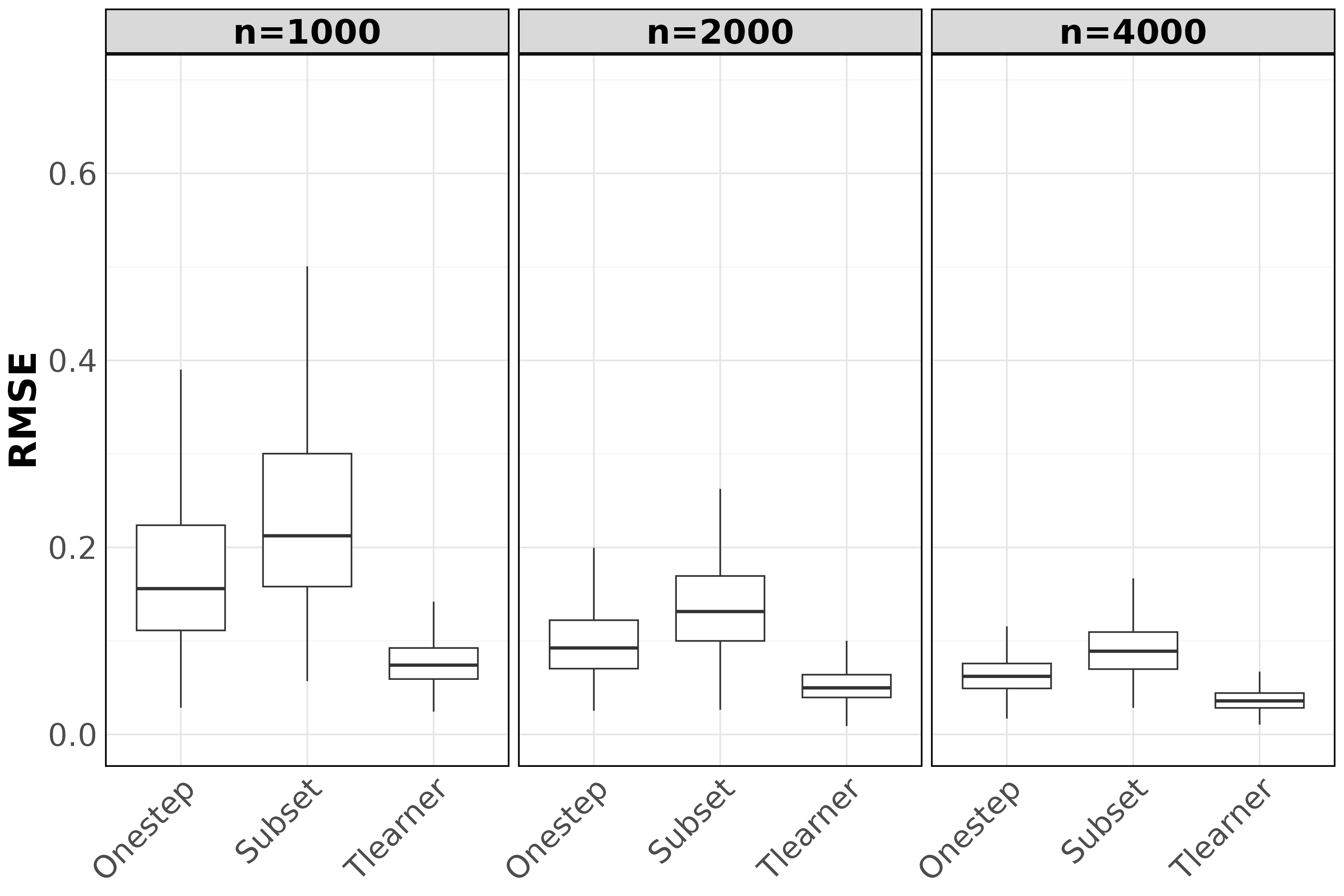}
    \caption{Boxplots show the distribution of RMSE for three estimators ( one-step, subset, and T-learner), using GAM estimation, across three sample sizes when $S=1$ is unbalanced($n=1000$, $n=2000$, $n=4000$).}
    \label{fig:gam_tau11_s3}
\end{figure}

In our simulation studies, we generally do not recommend the EIF estimator, as it performs poorly in small samples and requires a more complex three-stage estimation procedure. 

Between the subset and one-step estimators, our results suggest that when the observed cells are highly unbalanced, the one-step estimator is preferable. However, we also find that the one-step estimator can exhibit substantial variance in small samples, especially for $\tau^{10}(x)$. We suspect this is driven by instability in the denominator $\tilde p_1(x)-\tilde p_0(x)$, which is a difference of two estimated quantities and can therefore be sensitive to sampling variability (particularly when $\tilde p_1(x)$ and $\tilde p_0(x)$ are close).

\subsection{Hotspotting}\label{app:hot}
Desciption about dataset:

We used the publicly available dataset, which differs somewhat from the data used in the JAMA engagement re-analysis, as some variables are unavailable in the public release. The dataset contains 774 patients (after removing all units with any missing data from the original 782), each characterized by 20 baseline covariates covering demographic, socioeconomic, and clinical characteristics.  Engagement is defined using the timing of the first home visit or PCP/specialist appointment, the counts of home visits, phone calls, and other encounters, as well as the duration of the intervention. In our sample, 386 were controls, 137 treated non-engagers, and 251 treated engagers. The covariates includes demographic, socioeconomic, health-status, and clinical utilization covariates. Specifically, these cover sex and race/ethnicity indicators (male, Black, Hispanic, White), employment status, age-group categories, education levels, marital status, housing situation, social support, and self-rated health. The dataset also includes several baseline clinical conditions, such as arthritis admission, mental health condition, substance abuse, HIV/AIDS, COPD/emphysema, congestive heart failure, and diabetes, as well as prior healthcare utilization measures, including the number of inpatient admissions in the 180 days before the index admission and the number of days hospitalized at the index admission. The public dataset includes readmission outcomes measured at 30, 90, and 180 days, and in our analysis we focus on 30-day readmission.

Figure~\ref{fig:CATE} shows the DR-learner estimates of the CATE, using Super Learner to estimate the outcome regression and a generalized random forest (grf) for the second-stage regression. From the results below, there is a clear shift from negative to positive estimates rather than a constant pattern, indicating treatment effect heterogeneity. Figure \ref{fig: one} displays the estimated complier CPCEs and corresponding 95\% confidence
intervals from the one-step estimator.
\begin{figure}
    \centering
    \includegraphics[width=\linewidth]{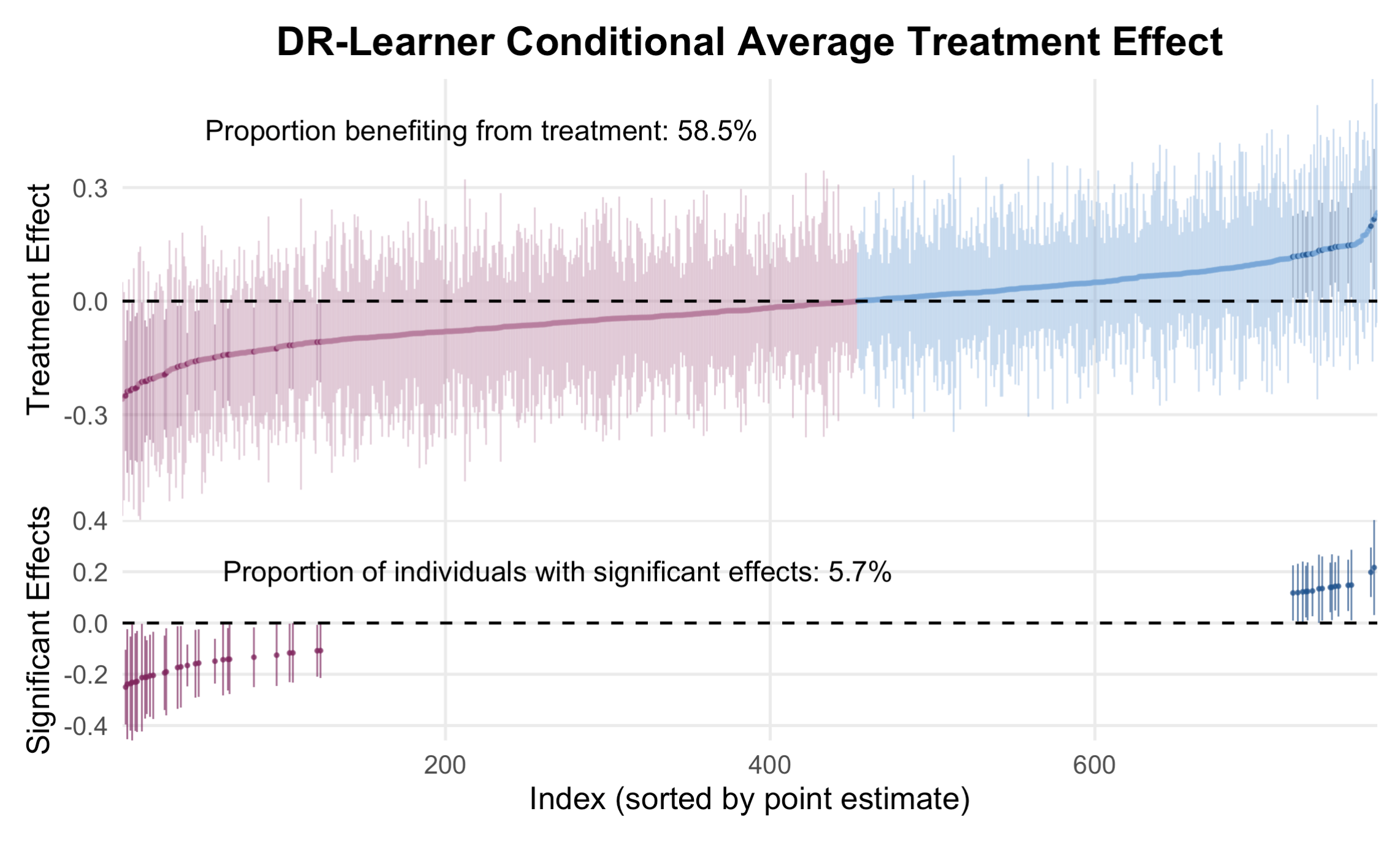}
    \caption{Estimated CATE by DR-learner, ordered by the point estimates, with 95\% confidence intervals computed by \texttt{grf}.}
    \label{fig:CATE}
\end{figure}

\begin{figure}
    \centering
    \includegraphics[width=\linewidth]{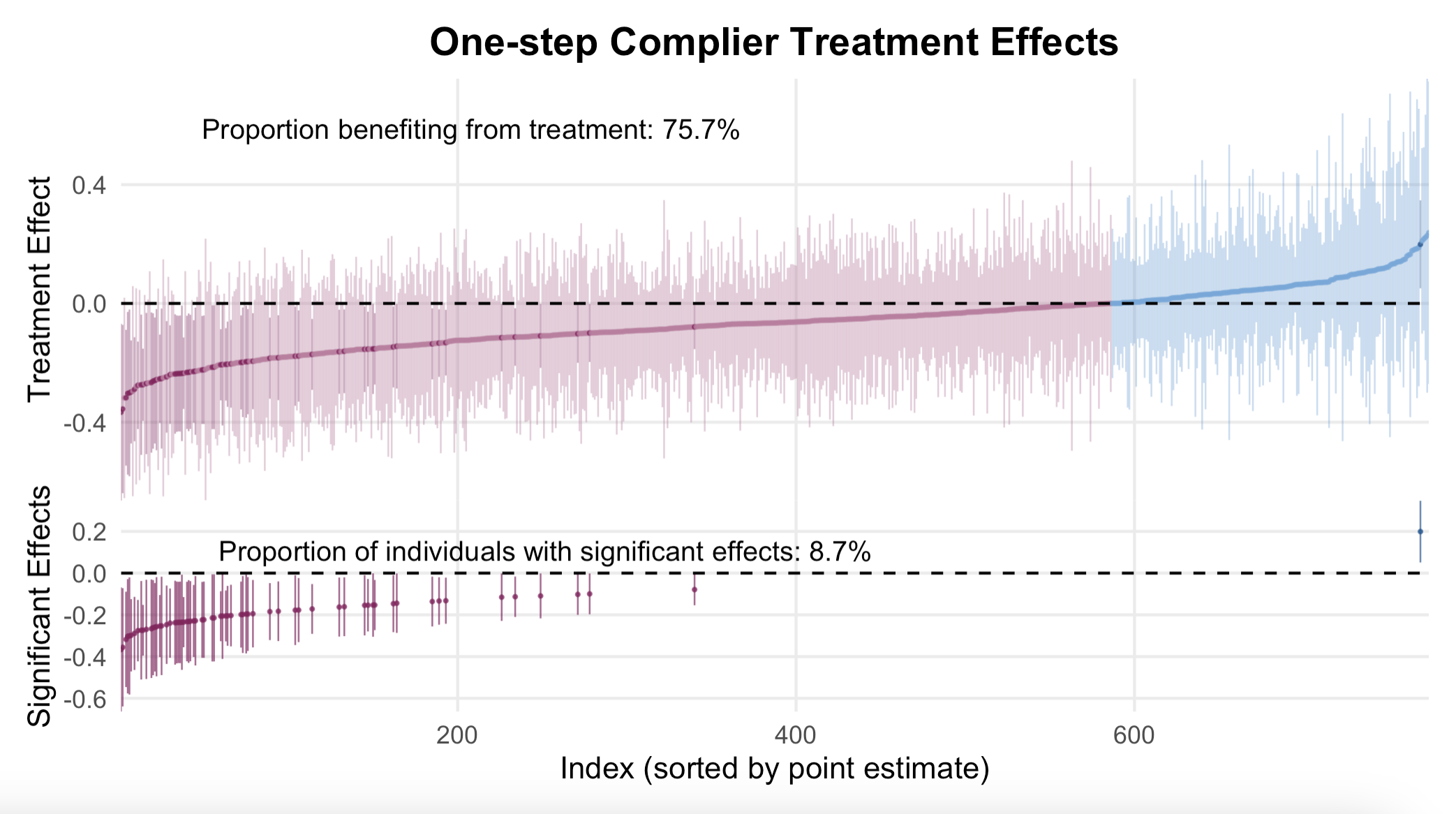}
    \caption{Estimated complier-specific treatment effects by one-step, ordered by the point estimates, with 95\% confidence intervals computed by \texttt{grf}.}
    \label{fig: one}
\end{figure}

For the CPCE among compliers, the following  Figures~\ref{fig:importance} display variable importance from the second-stage regression models used by the subset and one-step estimators, indicating which variables drive heterogeneity in the complier CPCE. The variable importance rankings for the one-step and subset estimators are quite similar. Moreover, the top three variables account for a large proportion of the total importance, which are prior inpatient admissions in the past 180 days, the duration of the initial hospital stay, and sex.

\begin{figure}[htbp]
    \centering
    \begin{subfigure}{\linewidth}
        \centering
        \includegraphics[width=0.8\linewidth]{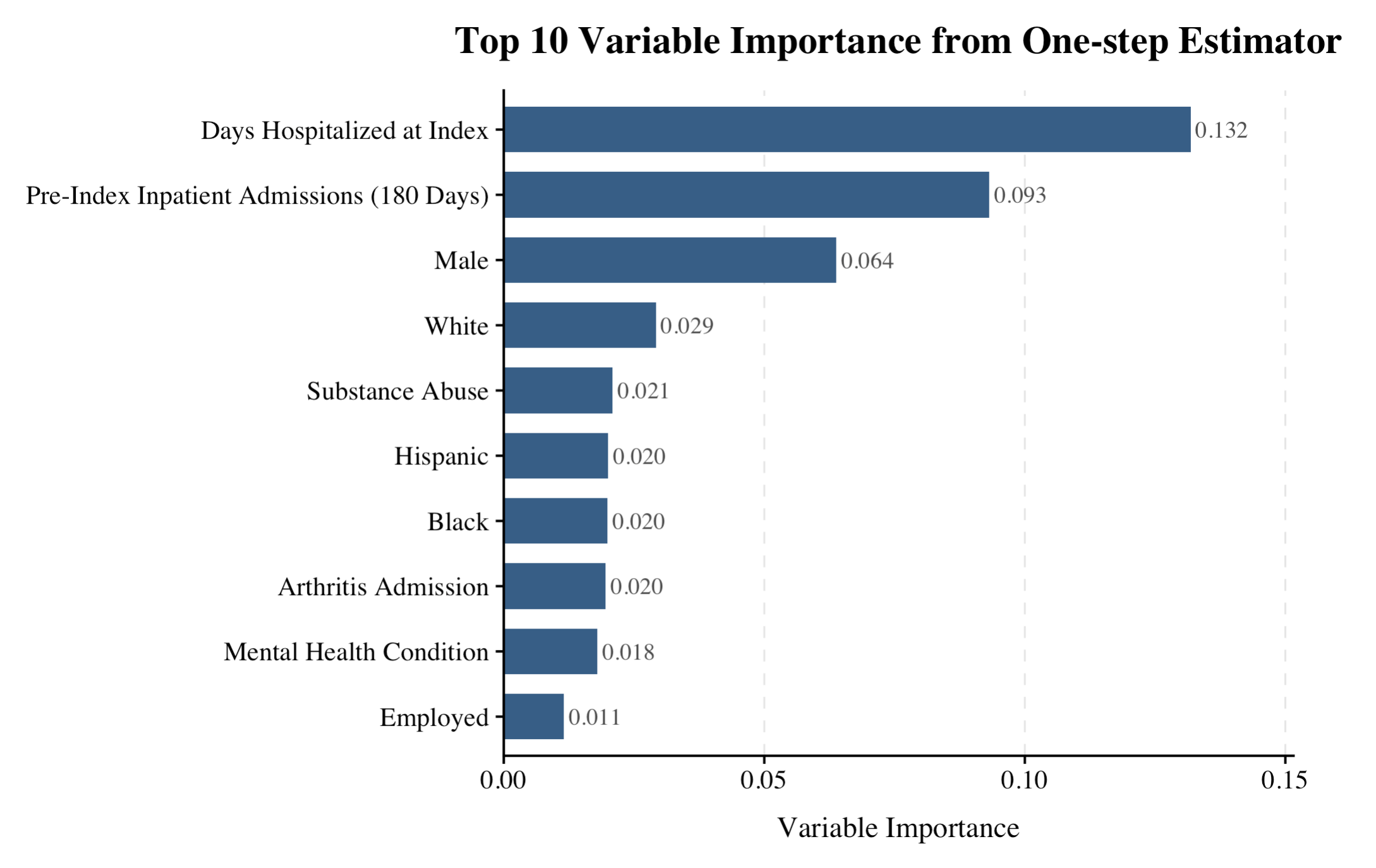}
    \end{subfigure}
    
    \vspace{0.5em}
    
    \begin{subfigure}{\linewidth}
        \centering
        \includegraphics[width=0.8\linewidth]{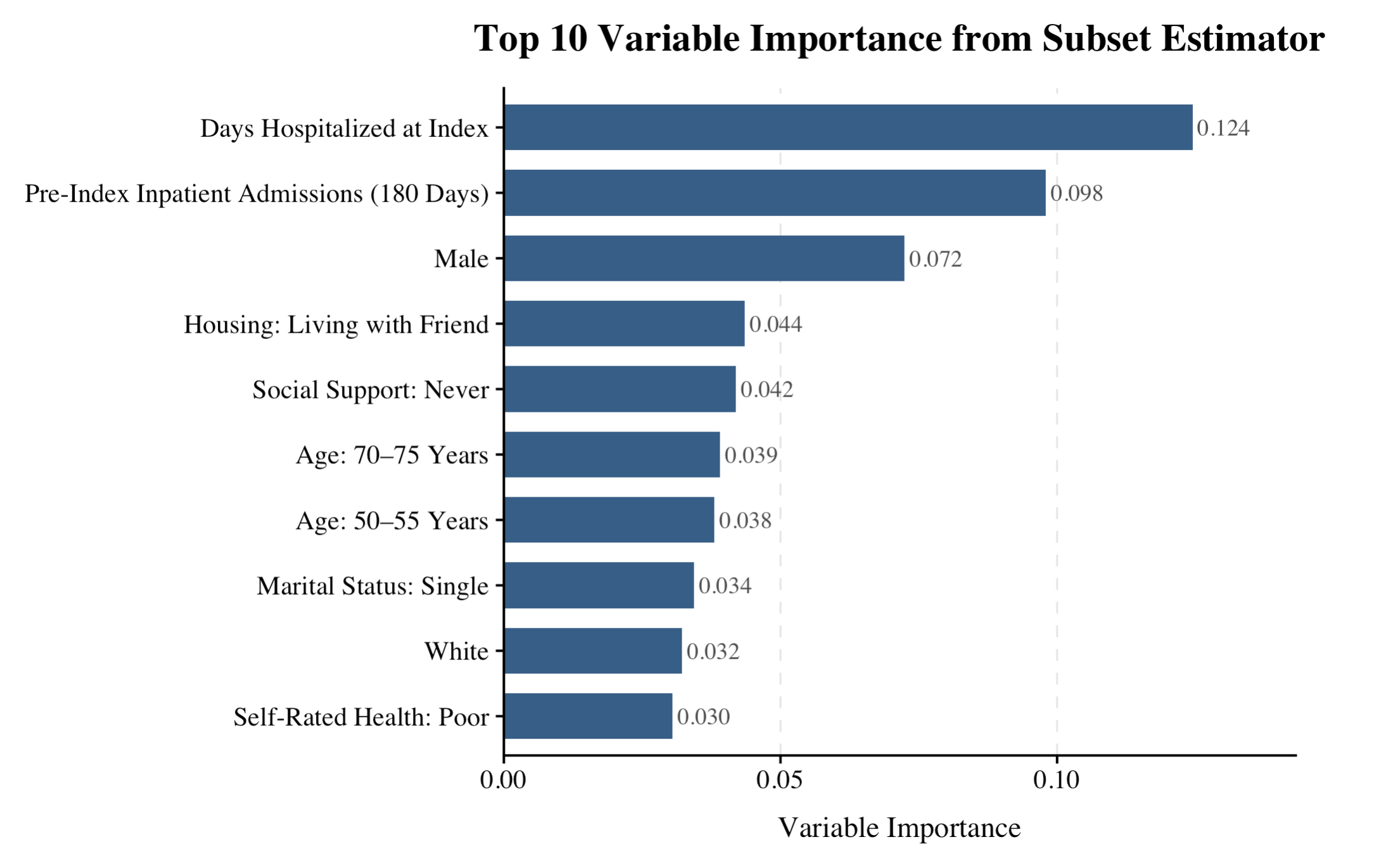}
    \end{subfigure}
    
\caption{Variable importance from the second-stage regression used to estimate $\tau^u(x)$ for the one-step (top) and subset (bottom) estimators.}
    \label{fig:importance}
\end{figure}

We examine how specific variables relate to the estimated effects using boxplots for the binary covariate (sex) and partial dependence plots for the continuous covariates (prior inpatient admissions in the past 180 days and length of the index hospitalization). Boxplot \ref{fig:box} shows that women tend to have a better treatment than the men for both one-step and subset estimator. For the partial dependence plots, we truncate the x-axis at the 99th percentile to avoid extrapolating into regions with sparse data, where the curves tend to flatten. Figure \ref{fig:admission} shows that the estimated complier CPCE becomes more negative with increasing prior inpatient admissions, suggesting larger benefits among patients with greater recent utilization.
Figure \ref{fig:hosp} indicates that the average partial dependence curve (red) is nonlinear but stays mostly below 0, suggesting the estimated complier effect is generally favorable. And the treatment become less effective for longer initial hospital days (probably sicker patients) and more effective when initial hospital days over 10 days (probably serious sicker patients). Moreover, \cite{yang2023hospital} found that education is an important variable for engagement. However, Figure~\ref{fig:education} provides little visual evidence of a statistically significant difference in $\hat\tau^{10}(x)$ across education levels.

\begin{figure}
    \centering
    \includegraphics[width=0.8\linewidth]{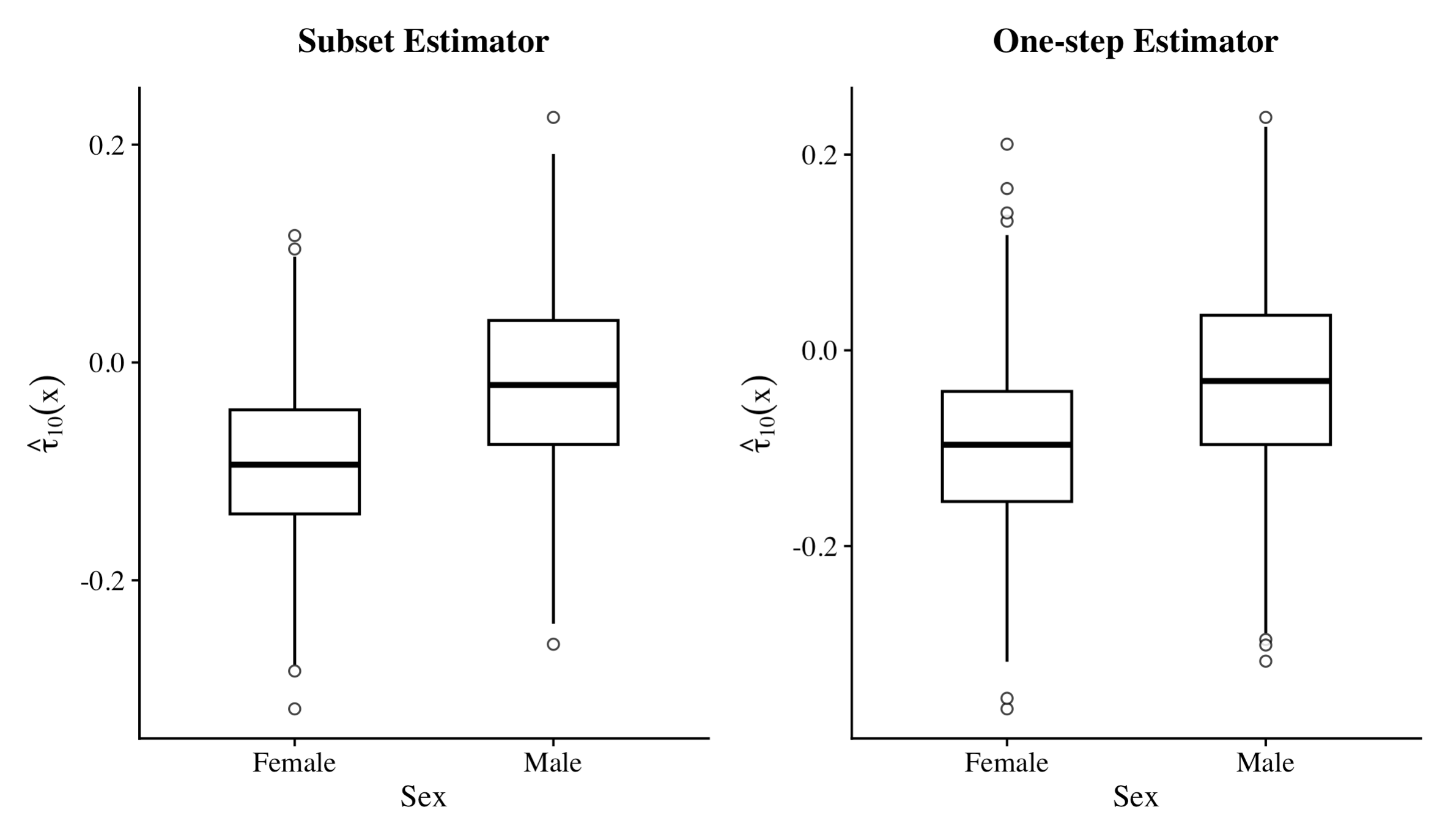}
\caption{Distribution of estimated CPCEs $\hat{\tau}_{10}(X)$ stratified by sex level. Boxplots summarize the estimates within each education category for the subset estimator (left) and the one-step estimator (right).}
    \label{fig:box}
\end{figure}

\begin{figure}
    \centering
    \includegraphics[width=\linewidth]{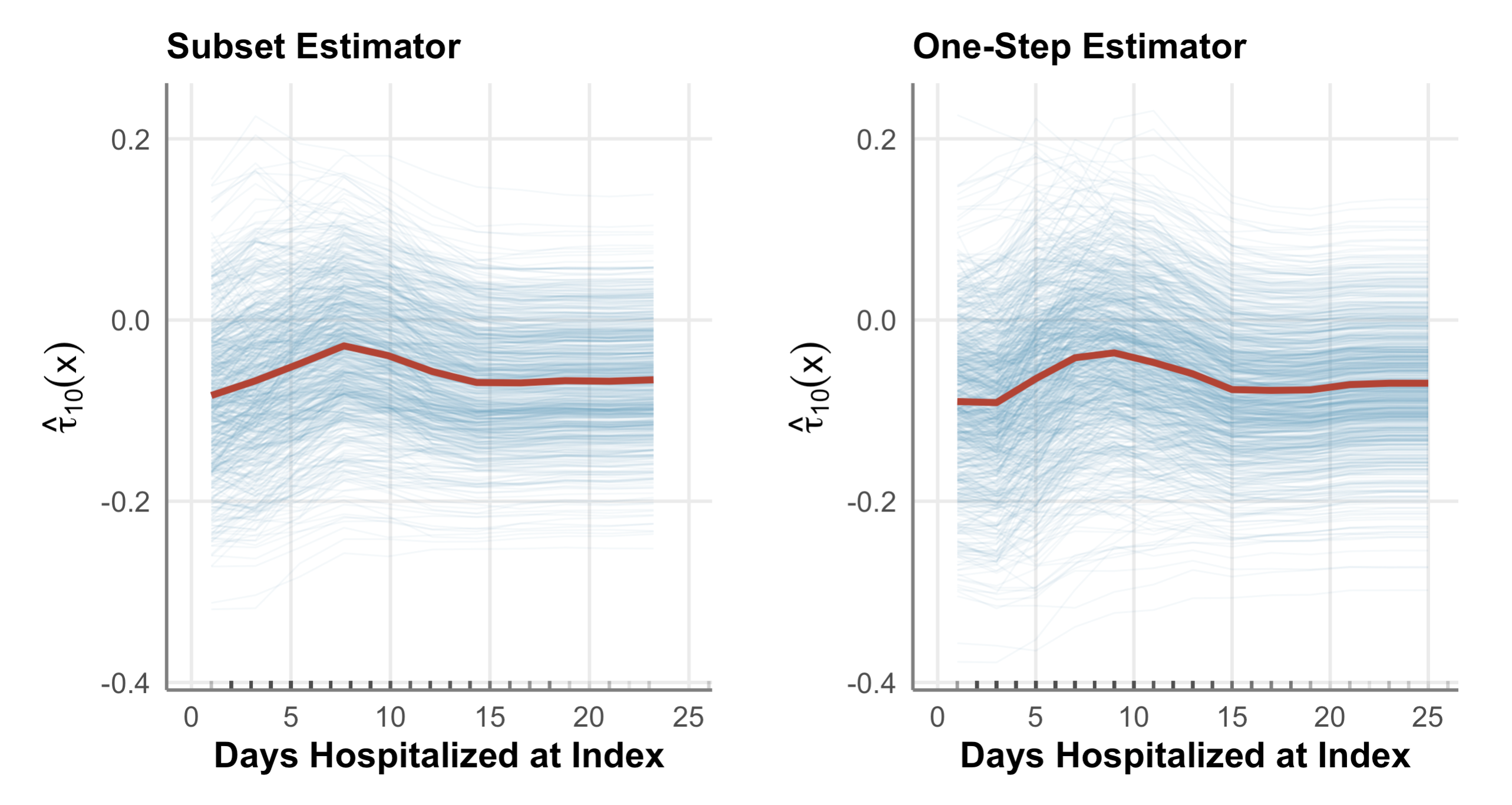}
\caption{Partial dependence of the estimated CPCE $\hat{\tau}_{10}(X)$ on the duration of the initial hospital stay. Thin curves show individual conditional expectation (ICE) trajectories, and the thick curve shows the average partial dependence function, for the subset estimator (left) and the one-step estimator (right). The rug indicates the empirical distribution of the covariate.}
    \label{fig:hosp}
\end{figure}

\begin{figure}
    \centering
    \includegraphics[width=\linewidth]{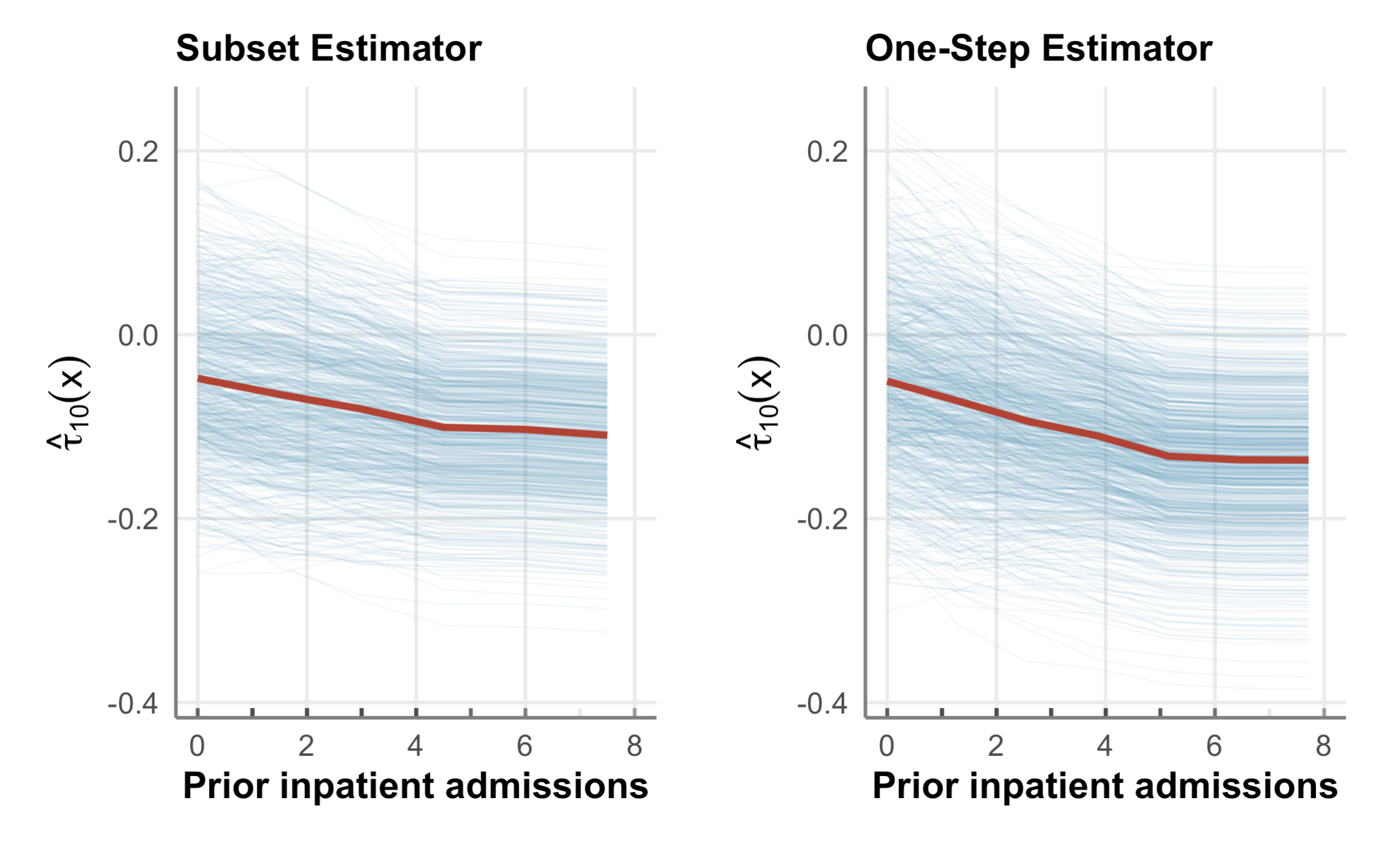}
\caption{Partial dependence of the estimated CPCE $\hat{\tau}_{10}(X)$ on the number of prior inpatient admissions in the past 180 days. Thin curves show individual conditional expectation (ICE) trajectories, and the thick curve shows the average partial dependence function, for the subset estimator (left) and the one-step estimator (right). The rug indicates the empirical distribution of the covariate.}

    \label{fig:admission}
\end{figure}

\begin{figure}
    \centering
    \includegraphics[width=\linewidth]{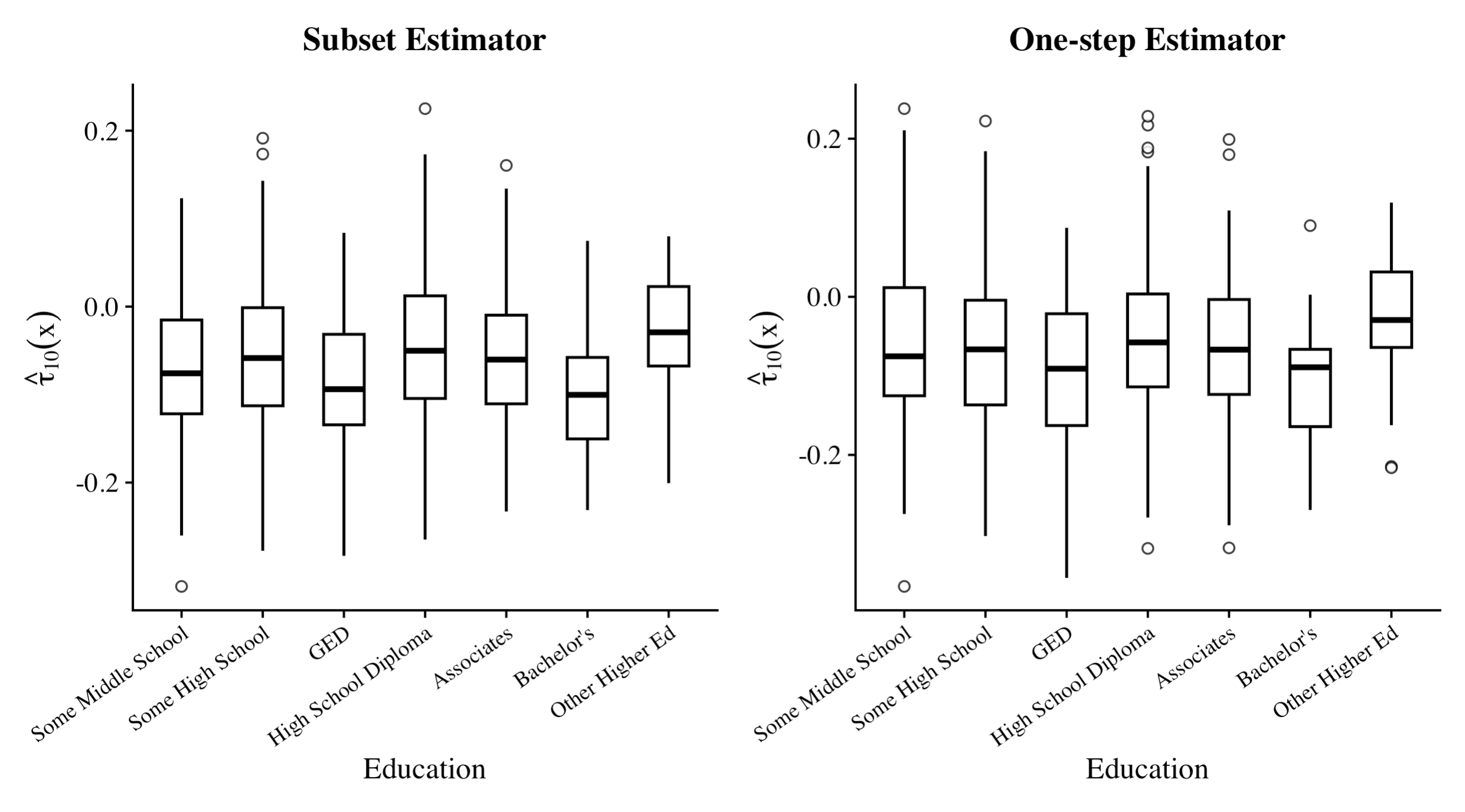}
\caption{Distribution of estimated CPCEs $\hat{\tau}_{10}(X)$ stratified by education level. Boxplots summarize the estimates within each education category for the subset estimator (left) and the one-step estimator (right).}
    \label{fig:education}
\end{figure}

\end{document}